\documentclass[a4paper,USenglish,cleveref, autoref, numberwithinsect,thm-restate]{lipics-v2021}
\pdfoutput=1 %uncomment to ensure pdflatex processing (mandatatory e.g. to submit to arXiv)

\usepackage{xcolor}
\usepackage{amsmath}
\usepackage{todonotes}
\usepackage{xspace}
\usepackage{pdfpages}

\newcommand{\defparprob}[4]{
  \vspace{1mm}
\noindent\fbox{
  \begin{minipage}{0.96\textwidth}
  \begin{tabular*}{\textwidth}{@{\extracolsep{\fill}}lr} #1  & {\bf{Parameter:}} #3 \\ \end{tabular*}
  {\bf{Input:}} #2  \\
  {\bf{Question:}} #4
  \end{minipage}
  }
  \vspace{1mm}
}

\ifdefined\DEBUG{}
\def\rem#1{{\marginpar{\raggedright\scriptsize #1}}}
\newcommand{\disc}[1]{\rem{\textcolor{cyan}{\(\bullet \) #1}}}
\newcommand{\rearng}[1]{{\color{olive}{#1}}}

\newcommand{\bmpr}[1]{\rem{\textcolor{red}{\(\bullet \) #1}}}

\newcommand{\sk}[1]{\rem{\textcolor{blue}{\(\bullet \) #1}}}
\else

\newcommand{\bmpr}[1]{}
\newcommand{\sk}[1]{}

\newcommand{\disc}[1]{}
\newcommand{\rearng}[1]{{#1}}

\fi
\newcommand{\depr}[1]{}

\newcommand{\nphard}{\textsc{NP-}hard\xspace}
\newcommand{\gtk}{$(G, T, k)$\xspace}
\newcommand{\gtminustk}{$(G, T \setminus \{t\}, k)$\xspace}
\newcommand{\tminust}{T \setminus \{t\}\xspace}
\newcommand{\ivd}{internally vertex-disjoint\xspace}
\newcommand{\cct}{connected component\xspace}
\newcommand{\ccts}{connected components\xspace}

\newcommand{\calF}{\mathcal{F}\xspace}
\newcommand{\calT}{\mathcal{T}\xspace}
\newcommand{\calC}{\mathcal{C}\xspace}
\newcommand{\calP}{\mathcal{P}\xspace}
\newcommand{\calB}{\mathcal{B}\xspace}
\newcommand{\calL}{\mathcal{L}\xspace}

\newcommand{\optx}{\emph{OPT}$_x$}
\newcommand{\puv}{P_{uv}\xspace}
\newcommand{\puvprime}{P'_{u'v'}\xspace}

\newcommand{\Oh}{\mathcal{O}}

\newcommand{\ttsep}{$(t, T \setminus \{t\})$-separator\xspace}

\newcommand{\mwns}{\textsc{Multiway Near-Separator}\xspace}
\newcommand{\mwnsshort}{\textsc{mwns}\xspace}

\newcommand{\mws}{\textsc{Multiway Separator}\xspace}
\newcommand{\mwsshort}{\textsc{mws}\xspace}

\usepackage{graphicx}

\newtheorem{reductionrule}{Reduction Rule}
\newtheorem{markingscheme}{Marking Scheme}

\title{On the Parameterized Complexity of \mwns}

\author{Bart M. P. Jansen}{Eindhoven University of Technology, The Netherlands \and \url{https://www.win.tue.nl/~bjansen/}}{b.m.p.jansen@tue.nl}{https://orcid.org/0000-0001-8204-1268}{}

\author{Shivesh K. Roy}{Eindhoven University of Technology, The Netherlands \and \url{https://sites.google.com/view/shiveshroy}}{s.k.roy@tue.nl}{https://orcid.org/0000-0003-0896-3437}{}

\authorrunning{B.\,M.\,P. Jansen and S.\,K. Roy}

\Copyright{Bart M. P. Jansen, and Shivesh K. Roy}%TODO mandatory, please use full first names. LIPIcs license is "CC-BY";  http://creativecommons.org/licenses/by/3.0/

\keywords{FPT algorithm}

\category{}%optional, e.g. invited paper

\relatedversion{}%optional, e.g. full version hosted on arXiv, HAL, or other respository/website
\supplement{}%optional, e.g. related research data, source code, ... hosted on a repository like zenodo, figshare, GitHub, ...

\funding{\flag[3cm]{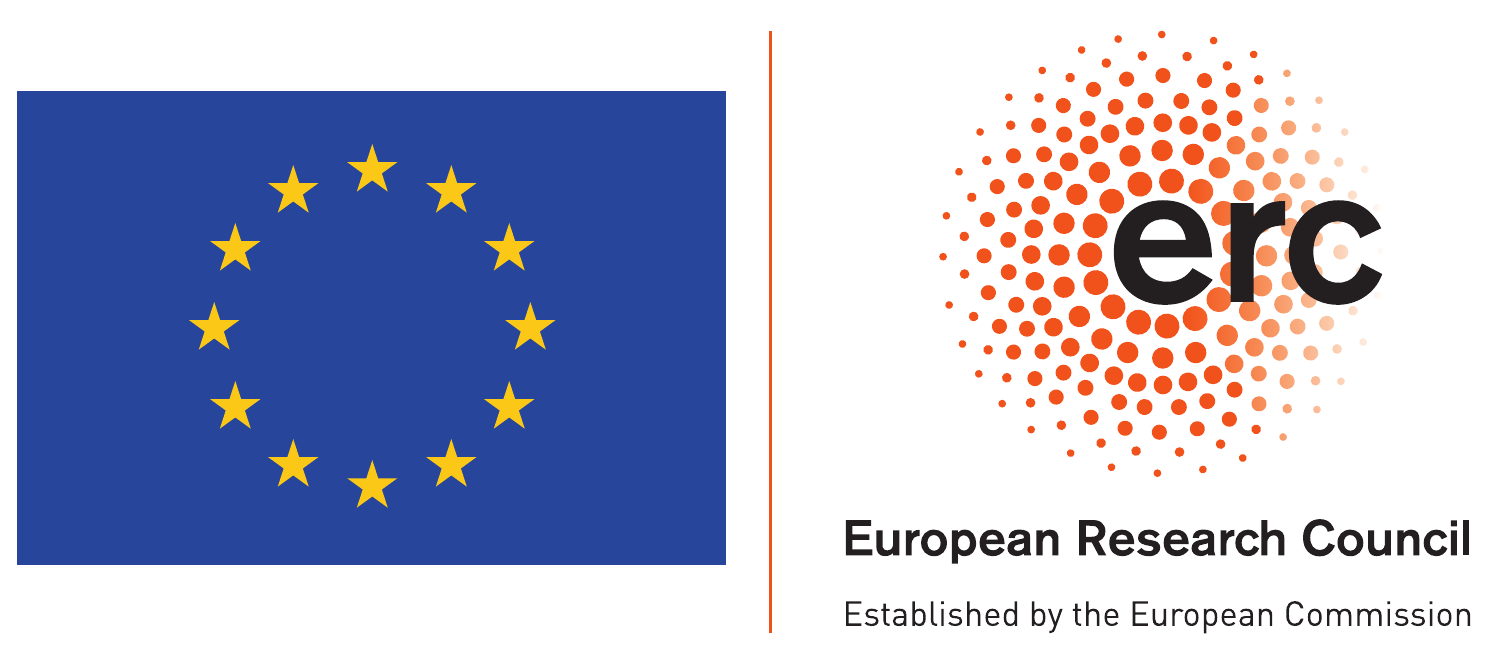}This project has received funding from the European Research Council (ERC) under the European Union's Horizon 2020 research and innovation programme (grant agreement No 803421, ReduceSearch).} 

\nolinenumbers %uncomment to disable line numbering

\hideLIPIcs  %uncomment to remove references to LIPIcs series (logo, DOI, ...), e.g. when preparing a pre-final version to be uploaded to arXiv or another public repository

\EventEditors{Neeldhara Misra and Magnus Wahlstr\"{o}m}
\EventNoEds{2}
\EventLongTitle{18th International Symposium on Parameterized and Exact Computation (IPEC 2023)}
\EventShortTitle{IPEC 2023}
\EventAcronym{IPEC}
\EventYear{2023}
\EventDate{September 6--8, 2023}
\EventLocation{Amsterdam, The Netherlands}
\EventLogo{}
\SeriesVolume{285}
\ArticleNo{15}
\ccsdesc[500]{Mathematics of computing~Graph algorithms}
\ccsdesc[500]{Theory of computation~Parameterized complexity and exact algorithms}
\ccsdesc[300]{Theory of computation~Approximation algorithms analysis}

\keywords{fixed-parameter tractability, multiway cut, near-separator}

\begin{document}

\maketitle

\begin{abstract}
We study a new graph separation problem called \textsc{Multiway Near-Separator}. Given an undirected graph~$G$, integer~$k$, and terminal set~$T \subseteq V(G)$, it asks whether there is a vertex set~$S \subseteq V(G) \setminus T$ of size at most~$k$ such that in graph~$G-S$, no pair of distinct terminals can be connected by two pairwise \ivd paths. Hence each terminal pair can be separated in~$G-S$ by removing at most one vertex. The problem is therefore a generalization of \textsc{(Node) Multiway Cut}, which asks for a vertex set for which each terminal is in a different component of~$G-S$. We develop a fixed-parameter tractable algorithm for \textsc{Multiway Near-Separator} running in time~$2^{\Oh(k \log k)} \cdot n^{\Oh(1)}$. Our algorithm is based on a new pushing lemma for solutions with respect to important separators, along with two problem-specific ingredients. The first is a polynomial-time subroutine to reduce the number of terminals in the instance to a polynomial in the solution size~$k$ plus the size of a given suboptimal solution. The second is a polynomial-time algorithm that, given a graph~$G$ and terminal set~$T \subseteq V(G)$ along with a single vertex~$x \in V(G)$ that forms a multiway near-separator, computes a 14-approximation for the problem of finding a multiway near-separator not containing~$x$.
\end{abstract}

\clearpage
\section{Introduction}
\label{sec:intro}
Graph separation problems play an important role in the study of graph algorithms. While the problem of finding a minimum vertex set whose removal separates two terminals~$s$ and~$t$ can be solved in polynomial-time via the Ford-Fulkerson algorithm~\cite{FordF56}, many variations of the problem are NP-complete. They form a fruitful subject of investigation in the study of parameterized algorithmics, where a typical goal is to develop an algorithm that finds a suitable separator of size~$k$ in an $n$-vertex input graph in time~$f(k) \cdot n^{\Oh(1)}$, or concludes that no such solution exists. Landmark results in this area include the FPT algorithms for \textsc{Multiway Cut}~\cite{ChenLL09,CyganPPW13a,Guillemot11a,Marx06,Xiao10} (in which the goal is to find a vertex set which separates any pair of terminals from a given set~$T$) and \textsc{Multicut}~\cite{BousquetDT18,MarxR14} (in which only a specified subset of the terminal pairs must be separated) in undirected graphs.

After the parameterized complexity of the most fundamental separation problems in this area were settled, researchers started considering variations on the theme of graph separation, including \textsc{Steiner Multicut}~\cite{BringmannHML16} (given a sequence of subsets of terminals~$T_1, \ldots, T_\ell$, find a vertex set separating at least one pair $\{t_i \in T_i, t_j \in T_j\}$ for each~$i \neq j$), \textsc{Multiway Cut-Uncut}~\cite{ChitnisCHPP16,LokshtanovR0Z18} (given an equivalence relation~$\mathcal{R}$ on a set of terminal vertices, find a vertex set whose removal leaves terminals~$t_i, t_j$ in the same \cct if and only if~$t_i \equiv_{\mathcal{R}} t_j$), and \textsc{Stable Multiway cut}~\cite{MarxOR13} (find a multiway cut that is independent). 

In this paper we start to explore a new variation of the separation theme from the parameterized perspective. Rather than asking for a vertex set which fully separates some given terminal pairs, we are interested in \emph{nearly} separating terminals: in the remaining graph, there should not be \emph{two or more} \ivd paths connecting a terminal pair. We therefore study the following problem, which we believe is a natural extension in the well-studied area of graph separation problems.

\defparprob{\mwns(\mwnsshort)}{An undirected graph~$G$, terminal set~$T \subseteq V(G)$, and a positive integer~$k$.}{$k$}{Is there a set~$S \subseteq V(G) \setminus T$ with~$|S| \leq k$ such that there does not exist a pair of distinct terminals~$t_i, t_j \in T$ with \emph{two} \ivd~$t_i$-$t_j$~paths in~$G-S$?}

Note that, by Menger's theorem, the requirement on solutions~$S$ to \mwnsshort is equivalent to the requirement that in the graph~$G-S$, any terminal pair can be separated by removing a non-terminal vertex.
One could therefore imagine applications of this problem in the study of disrupting communications between nodes in a network. While the standard \textsc{Multiway Cut} problem captures the setting that all potential for communication between terminals has to be broken, the size of a solution to the near-separation problem can be arbitrarily much smaller while it still ensures the following property: the communication between each terminal pair is either broken by the solution, or there is at least one non-terminal vertex through which all communications of the pair must pass, so that it may be intercepted at that point. A related problem of reducing connectivity between nodes by removing edges or vertices was studied by Barman and Chawla~\cite{BarmanC10}, who presented various approximation algorithms.

Using the (perhaps non-standard) view that a direct edge between two vertices~$t_i, t_j$ means there are two internally vertex-disjoint paths between them (the intersection of the set of interior vertices is empty), the requirement that in~$G-S$ there do not exist two \ivd paths between any pair of distinct terminals, can alternatively be shown to be equivalent to demanding that~$T$ is an independent set and there is no simple cycle containing at least two vertices from~$T$ (see Proposition~\ref{Proposition:DifferentWaysOfLookingAtMWNS}). Simple cycles containing two vertices from~$T$, henceforth called $T$-cycles, therefore play an important role in our arguments. Based on this alternative characterization, the near-separator problem is related to the \textsc{Subset Feedback Vertex Set} problem, which asks for a minimum vertex set intersecting all cycles that contain at least one terminal~\cite{CyganPPW13}, although there the solution is allowed to contain terminal vertices.

A simple reduction (Lemma~\ref{Lemma:NP-hardness_proof_of_MWNS}) shows that \mwns is NP-complete. It forms a generalization of the \textsc{(Node) Multiway Cut} problem with undeletable terminals: an instance~$(G,T = \{t_1, \ldots, t_\ell\},k)$ of the latter problem can be reduced to an equivalent instance~$(G',T,k)$ of \mwnsshort by inserting~$|T|-1$ new non-terminal vertices~$w_1, \ldots, w_{|T|-1}$ with~$N(w_i) = \{t_i, t_{i+1}\}$.

The \mwns problem can be shown to be non-uniformly fixed-parameter tractable parameterized by~$k$ using the technique of recursive understanding~\cite{LokshtanovR0Z18}, as the problem can be formulated in Monadic Second-Order Logic and can be shown to become fixed-parameter tractable on~$(s(k), k+1)$-unbreakable graphs for some function~$s \colon \mathbb{N} \to \mathbb{N}$ by branching on small connected vertex sets with small neighborhoods. This only serves as a complexity classification, however, as the resulting algorithms are non-uniform and have a large (and unknown) parameter dependence~$f(k)$. In this paper, our goal is to understand the structure of \mwns and develop an efficient (and uniform) parameterized FPT algorithm for the problem.

\paragraph*{Our results}
The main result of our paper is the following theorem, showing that \mwnsshort has a uniform FPT algorithm with parameter dependence~$2^{\Oh(k \log k)}$.

\begin{restatable}{theorem}{MWNSIsFPTState}
\label{theorem:MWNS_is_FPT}
$(\bigstar)$ \mwns can be solved in time~$2^{\Oh(k \log k)} \cdot n^{\Oh(1)}$.
\end{restatable}

The starting point for the algorithm is the approach for solving \textsc{Multiway Cut} via important separators. The \emph{pushing lemma} due to Marx~\cite{Marx06} \cite[Lemma 8.53]{CyganFKLMPPS15} states that for any instance~$(G,T,k)$ of~\textsc{Multiway Cut} and any choice of terminal~$t \in T$, there is an optimal solution which contains an \emph{important}~$(t, T \setminus \{t\})$-separator. As the number of important separators of size at most~$k$ is bounded by~$4^k$, we can construct a solution by picking an arbitrary terminal which is not yet fully separated from the remaining terminals and branching on all choices of including a~$(t, T \setminus \{t\})$-separator in the solution.

Adapting this strategy directly for \mwns fails for several reasons. Most importantly, since terminal pairs are allowed to remain in the same \cct, it is possible that in an instance with a solution~$S$ of size~$k$, there exists a terminal~$t \in T$ for which there are no~$(t, T\setminus \{t\})$-separators of size~$f(k)$ for any function~$f$. This happens when such a terminal~$t$ is located in a `central' block in the block-cut tree of~$G-S$. However, there always exists a terminal~$t' \in T$ for which there does exist a~$(t', T \setminus \{t'\})$-separator of size at most~$k+1$, and for which furthermore a variation of the pushing lemma can be proven: there is an optimal solution which contains all-but-one vertex of an important~$(t', T \setminus \{t'\})$-separator of size at most~$k+1$. Intuitively, such a terminal~$t'$ can be found in a leaf block of the block-cut tree of~$G-S$. Hence there \emph{exists} a terminal for which branching on important separators can make progress in identifying the solution, but not all terminals have this property and a priori it is not clear which is the right one.

To resolve this issue, we will effectively have the algorithm try all choices for the terminal~$t$ which is near-separated from all other terminals by an important separator. To ensure the branching factor of the resulting algorithm is bounded in terms of the parameter~$k$, while the number of terminals~$T$ may initially be arbitrarily large compared to~$k$, we therefore have to reduce the number of terminals to~$k^{\Oh(1)}$ in a preprocessing phase.

For the standard \textsc{Multiway Cut} problem, a preprocessing step based on the linear-programming relaxation of the problem can be used to reduce the number of terminals to~$2k$~\cite{CyganPPW13a} (cf.~\cite{Razgon11}). For the near-separation variant we consider, it seems unlikely that the linear-programming relaxation has the same nice properties (such as half-integrality) as for the original problem, let alone that the resulting fractional solutions are useful for a reduction in the number of terminals. As one of our main technical ingredients, we therefore develop a combinatorial preprocessing algorithm to reduce the number of terminals to a polynomial in the solution size~$k$ plus the size of a given (suboptimal) near-separator~$S$, which will be available via the technique of iterative compression~\cite{ReedSV04} \cite[\S 4]{CyganFKLMPPS15}. The preprocessing step is based on concrete reduction rules operating in the graph.

\begin{restatable}{theorem}{TerminalBoundingState}
\label{theorem:terminalbounding}
$(\bigstar)$ There is a polynomial-time algorithm that, given an instance~\gtk of \mwnsshort and a multiway near-separator~$\hat{S}$ for terminal set~$T$ in~$G$, outputs an equivalent instance~$(G',T',k')$ such that: 
\begin{enumerate}
    \item $G'$ is an induced subgraph of $G$,
    \item $T'$ is a subset of~$T$ of size~$\Oh (k^5 \cdot |\hat{S}|^4 )$, and
    \item $k' \leq k$.
\end{enumerate}
Moreover, there is a polynomial-time algorithm that, given a solution~$S'$ for~$(G', T', k')$, outputs a solution~$S$ of~\gtk.
\end{restatable}

Note that the algorithm even runs in \emph{polynomial-time}, and may therefore be a useful ingredient to build a polynomial kernelization for this problem or variations thereof. To obtain this terminal-reduction algorithm, it turns out to be useful to know whether the role of a vertex~$x$ in a suboptimal near-separator~$S$ can be taken over by~$\Oh(k)$ alternative vertices. If not, then this immediately leads to the conclusion that such a vertex~$x$ belongs to any optimal solution to the problem. On the other hand, knowing a small vertex set~$S_x$ for which~$(S \setminus \{x\}) \cup S_x$ is also a near-separator reveals a lot of structure in the instance which can be exploited by the reduction rule. This usage is similar as the use of the \emph{blocker}~\cite[\S 9.1.3]{CyganFKLMPPS15} in Thomass\'e's kernelization algorithm for \textsc{Feedback Vertex Set}~\cite{Thomasse10}. 

Given a suboptimal near-separator~$S$, we are therefore interested in determining, for a given~$x \in S$, whether it is possible to obtain a near-separator~$S'$ by replacing~$x$ by a set of~$\Oh(k)$ vertices. This is equivalent to finding a near-separator of size~$\Oh(k)$ which avoids the use of vertex~$x$ in the graph~$G' := G - (S \setminus \{x\})$. Hence this task effectively reduces to finding a solution not containing~$x$ in the graph~$G'$ for which~$\{x\}$ forms a near-separator. We give a polynomial-time 14-approximation for this problem.

\begin{restatable}{theorem}{BlockerState}
\label{theorem:blocker}
$(\bigstar)$ There is a polynomial-time algorithm that, given a graph~$G$, terminal set~$T\subseteq V(G)$, and a vertex~$x \in V(G)$ such that~$\{x\}$ is a multiway near-separator for terminal set~$T$ in~$G$, outputs a multiway near-separator~$S_x \subseteq V(G) \setminus \{x\}$ for~$T$ in~$G$ such that~$|S_x| \leq 14 |S^*_x|$, where~$S^*_x \not \ni x$ is a smallest multiway near-separator for~$T$ in~$G$ that avoids~$x$.
\end{restatable} 

Theorem~\ref{theorem:blocker} can be compared to a result for the \textsc{Chordal Deletion} problem, where the goal is to delete a minimum number of vertices to break all induced cycles of length at least four (holes). A key step in the polynomial kernelization algorithm for the problem due to Jansen and Pilipczuk is a subroutine (\cite[Lemma 1.3]{JansenP18}) which, given a graph~$G$ and vertex~$x$ for which~$G-\{x\}$ is chordal, outputs a set of some~$\ell \geq 0$ holes pairwise intersecting only in~$x$, together with a vertex set~$S$ of size at most~$12\ell$ not containing~$x$ whose removal makes~$G$ chordal. Hence~$S$ is a 12-approximation for the problem of finding a chordal deletion set which avoids~$x$.

\paragraph*{Related work}
Apart from the aforementioned work on graph separation problems, the work of Golovach and Thilikos~\cite{GolovachT19} is related to our setting. They consider the problem of removing at most~$k$ edges from a graph to split it into exactly~$t$ \ccts~$C_1, \ldots, C_t$ such that~$C_i$ has edge-connectivity at least~$\lambda_i$ for a given sequence~$(\lambda_1, \ldots, \lambda_t)$. Besides recursive understanding, which only leads to non-uniform FPT classifications, another generic tool for deriving fixed-parameter tractability of separation problems is the treewidth reduction technique by Marx and Razgon~\cite{MarxR14}. To be able to apply the technique, the number of terminals must be bounded in terms of the parameter~$k$, which is not the case in general. Even if one uses Theorem~\ref{theorem:terminalbounding} to bound the number of terminals first, it is not clear if the technique can be applied since the solutions to be preserved are not minimal separators in the graph. Furthermore, successful application of the technique would give double-exponential algorithms at best, due to having to perform dynamic programming on a tree decomposition of width~$2^{\Omega(k)}$.

The problem of computing a near-separator avoiding a vertex~$x$ is related to the problem of computing an  \emph{$r$-fault tolerant} solution to a vertex deletion problem, which is a solution from which any $r$ vertices may be omitted without invalidating the solution. The computation of $r$-fault tolerant solutions has been studied for the \textsc{Feedback Vertex Set} problem~\cite{BlazejCKKSV21}, which has a polynomial-time $\Oh(r)$-approximation.

The FPT algorithm for \textsc{Subset Feedback Vertex Set} in undirected graphs due to Cygan et al.~\cite{CyganPPW13} bears some similarity to ours, in that it also uses reduction rules to bound the number of terminals followed by an algorithm which is exponential in the solution size and the number of terminals. However,  \textsc{Subset Feedback Vertex Set} behaves differently from the problem we consider, since in the latter the structures to be hit always involve \emph{pairs} of terminals to be near-separated. On the other hand, a solution to \textsc{Subset Feedback Vertex Set} will reduce the connectivity in the graph to the extent that there will no longer be two \ivd paths between any terminal pair, but also has to ensure that there are no cycles through a single terminal. This leads to significant differences in the approach.

\paragraph*{Organization} 
We begin with short preliminaries with the crucial definitions.
We prove our main Theorem~\ref{theorem:MWNS_is_FPT} in Section~\ref{FPT_algorithm_for_MWNS} by assuming Theorem~\ref{theorem:terminalbounding}. 
Next, in Section~\ref{section:BlockerConstruction}, we prove Theorem~\ref{theorem:blocker} by giving a polynomial-time construction of a near-separator avoiding a specific vertex. 
The proof of Theorem~\ref{theorem:terminalbounding} is given in Section~\ref{sec:bounding_terminals}. \rearng{The proofs of statements marked with $(\bigstar)$ are located in the appendix.} \sk{The proofs of statements marked with $(\bigstar)$ are located in the full version TODO: Add Citation.}

\section{Preliminaries}
\label{sec:prelims}
\subparagraph{Graphs.} 
We use standard graph-theoretic notation, and we refer the reader to Diestel~\cite{Diestel17} for any undefined terms. We consider simple unweighted undirected graphs. A graph~$G$ has vertex set~$V(G)$ and edge set~$E(G)$. We use shorthand~$n = |V(G)|$ and~$m = |E(G)|$. The set~$\{1, \ldots, \ell\}$ is denoted by~$[\ell]$. The open neighborhood of~$v \in V (G)$ is~$N_G(v) := \{u \mid \{u,v\} \in E(G)\}$, where we omit the subscript~$G$ if it is clear from context. For a vertex set~$S \subseteq V (G)$ the open neighborhood of~$S$, denoted~$N_G (S)$, is defined as~$S := \bigcup_{v \in S} N_G (v) \setminus S$.
For~$S \subseteq V(G)$, the graph induced by~$S$ is denoted by~$G[S]$.
For two vertices~$x,y$ in a graph~$G$, an $x$-$y$ path is a sequence~$(x=v_1, \ldots, v_k = y)$ of vertices such that~$\{v_i, v_{i+1}\} \in E(G)$ for all~$i \in [k-1]$.
Furthermore, the vertices~$v_2. \ldots, v_{k-1}$ are called the \emph{internal vertices} of the $x$-$y$ path. 
Given a path~$P = (v_1, \ldots, v_k)$ and indices~$i, j \in [k]$, with~$j \geq i$, we use~$P[v_i, v_j]$ to denote the subpath of the path~$P$ which starts from~$v_i$ and ends at~$v_j$. 
Moreover, we use shorthand~$P(v_i, v_j] = P[v_i, v_j] -\{v_i\}, P[v_i, v_j) = P[v_i, v_j] -\{v_j\}$, and~$P(v_i, v_j) = P[v_i, v_j] -\{v_i, v_j\}$.
Given a $p_1$-$p_k$ path~$P = (p_1, \ldots, p_k)$ and a~$q_1$-$q_\ell$ path~$Q = (q_1, \ldots, q_\ell)$ with~$p_k = q_1$ such that~$P$ and~$Q$ are \ivd, we use~$P \cdot Q$ to denote the $p_1$-$q_\ell$ path~$(p_1, \ldots, p_k, q_2, \ldots, q_\ell)$ obtained by first traversing $P$ and then~$Q$. 
We say that a path~$P$ in $G$ \emph{intersects} a vertex~$v_i \in V(G)$ if $v_i \in V(P)$, similarly, for a set~$S \subseteq V(G)$, we say that path~$P$ intersects~$S$ if~$V(P) \cap S \neq \emptyset$.
For~$S \subseteq V(G)$, an~$x$-$y$ path in~$G$ is called an $S$-path if~$x, y \in S$.
For~$S \subseteq V(G)$, cycles~$C_1, C_2$ in~$G$ are said to be~\emph{$S$-disjoint} if~$V(C_1) \cap V(C_2) \subseteq S$.

We now define a few basic notations about block-cut graphs (for completeness we define the notion of block-cut graph in Definition~\ref{definition:block-cut_graph}) that we will use henceforth. Given a graph~$G$ with \ccts $C_1, \ldots, C_m$, a \emph{rooted block-cut forest}~$\calF$ of~$G$ is a block-cut forest containing block-cut trees~$\calT_1, \ldots, \calT_m$ such that for each~$i \in [m]$, the tree~$\calT_i$ is a block-cut tree of~$C_i$ that is rooted at an arbitrary block of~$\calT_i$. 
Given a rooted forest~$\calF$ and a vertex~$v \in V(\calF)$, we use~$\mbox{parent}_{\calF} (v)$ to denote the parent of~$v$ (if~$v$ is a root then~$\mbox{parent}_{\calF} (v) = \emptyset$)
and~$\mbox{child}_\calF (v)$ to denote the set containing all children of~$v$ (if~$v$ is a leaf then~$\mbox{child}_{\calF} (v) = \emptyset$).
Given a rooted block-cut forest~$\calF$ of~$G$, and a node~$d$ of~$\calF$, we use~$V_G(\calF_d)$ to denote the vertices of~$G$ occurring in blocks of the subtree rooted at~$d$. 
Furthermore, we use~$G_d$ to denote the graph induced by the vertex set~$V_G (\calF_d)$, i.e., $G_d := G[V_G(\calF_d)]$.
We also need the following observations.
\begin{observation}
\label{observation:property_of_blocks}
	Let~$G$ be a graph, and let~$B_1, B_2$ be two distinct blocks of~$G$ such that~$V(B_1) \cap V(B_2) = \{v\}$. In the block-cut graph~$G'$ of~$G$, it holds that the distance between blocks~$B_1$ and~$B_2$ is two with~$v$ as an intermediate vertex.
\end{observation}

\begin{observation}\label{Observation:Block_contains_atmost_1-terminal}
    Consider a graph~$G$, terminal set~$T \subseteq V(G)$, and a MWNS~$S \subseteq V(G)$ of~$(G, T)$. Then each block~$B$ of~$G-S$ contains at most one terminal. 
\end{observation}
Two \mwnsshort instances~\gtk and~$(G', T', k')$ are said to be \emph{equivalent} if it holds that~\gtk is a YES-instance of \mwnsshort if and only if~$(G', T', k')$ is a YES-instance of~\mwnsshort. 
An instance~\gtk of \mwnsshort is said to be \emph{non-trivial} 
if~$\emptyset$ is not a solution of~\gtk. 
A terminal~$t \in T$ is said to be \emph{nearly-separated} in~$G$ if there does not exist another terminal~$t' \in \tminust$ such that there are 2 \ivd $t$-$t'$ paths in~$G$.

Throughout this manuscript we use~\mwns (\mwnsshort) to denote the parameterized version of the multiway near-separator problem, 
whereas given a graph~$G$ and terminal set~$T$ we use multiway near-separator (MWNS) to refer to the graph-theoretic concept of nearly-separating a terminal set~$T$ in~$G$. 
Formally, it is defined as follows.
\begin{definition}[Multiway near-separator (MWNS)]
       Given a graph~$G$ and terminal set~$T \subseteq V(G)$, a set~$S \subseteq V(G)$ is called a multiway near-separator (MWNS) of~$(G, T)$ if~$S \cap T = \emptyset$ and there does not exist a pair of distinct terminals~$t_i, t_j \in T$ such that~$G - S$ contains two \ivd $t_i$-$t_j$ paths.
\end{definition}

\begin{definition}[$r$-redundant~MWNS]
Given a graph~$G$ and terminal set~$T \subseteq V(G)$, a set~$S^* \subseteq V(G) \setminus T$ is an~$r$-redundant~MWNS of~$(G, T)$ if for all~$R \subseteq S^*$ with~$|R| \leq r$, the set~$S^* \setminus R$ is a~MWNS of~$(G, T)$.
\end{definition}

\begin{definition}[$x$-avoiding~MWNS]
    Given a graph $G$, terminal set~$T \subseteq V(G)$, and a vertex $x\in V(G)$, a set~$S_x \subseteq V(G) \setminus T$ is called an $x$-avoiding~MWNS of~$(G, T)$ if $x \notin S_x$ and $S_x$ is a MWNS of~$(G, T)$. Among all $x$-avoiding MWNS of~$(G, T)$, one with the minimum cardinality is called a minimum $x$-avoiding~MWNS of~$(G, T)$.
\end{definition}

Next, we define~$T$-cycle and give a characterization of MWNS in terms of hitting $T$-cycles.

\begin{definition}[$T$-cycle and~$T$-cycle on~$x$] 
        Given a graph~$G$ and terminal set~$T \subseteq V(G)$, a cycle~$C$ in~$G$ is called a~$T$-cycle if~$|V(C) \cap T| \geq 2$. Moreover, if~$C$ also contains a vertex~$x \in V(G)$, then~$C$ is called a~$T$-cycle on~$x$.
\end{definition}

We now show that several ways of looking at a near-separator are equivalent.
\begin{restatable}{proposition}{DifferentWaysOfLookingAtMWNSState} 
\label{Proposition:DifferentWaysOfLookingAtMWNS}
$(\bigstar)$ Given a graph~$G$, terminal set~$T \subseteq V(G)$, and a non-empty set~$S \subseteq V(G) \setminus T$, the following conditions are equivalent:
\begin{enumerate}
    \item For each pair of distinct terminals~$t_i, t_j \in T$, the graph~$G - S$ does not contain $t_i$-$t_j$ paths~$P_1, P_2$ which are pairwise \ivd. (Note that~$P_1$ may be identical to~$P_2$ if there are no internal vertices.)\label{Condition1:DifferentWaysOfLookingAtMWNS}
    
    \item For each pair of distinct terminals~$t_i,t_j \in T$, there is a vertex~$v \in V(G) \setminus T$ such that~$t_i$ and~$t_j$ belong to different \ccts of~$G-(S \cup \{v\})$.\label{Condition2:DifferentWaysOfLookingAtMWNS}
    
    \item The set~$T$ is an independent set and~$G-S$ does not contain a simple cycle~$C$ containing at least two terminals (i.e., a $T$-cycle).\label{Condition3:DifferentWaysOfLookingAtMWNS}
\end{enumerate}
\end{restatable}

Due to space constraints we defer the remaining preliminaries about graphs (including block-cut graphs and important separators) and parameterized algorithms to Appendix~\ref{sec:AdditionalPreliminaries}.

\section{\texorpdfstring{FPT algorithm for~\mwns}{FPT algorithm for MWNS}}
\label{FPT_algorithm_for_MWNS}
In this section we prove Theorem~\ref{theorem:MWNS_is_FPT} assuming Theorem~\ref{theorem:terminalbounding}, which we prove later in Section~\ref{sec:bounding_terminals}. We use the combination of bounded search trees and iterative compression \cite[\S 3--4] {CyganFKLMPPS15} to obtain the FPT algorithm. Towards this, we first present the following structural lemma for a~MWNS $S \subseteq V(G)$ of~$(G, T)$. It says that in~$G-S$, there is a terminal that can simultaneously be separated from \emph{all} other terminals by the removal of a single non-terminal~$v$.

\begin{restatable}{lemma}{EasilySeparableTerminalState}
\label{lemma:easilyseparableterminal}
$(\bigstar)$ Let~\gtk be a non-trivial instance of \mwnsshort, and let~$S \subseteq V(G) \setminus T$ be a solution. Then there exists a terminal~$t \in T$ and a non-terminal vertex~$v \in V(G) \setminus T$ such that~$S \cup \{v\}$ is a~\ttsep. \end{restatable}

Marx~\cite{Marx06}~\cite[Lemma 8.18]{CyganFKLMPPS15} introduced a pushing lemma for~\textsc{Multiway Cut} to prove that \textsc{Multiway Cut} is FPT. In the following lemma, we present a pushing lemma for~\mwnsshort. 

\begin{restatable}{lemma}{PushingLemmaForMWNSState}[Pushing lemma for {\mwnsshort}]
\label{lemma:pushinglemmaformwns}
$(\bigstar)$ Let~\gtk be a non-trivial instance of \mwnsshort and let~$S \subseteq V(G) \setminus T$ be a solution. Then there exists a terminal~$t \in T$ and a solution~$S^* \subseteq V(G) \setminus T$ with~$|S^*|\leq |S|$ for which one of the following holds:
\begin{enumerate}
    \item there is an important~\ttsep~$S_t^*$ of size at most~$k$ such that~$S_t^* \subseteq S^*$, or
    \label{condition:pushinglemmaformwns1}
    
    \item there is an important~\ttsep~$S_t$ of size at most~$(k+1)$, and there exists a vertex~$v \in S_t$ such that~$(S_t \setminus \{v\}) \subseteq S^*$.
    \label{condition:pushinglemmaformwns2}
\end{enumerate}
\end{restatable}

The following lemma forms the heart of the FPT algorithm (Theorem~\ref{theorem:MWNS_is_FPT}). It says that there exists an FPT algorithm that can compress a~$k+1$-sized MWNS of~$(G, T)$ to a~$k$-sized MWNS if~\gtk is a YES-instance of \mwnsshort. This is effectively the compression step of the iterative compression technique.  
\begin{restatable}{lemma}{CompressionStepForMWNSState}
\label{lemma:Compression_Step_for_MWNS}
    $(\bigstar)$ There is an algorithm that, given an instance~\gtk of \mwnsshort and a set~$S_{k+1} \subseteq V(G) \setminus T$ of size~$k+1$ such that~$S_{k+1}$ is a MWNS of~$(G, T)$, runs in time~$2^{\Oh(k \log k)} \cdot n^{\Oh(1)}$ and outputs a solution of~\gtk (of size~$k$) if it exists.
\end{restatable}

Given Lemma~\ref{lemma:Compression_Step_for_MWNS}, the proof of Theorem~\ref{theorem:MWNS_is_FPT} follows by applying the standard technique of iterative compression. \rearng{We defer the formal proof to Appendix~\ref{section:ProofOfThm:MWNS_Is_FPT}.}\sk{The formal proof can be found in the full version.}

\section{Constructing a near-separator avoiding a specified vertex}
\label{section:BlockerConstruction}
In this section we prove Theorem~\ref{theorem:blocker}. Throughout the algorithm, we use the perspective provided by Proposition~\ref{Proposition:DifferentWaysOfLookingAtMWNS} that a MWNS is a set intersecting all~$T$-cycles. Note that since~$\{x\}$ is a MWNS for~$(G,T)$, the set~$T$ must be an independent set. Before presenting the algorithm, we define some notations which we will use during the algorithm.

\begin{definition}[$\mathcal{C}(v)$ and~$\mathcal{C}_{\geq 1} (v)$]
\label{Definition:GrandchildrenAndInterestingGrandchildren}
    Given a graph~$G$, terminal set~$T \subseteq V(G)$, and a MWNS~$\{x\} \subseteq V(G)$ of~$(G, T)$, let~$\calF$ be a rooted block-cut forest of~$G - \{x\}$. Let~$v \in V(\calF)$ be a cutvertex. Then we use~$\mathcal{C}(v)$ to denote all the grandchildren (cutvertices) of~$v$ in the subtree~$\calF_v$, i.e., $\mathcal{C}(v) := \bigcup_{y \in \text{child}_\calF(v)}  \text{child}_\calF (y)$. If~$v$ does not have a grandchild then~$\mathcal{C} (v) := \emptyset$.
    We use~$\mathcal{C}_{\geq 1} (v) \subseteq \mathcal{C} (v)$ to denote the cutvertices of~$\mathcal{C} (v)$ such that for each vertex~$c \in \mathcal{C}_{\geq 1} (v)$, the graph~$G_c = G[V_G (\calF_c)]$ contains a vertex~$p \in N_G(x)$ such that there is a~$c$-$p$ path~$P$ in~$G_c$ which contains at least one terminal, i.e.,~$|V(P) \cap T| \geq 1$.
\end{definition}

During the construction of an approximate $x$-avoiding MWNS, we will often make use of Definition~\ref{Definition:GrandchildrenAndInterestingGrandchildren} to keep track of which cutvertices have a pending subgraph attached that can reach a neighbor of~$x$ by a simple path containing a terminal. Such subpaths can be combined to form $T$-cycles. We often use the fact that, in an undirected graph~$G$, it is possible to test in polynomial time whether there is a simple~$p$-$q$ path through a specified vertex~$t$; for example, by constructing a vertex-capacitated flow network in which~$t$ has a capacity of~$2$ and all other vertices a capacity of~$1$, and testing for a flow from~$\{p,q\}$ to~$\{t\}$.

Next, we prove some properties about the sets~$\mathcal{C} (t)$ and~$\mathcal{C}_{\geq 1} (t)$ defined above. We need these properties during the analysis phase (Section~\ref{section:AnalysisOfBlockerAlgo}) of the blocker algorithm.
\begin{restatable}{proposition}{GrandchildCannotBeATerminalState}
\label{Proposition:Grandchild_of_a_Terminal_can_NOT_be_a_Terminal}
    Given a graph~$G$, terminal set~$T \subseteq V(G)$, and a MWNS~$\{x\} \subseteq V(G)$ of~$(G, T)$, let $\calF$ be a rooted block-cut forest of~$G-\{x\}$. Let~$t \in T$ be a cutvertex of~$\calF$ and let~$\calC (t)$ be the set of grandchildren of~$t$ in the subtree~$\calF_t$ as defined in Definition~\ref{Definition:GrandchildrenAndInterestingGrandchildren}. Then we have~$\mathcal{C} (t) \cap T = \emptyset$. 
\end{restatable}
\begin{claimproof}
    Assume for a contradiction that there exists a vertex~$t' \in \mathcal{C} (t) \cap T$. First, note that~$t' \neq t$, as a cutvertex is present exactly once in a block-cut forest. Thus, we have~$t' \in \tminust$. 
    Let~$B := \text{parent}_{\calF} (t')$. Since~$\{t', B\} \in E (\calF)$, we have~$t' \in V(B)$ by definition of block-cut forest.
    Moreover, as~$B$ is the parent of~$t'$ and $t'$ is a grandchild of~$t$, we have~$\{t, B\} \in E(\calF)$ and hence~$t \in V(B)$. Note that~$B$ is a block in~$G-\{x\}$ which contains two distinct terminals~$t, t'$, a contradiction to Observation~\ref{Observation:Block_contains_atmost_1-terminal}.
\end{claimproof}

\begin{restatable}{proposition}{NumberOfInterestingCutVerticesAreBoundedState}
\label{Proposition:NumberOfInterestingCutVerticesAreBounded}
   Given a graph~$G$, terminal set~$T \subseteq V(G)$, and a MWNS~$\{x\} \subseteq V(G)$ of~$(G, T)$, let $\calF$ be a rooted block-cut forest of~$G-\{x\}$.
   Let~$t \in T$ be a cutvertex of~$\calF$ and let~$\mathcal{C}_{\geq 1} (t)$ be the subset of grandchildren of~$t$ in~$\calF_t$ defined in Definition~\ref{Definition:GrandchildrenAndInterestingGrandchildren}.
    Let~$B$ be a node of~$\calF_t$ such that the graph~$G[V_G(\calF_B) \cup \{x\}]$ does not contain a~$T$-cycle.
    Then the number of cutvertices below~$B$ in~$\calF_B$ which also belong to the set~$\mathcal{C}_{\geq 1} (t)$ is at most one, i.e., we have~$|V_G (\calF_B) \cap {C}_{\geq 1} (t)| \leq 1$.
\end{restatable}

\begin{claimproof}
    First of all, note that if~$B$ belongs to~$\calC (t)$ (see Definition~\ref{Definition:GrandchildrenAndInterestingGrandchildren} for the definition of~$\calC (t)$) or below in the subtree~$\calF_t$ then the claim trivially holds, because in that case we have~$|V_G(\calF_B) \cap {C}_{\geq 1} (t)| \leq 1$. Hence consider the case when either~$B \in \mbox{child}_\calF (t)$ or~$B=t$, and assume for a contradiction that $|V_G(\calF_B) \cap {C}_{\geq 1} (t)| \geq 2$. Let~$c_1, c_2$ be two distinct cutvertices in~$V_G(\calF_B) \cap {C}_{\geq 1} (t)$.
    By definition of the set~$\mathcal{C}_{\geq 1} (t)$ and the fact that~$c_1, c_2 \in \mathcal{C}_{\geq 1} (t)$, we know that for each~$i \in [2]$, the graph~$G_{c_i}$ contains a vertex~$p_i \in N_G(x)$ such that there is a~$c_i$-$p_i$ path~$P_i$ in~$G_{c_i}$ containing a terminal~$t_i$. Moreover, since~$c_1 \neq c_2$, the paths~$P_1$ and~$P_2$ are vertex disjoint. Next, we do a case distinction based on whether~$B \in \mbox{child}_\calF (t)$ or~$B=t$.
    
    \subparagraph*{Case 1. When~$B \in \mbox{child}_\calF (t)$.}
    
            Since~$B \in \mbox{child}_\calF (t)$ and by definition (of~$\calC_{\geq 1} (t)$)~$c_1, c_2$ are grandchildren of~$t$, we have~$c_1, c_2 \in \text{child}_\calF (B)$. Hence, we have~$c_1, c_2 \in V_G(B)$. Next, we construct a cycle~$C$ in~$G[V_G(\calF_B) \cup \{x\}]$ as follows.
            Let~$C := \{x, p_1\} \cdot P_1 [p_1, c_1] \cdot R_{12} [c_1, c_2] \cdot P_2 [c_2, p_2] \cdot \{p_2, x\}$, where~$R_{12}$ is a path between cutvertices~$c_1, c_2 \in V_G (B)$ inside the block~$B$.
            Note that the cycle~$C$ in~$G[V_G(\calF_B) \cup \{x\}]$ is simple and contains two distinct terminals~$t_1, t_2 \in T$, a contradiction to the fact that there is no~$T$-cycle in~$G[V_G (\calF_B) \cup \{x\}]$.
            
    \subparagraph*{Case 2. When~$B=t$.}
    
            For~$i \in [2]$, let~$B_i$ be the parent of~$c_i$. Note that if~$B_1 = B_2$ then similarly to Case 1, we can obtain a~$T$-cycle in~$G[V_G (\calF_{B_1}) \cup \{x\}]$, which is also a~$T$-cycle in the supergraph~$G[V_G (\calF_B) \cup \{x\}]$, again a contradiction to the fact that there is no~$T$-cycle in~$G[V_G (\calF_B) \cup \{x\}]$.
            Hence assume that~$B_1 \neq B_2$. Next, we show that even in this case we can obtain a~$T$-cycle $C$ in $G[V_G (\calF_B) \cup \{x\}]$, yielding a contradiction.
            Indeed, we can use~$C := \{x, p_1\} \cdot P_1[p_1, c_1] \cdot R_1[c_1, B] \cdot R_2[B, c_2] \cdot P_2[c_2, p_2] \cdot \{p_2, x\}$, where for~$i \in [2]$, the path~$R_i$ is a path between vertices~$c_i, B \in V_G(B_i)$ inside the block~$B_i$.
            
            Since both the above cases lead to a contradiction, this concludes the proof of Proposition~\ref{Proposition:NumberOfInterestingCutVerticesAreBounded}.
\end{claimproof}

\subsection{Algorithm}
\label{subsection:BlockerAlgo}
In this section we present a recursive algorithm~$\mbox{Blocker}(G, T, x)$ to construct a set~$S_x \subseteq V(G) \setminus (T \cup \{x\})$ such that~$S_x$ is a \mwnsshort of~$(G, T)$ and~$|S_x| \leq 14$\optx$(G, T)$, where~\optx$(G, T)$ is the cardinality of a minimum $x$-avoiding \mwnsshort of~$(G, T)$. The algorithm effectively takes a graph~$G$, terminal set~$T$, a vertex~$x \in V(G) \setminus T$ (such that~$\{x\}$ is a MWNS of~$(G, T)$) as input, and computes a vertex set~$Z \subseteq V(G) \setminus (T \cup \{x\})$ to hit certain types of~$T$-cycles in~$G$. It combines~$Z$ with the result of recursively computing a solution for~$\mbox{Blocker}(G-Z, T, x)$. For ease of understanding, we present the algorithm step by step with interleaved comments in italic font, whenever required.

\begin{enumerate}

    \item If the graph~$G$ does not contain a~$T$-cycle on~$x$ then return~$Z = \emptyset$.
    
    \item Construct a rooted block-cut forest~$\mathcal{F}$ of~$G-\{x\}$.
    
    \item Choose a deepest node~$d$ in the block cut forest~$\mathcal{F}$ such that the graph~$G[V_G(\mathcal{F}_d) \cup \{x\}]$ contains a~$T$-cycle.\label{Algo:A:Step:ChooseADeepestNode}

    \textit{Note that such a vertex~$d$ exists as there is a~$T$-cycle on $x$ in the graph~$G$ while a simple cycle in~$G$ visits vertices from at most one tree~$\calT \in \calF$.}
        
    \item 
    Consider the following cases.\label{Algo:A:Step:CaseDistinction}
    
    \begin{enumerate}
        \item \textbf{If~$d$ is a cutvertex and~$d \notin T$.} Let~$Z := \{d\}$. Then return~$(Z \cup \mbox{Blocker}(G-Z, T, x))$.
    
    \textit{In this case, it is easy to observe that the set~$Z = \{d\}$ is a MWNS of~$(G[V_G (\calF_d) \cup \{x\}], T)$ because~$d$ is a deepest node satisfying the conditions of Step~{\ref{Algo:A:Step:ChooseADeepestNode}}}.
    
        \item \textbf{If~$d$ is a cutvertex and~$d \in T$.} Let~$\mathcal{C}_{\geq 1} (d)$ be the subset of grandchildren of~$d$ defined using Definition~\ref{Definition:GrandchildrenAndInterestingGrandchildren}. \label{Algo:A:Step:ConstructionOfZWhendIsATerminalcutvertex}
                Let~$Z := \mathcal{C}_{\geq 1} (d)$ and return~$(Z \cup \mbox{Blocker}(G-Z, T, x))$.

    \textit{The set~$Z$ is a MWNS of~$(G[V_G (\calF_d) \cup \{x\}], T)$ by our choice of~$d$, definition of the set~$\mathcal{C}_{\geq 1} (d)$, the fact that a $T$-cycle contains at least 2 terminals, and for each block~$B \in \mbox{child}_\calF (d)$ we have $V(B) \cap (T \setminus \{d\}) = \emptyset$ due to Observation~\ref{Observation:Block_contains_atmost_1-terminal}.} 

        \item \textbf{If~$d$ is a block.}\label{Algo:A:Step:ConstructionOfZWhendIsABlock}
            \begin{itemize}
                \item Let~$D_T := d-T$, i.e., $D_T := G[V(d) \setminus T]$. Note that the block~$d$ of~$G-\{x\}$ contains at most 1 terminal due to Observation~\ref{Observation:Block_contains_atmost_1-terminal}. In the case when~$V(d) \cap T = \emptyset$, we have~$D_T = d  = G[V(d)]$.
                Let~$\calC^d := \text{child}_\calF (d) \setminus T$ and partition~$\calC^d$ as follows.
                    \begin{itemize}
                        \item Let~$\calC^d _{\geq 2} \subseteq \calC^d$ be the set such that for each (cut)vertex~$c \in \calC^d _{\geq 2}$ the graph~$G_c := G[V_G (\calF_c)]$ contains a vertex~$p \in N_G(x)$ such that there is a $c$-$p$ path~$P$ in~$G_c$ which contains at least 2 terminals, i.e., $|V(P) \cap T| \geq 2$.

                        \item Let~$\calC^d _1 \subseteq \calC^d \setminus \calC^d _{\geq 2}$ be the subset of remaining vertices of~$\calC^d$ such that for each vertex~$c \in \calC^d _1$ the graph~$G_c$ 
                        contains a vertex~$p \in N_G(x)$ such that there is a $c$-$p$ path~$P$ in~$G_c$ which contains 1 terminal, i.e., $|V(P) \cap T| = 1$.

                        \item Let~$\calC^d _0 \subseteq \calC^d \setminus (\calC^d _{\geq 2} \cup \calC^d _1)$ be the subset of remaining vertices of~$\calC^d$ such that for each vertex~$c \in \calC^d _0$ the graph~$G_c$ 
                        contains a vertex~$p \in N_G(x)$ such that there is a $c$-$p$ path~$P$ in~$G_c$.

                        \item Let~$\calC^d _\emptyset := \calC^d \setminus (\calC^d _{\geq 2} \cup \calC^d _1 \cup \calC^d _0)$ be the remaining elements of~$\calC^d$.
                    \end{itemize}

                \item Apply Gallai's theorem (\cite[Thm 9.2, Lemma 9.3]{CyganFKLMPPS15}, cf.~\cite[Thm 73.1]{Schrijver03}) on the graph~$D_T$ with~$Q = \calC^d _{\geq 2} \cup \calC^d _1$ to obtain a maximum-cardinality family~$\mathcal{P}_Q$ of pairwise vertex disjoint~$Q$-paths in~$D_T$, along with a vertex set~$Z_1 \subseteq V(D_T)$ of size at most~$2 |\mathcal{P}_Q|$ such that the graph~$D_T - Z_1$ has no~$Q$-path.

                \item Let~$A := \calC^d _{\geq 2}$ and~$B := \calC^d _0 \cup (N_G(x) \cap V(d))$. Next, compute a minimum~$(A,B)$-separator in the graph~$D_T$ using Edmonds-Karp algorithm~\cite{EdmondsK72} which outputs a vertex set~$Z_2 \subseteq V (D_T)$.
                
            \end{itemize}

        \textit{Next, we do a case distinction based on whether or not the block~$d$ contains a terminal. Note that in the case when~$d$ does not contain a terminal then the set~$(Z_1 \cup Z_2)$ hits all $T$-cycles of~$G[ V_G(\calF_d) \cup \{x\} ]$. On the other hand, when~$d$ contains a terminal there could still be a $T$-cycle in the graph~$G[ V_G(\calF_d) \cup \{x\} ] - (Z_1 \cup Z_2)$ (see the third figure of Figure~\ref{fig:BlockerConstruction}). So our next steps are aimed at hitting those $T$-cycles (if any) that are not hit by the set~$(Z_1 \cup Z_2)$.}
        
        \begin{itemize}
            \item If~$V(d) \cap T = \emptyset$, then let~$Z_3, Z_4 := \emptyset$.
            
            \item Otherwise, there is a unique terminal in block~$d \subseteq G - x$ since~$x$ is a MWNS for~$(G,T)$. Let~$t$ be the terminal that belongs to the block~$d$.
                \begin{itemize}
                    \item Let~$\mathcal{D}$ be the set containing \ccts of~$D_T - (Z_1 \cup Z_2)$ that contain at least one neighbor of~$t$, i.e., for each \cct~$D \in \mathcal{D}$, we have~$V (D) \cap N(t) \neq \emptyset$.

                    \item Let~$\mathcal{D}^* \subseteq \mathcal{D}$ be the set such that for each~$D^* \in \mathcal{D}^*$, the \cct~$D^*$ contains a vertex from~$Q = (\mathcal{C}^d _{\geq 2} \cup \mathcal{C}^d _1)$. More precisely, we have $|V (D^*) \cap Q| = 1$ for each~$D^*$ as the set~$Z_1$ is hitting all~$Q$-paths.

                    \item Let~$V(\mathcal{D}^*) := \bigcup _{D^* \in \mathcal{D}^*} V(D^*)$ and define~$Z_3 := (\mathcal{C}^d _{\geq 2} \cup \mathcal{C}^d _1) \cap V(\mathcal{D}^*)$.
            
    \textit{Note that in the case when~$d$ contains a terminal~$t$ and~$t \in \text{child}_\calF (d)$, the way we have defined~$\mathcal{C}^d$, it does not contain~$t$. Hence, when~$t \in \text{child}_\calF (d)$ and the graph~$G_t$ has a vertex~$p \in N_G(x)$ such that there is a~$t$-$p$ path~$P$ in~$G_t$ that contains at least one terminal other than~$t$, i.e, $|V(P) \cap (T \setminus \{t\})| \geq 1$, there could still be $T$-cycles in~$G [V_G (\calF_d)] - \bigcup^3 _{i=1} Z_i$ (see the last figure of Figure~\ref{fig:BlockerConstruction}). So our next step is to hit all $T$-cycles (if any) containing the~$t$-$p$ path~$P$. Recall $\mathcal{C} _{\geq 1} (t)$ from  Definition~\ref{Definition:GrandchildrenAndInterestingGrandchildren}.}
    
            \begin{itemize}
                \item If $t \in \text{child}_\calF (d)$ and the graph~$G_t$ has a vertex~$p \in N_G(x)$ such that there is a $t$-$p$ path~$P$ in~$G_t$ with~$|V(P) \cap T| \geq 2$, then let~$Z_4 := \mathcal{C} _{\geq 1} (t)$.
        
                \item Otherwise, define~$Z_4 := \emptyset$.
            \end{itemize}           
        \end{itemize}
    \end{itemize}

    \textit{Finally, we try to break any interaction between vertices of~$V_G (\calF_d)$ and vertices of~$V_G (\calF) \setminus V_G (\calF_d)$ by adding the parent of~$d$ into the hitting set. But note that in the case when~$\text{parent}_{\calF} (d) \in T$, we can not add it to the hitting set~$Z \subseteq V(G) \setminus (T \cup \{x\})$.}

        \begin{itemize}
            \item If $\text{parent}_{\calF} (d) \in T$, then let $Z_5 := \emptyset$.

            \item Otherwise, $Z_5 := \text{parent}_{\calF}(d)$.
        \end{itemize}
        
        Let~$Z := \bigcup^5 _{i=1} Z_i$ and $\text{return} (Z \cup \mbox{Blocker} (G-Z, T, x))$.
    \end{enumerate}    
\end{enumerate}

This concludes the description of the algorithm. Summarizing, its main structure is to define a vertex set~$Z \subseteq V(G) \setminus (T \cup \{x\})$ to break $T$-cycles which are \emph{lowest} in the block-cut tree, include that set~$Z$ in the approximate solution, and complete the solution by recursively solving the problem on~$G - Z$.

\begin{figure}
    \centering
    \includegraphics[height=4cm]{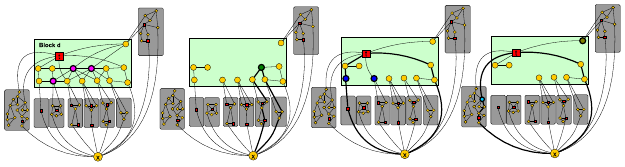}
    \caption{Illustration of Step~\ref{Algo:A:Step:ConstructionOfZWhendIsABlock}. 
    The leftmost figure shows the original graph~$G$ where terminals are represented by red squares and the pink vertices of block~$d$ represent the vertices of~$Z_1$.
     The second figure shows the construction of the set~$Z_2$ in the graph~$G-(Z_1 \cup \{t\})$ where the green vertex represents the vertex of~$Z_2$. It also shows a $T$-cycle represented by thick edges.
     The third figure shows a $T$-cycle (represented by thick edges) which appears after putting the terminal~$t$ back and the blue vertices represent the vertices of~$Z_3$.
     The last figure illustrates the construction of the set~$Z_4$ and~$Z_5$ where the cyan vertex and olive vertex represent the vertex of~$Z_4$ and~$Z_5$, respectively.}
    \label{fig:BlockerConstruction}
\end{figure}

\subsection{Analysis} 
\label{section:AnalysisOfBlockerAlgo}
The following lemma forms the heart of Theorem~\ref{theorem:blocker}. It says that if the above procedure adds a set~$Z$ during the construction of the approximate solution~$S_x$, then the optimum value of the remaining instance decreases by at least~$\frac{|Z|}{14}$. 

\begin{restatable}{lemma}{OPTDecreasesByAConstantFactorState}
\label{Lemma:OPTDecreasesByAFactorOf14}
If a single iteration of the algorithm on input~$(G,T,x)$ yields the set~$Z$, then~\optx$(G-Z, T) \leq$ \optx$(G, T) - \frac{|Z|}{14}$, where~\optx$(G, T)$ and~\optx$(G-Z, T)$ are the cardinalities of a minimum $x$-avoiding MWNS of~$(G, T)$ and~$(G-Z, T)$, respectively.
\end{restatable}

\begin{claimproof}
            Let~$S^*_x \subseteq V(G)$ be a minimum $x$-avoiding~MWNS of~$(G, T)$. If the algorithm stops before Step~$({\ref{Algo:A:Step:ChooseADeepestNode}})$ then~$Z = \emptyset$. Hence the inequality of the lemma trivially holds. Therefore assume that the algorithm reaches Step~\ref{Algo:A:Step:CaseDistinction}. Let~$d$ be the deepest node selected by the algorithm in Step~$({\ref{Algo:A:Step:ChooseADeepestNode}})$ to compute~$Z$. Let~$\hat{S} := S^*_x \setminus V_G (\calF_d)$. Note that~$\hat{S} \cap (T \cup \{x\}) = \emptyset$. Next, we observe the following properties about the set~$Z$ computed in Step~\ref{Algo:A:Step:CaseDistinction}.
            
            \begin{restatable}{claim}{ZisaMWNSofSubtreeBelowdState}
            \label{Claim:Z_is_a_MWNS_of_Subtree_below_d}
                $(\bigstar)$ Suppose the algorithm reaches Step~\ref{Algo:A:Step:CaseDistinction} and computes the set~$Z \subseteq V_G (\calF_d) \setminus T$. Then~$Z$ is a MWNS of~$(G_d + x, T)$, where~$G_d+x = G [V_G (\calF_d) \cup \{x\}]$.
            \end{restatable}
            \begin{claimproof}[Proof sketch]
            By choice of~$d$, each~$T$-cycle of~$G_d+x$ contains a vertex from block~$d$, or the cutvertex~$d$ itself. Since~$x$ is a MWNS for~$(G,T)$, each $T$-cycle also contains~$x$. The sets~$Z_1, \ldots, Z_4$ added to~$Z$ in the algorithm ensure different types of $T$-cycles of~$G_d+x$ are broken:~$Z_1$ covers all $T$-cycles consisting of two paths between~$x$ and~$d$, each containing at least one terminal;~$Z_2$ covers $T$-cycles consisting of one path between~$x$ and~$d$ containing two or more terminals, and another such path containing no terminals; the sets~$Z_3$ and~$Z_4$ cover $T$-cycles that go through a terminal in~$d$ (if there is one), as explained in the algorithm.            
            \end{claimproof}
        
            \begin{restatable}{claim}{AtLeastOneVertexFromBothInsideAndOutsideState}
            \label{Claim:Contains_at_least_one_vertex_from_both_inside_and_outside}
               $(\bigstar)$ Suppose the algorithm reaches Step~\ref{Algo:A:Step:CaseDistinction} and computes the set~$Z \subseteq V_G (\calF_d) \setminus T$. Then any~$T$-cycle in~$G- (Z \cup \hat{S})$ contains a vertex of~$V_G (\calF_d)$ (i.e., a vertex from~$G_d$) and a vertex from~$V_G (\calF) \setminus (V_G (\calF_d))$ (i.e., a vertex outside~$G_d + x$).
            \end{restatable}
            \begin{claimproof}[Proof sketch]
            Since~$\hat{S}$ contains all vertices of the solution~$S^*_x$ except those occurring in a block of the subtree~$\mathcal{F}_d$ of the block-cut tree, any~$T$-cycle disjoint from~$\mathcal{S}$ was intersected by~$S^*_x$ in a vertex of~$V_G(\mathcal{F}_d)$ and therefore uses a vertex of the latter set. On the other hand, since the previous claim shows that~$Z$ hits all the $T$-cycles which live in~$G_d+x$, a~$T$-cycle disjoint from~$Z$ has to use a vertex outside~$V_G(\mathcal{F}_d)$.
            \end{claimproof}
            
            \begin{restatable}{claim}{AnyOptimalMWNSEPicksEnoughVerticesBelowdState}
            \label{Claim:Any_Optimal_MWNS_excluding_x_picks_enough_vertices_below_d}
                $(\bigstar)$ Suppose the algorithm reaches Step~\ref{Algo:A:Step:CaseDistinction} and computes the set~$Z \subseteq V_G (\calF_d) \setminus T$. Then any minimum $x$-avoiding MWNS~$S^*_x$ of~$(G, T)$ satisfies~$|S^*_x \cap  V_G (\calF_d)| \geq \max \{1, \frac{|Z| - 2}{6}\}$.
            \end{restatable}
            \begin{claimproof}[Proof sketch]
            Any $x$-avoiding MWNS~$S^*_x$ contains a vertex from~$V_G(\mathcal{F}_d)$ because there is a $T$-cycle in~$G_d+x$, by choice of~$d$, which explains why the intersection is nonempty. The fact that~$S^*_x$ contains at least~$\frac{|Z|-2}{6}$ vertices from~$V_G(\mathcal{F}_d)$ can be seen as follows. We have~$|Z_4|, |Z_5| \leq 1$ by definition, so the largest~$Z_i$ of~$Z_1, Z_2, Z_3$ has at least~$\frac{|Z|-2}{3}$ vertices. Any solution contains at least~$|Z_i|/2$ vertices from~$V_G(\mathcal{F}_d)$ because the covering/packing duality for the three types of separators used to define~$Z_1,Z_2,Z_3$ ensures, for each of these sets, that to hit all the~$T$-cycles of the corresponding form, at least~$|Z_i|/2$ vertices are needed. For example, each $Q$-path~$P$ in the family~$\mathcal{P}_Q$ obtained during the construction of~$Z_1$ yields a~$T$-cycle when combined with paths from endpoints of~$P$ (which are cutvertices in~$Q$) to neighbors of~$x$ in two different subtrees of the block-cut forest~$\mathcal{F}$, showing that~$S^*_x \cap V_G(\mathcal{F}_d) \geq |\mathcal{P}_Q| \geq |Z_1|/2$.
            \end{claimproof}

            Next, we show that if a minimum $x$-avoiding MWNS~$S^*_x$ of~$(G, T)$ contains exactly 1 vertex from the subtree~$\calF_d$ then the set~$\hat{S} := S^*_x \setminus V_G (\calF_d)$ is a MWNS of $(G-Z, T)$.
            \begin{restatable}{claim}{ShatIsAMWNSIfOptContainsOneVertexFromSubtreeBelowdState}
            \label{Claim:ShatIsAMWNSIfOptContains1VertexFromSubtreeBelow_d}
                $(\bigstar)$ If~$|S^*_x \cap  V_G (\calF_d)| = 1$ then~$\hat{S} = S^*_x \setminus V_G (\calF_d)$ is a MWNS of~$(G-Z, T)$.
            \end{restatable}
            \begin{claimproof}[Proof sketch]
            If~$|S^*_x \cap V_G(\mathcal{F}_d)| = 1$, then from the $T$-cycle in~$G_d+x$ which we know to exist, the set~$S^*_x$ contains at most one vertex. In the most crucial case that the~$d$ is a block whose parent is a terminal, this means that~$S^*_x$ cannot break all paths from~$x$ through blocks in the subtree~$\mathcal{F}_d$ to the parent of~$d$. The only types of connections that the set~$Z$ does not break, and which can be part of $T$-cycles in~$G$, can be shown to be precisely paths from a neighbor of~$x$ to~$d$ in~$G_d$ that do not contain a terminal. But if~$|S^*_x \cap V_G(\mathcal{F}_d)| = 1$, then~$S^*_x$ does not break all such paths either. By a rerouting argument, this allows us to show that updating~$S^*$ by replacing~$S^*_x \cap V_G(\mathcal{F}_d)$ with~$Z$ gives a valid $x$-avoiding MWNS, which is equivalent to saying that~$\hat{S}$ is a MWNS of~$(G-Z, T)$.
            \end{claimproof}

            The following claim shows that if the set~$\hat{S}$ is not a MWNS of~$(G-Z, T)$, then we can find a vertex~$\hat{c}$ such that the set obtained after adding~$\hat{c}$ to the set~$\hat{S}$ is a MWNS of $(G-Z, T)$.
            
            \begin{restatable}{claim}{ShatISaNearMWNSState}
            \label{Claim:S-hat_is_at_most_1_away_from_a_MWNS}
                $(\bigstar)$ If~$\hat{S}$ is not a MWNS of~$(G-Z, T)$ then there exists a vertex~$\hat{c} \in V(G) \setminus (T \cup \{x\})$ such that~$\hat{S} \cup \{\hat{c}\}$ is a MWNS of~$(G-Z, T)$.
            \end{restatable}
            \begin{claimproof}[Proof sketch]
            The proof consists of a delicate argument which essentially says that this situation only happens when~$d$ is a block whose parent is a terminal, and there is a cutvertex~$\hat{c}$ close to block~$d$ in the block-cut forest which combines with~$Z$ to break all paths in~$G-\{x\}$ from~$N_G(x) \cap V_G(\mathcal{F}_d)$ to vertices of~$N_G(x) \setminus V_G(\mathcal{F}_d)$, and therefore breaks all $T$-cycles intersecting~$V_G(\mathcal{F}_d)$.             \end{claimproof}
            
           The proof of Lemma~\ref{Lemma:OPTDecreasesByAFactorOf14} follows from the preceding statements by formula manipulation: whenever the approximation algorithm chooses a set~$Z$, an optimal solution chooses at least~$|Z|/14$ vertices from~$V_G(\mathcal{F}_d)$.  \rearng{The remainder of the proof is given in Appendix~\ref{subsection:RemainingProof:Lemma:OPTDecreasesByAFactorOf14}.} \sk{The remainder of the proof is given in the full version.} 
            \end{claimproof}

Using Lemma~\ref{Lemma:OPTDecreasesByAFactorOf14}, an easy induction shows that the algorithm indeed computes a 14-approximation. As each iteration can be implemented in polynomial time, this leads to a proof of Theorem~\ref{theorem:blocker}. \rearng{The details are given in Appendix~\ref{subsection:ProofOfBlocker}.} \sk{The details are given in the full version.}

\section{Bounding the number of terminals}
\label{sec:bounding_terminals}
Our main goal in this section is to prove Theorem~\ref{theorem:terminalbounding}. Intuitively, there are two distinct ways in which a \cct admitting a small solution can contain a large number of terminals: either there can be a star-like structure of blocks containing a terminal joined at a single cutvertex, or there exists a long path in the block-cut tree containing many blocks with a terminal. After applying reduction rules to attack both kinds of situations, we give a final reduction rule to bound the number of relevant connected components that contain a terminal, which will lead to the desired bound on the number of terminals. \rearng{Due to space constraints, the safeness of the reduction rules we present here is deferred to Appendix~\ref{section:SafenessOfRRs}.} 
\sk{Due to space constraints, the safeness of the reduction rules we present here is deferred to the full version.}

\begin{reductionrule}\label{RR:RemoveWeaklySeparatedTerminal}
    Let~\gtk be an instance of~\mwnsshort, and let~$t \in T$ be a terminal such that for any other terminal~$t' \in \tminust$, there do not exist~2 \ivd~$t$-$t'$ paths in~$G$, i.e., the terminal~$t$ is nearly-separated. Then remove~$t$ from the set~$T$. The new instance is~\gtminustk.
\end{reductionrule}

Next, we have the following general reduction rule. 

\begin{reductionrule} 
\label{RR:BoundingBlocksOfDegree2Path}
    Let~\gtk be an instance of \mwnsshort. Suppose there exist 2 non-terminal vertices~$x, y \in V(G) \setminus T$ such that~$G-\{x, y\}$ has a \cct~$D$ satisfying the following conditions:
    \begin{enumerate}[(a)]
        \item $|V(D) \cap T| \geq 3$,

        \item the graph~$G[V(D) \cup \{x, y\}]$ has no~$T$-cycle, and \label{RR:BoundingBlocksOfDegree2Path:NoT-Cycle}
        
        \item there is an~$x$-$y$ path in~$G[V(D) \cup \{x, y\}]$ containing distinct terminals~$t_1, t_2 \in T$. \label{RR:BoundingBlocksOfDegree2Path:KeepingAPathP}
    \end{enumerate}
    Then turn any terminal of~$D$ which is not~$t_1, t_2$ into a non-terminal. Formally, let~$t \in V(D) \cap (T \setminus \{t_1, t_2\})$, then~\gtminustk is a new instance of~\mwnsshort. 
\end{reductionrule}

        Next, we show that given an instance~\gtk of~\mwnsshort and a 1-redundant MWNS~$S^* \subseteq V(G) \setminus T$ of~$(G, T)$, in polynomial-time we can obtain an equivalent instance~$(G, T', k)$ such that the number of \ccts of~$G-S^*$ which contain at least one terminal from~$T'$ is bounded by~$\Oh(|S^*|^2 \cdot k)$. Towards this, we apply the following marking scheme with respect to the set~$S^*$ to mark~$\Oh (|S^*|^2 \cdot k)$ \ccts of~$G - S^*$.

        \begin{markingscheme}
        \label{MarkingScheme:Bounding_Interesting_CCTs}
        Let~\gtk be an instance of \mwnsshort, and let~$S^* \subseteq V(G)$ be a 1-redundant MWNS of~$(G, T)$. For each pair of distinct vertices~$x, y \in S^*$, we greedily mark~$(k+2)$ \ccts~$C^{x, y}_{i_1}, \ldots, C^{x, y}_{i_{k+2}}$ of~$G-S^*$ such that for each~$m \in [k+2]$, the \cct~$C^{x, y}_{i_m}$ contains two vertices (not necessarily distinct) $p_m \in N_G(x)$ and~$q_m \in N_G(y)$ such that there is a~$p_m$-$q_m$ path~$P_m$ (possibly of length 0) inside the \cct~$C^{x, y}_m$ which contains at least one terminal, i.e., we have~$|V(P_m) \cap T| \geq 1$. If there are fewer than~$k+2$ such components, we simply mark all of them.
        \end{markingscheme}
    
        We have the following reduction rule based on the above marking scheme. It requires a 1-redundant MWNS to be known. In the context of Theorem~\ref{theorem:terminalbounding}, where we are given a MWNS~$S$ of~$(G,T,k)$, we can exploit Theorem~\ref{theorem:blocker} to obtain a 1-redundant MWNS of size~$\Oh(|S| \cdot k)$ in polynomial time: for each~$x \in S$, if an $x$-avoiding MWNS~$S_x$ in~$G - (S \setminus \{x\})$ has size more than~$14k$ then all solutions of~$(G,T,k)$ contain~$x$ and we may safely remove~$x$ and decrease~$k$; otherwise, we can add~$S_x$ to~$S$ to ensure all $T$-cycles intersecting~$x$ are hit twice, without blowing up the size of~$S$ too much. \rearng{The details are given in Appendix~\ref{subsection:CostructionOfRedundantSet}.} \sk{The details are given in the full version.}
        \begin{reductionrule}
        \label{RR:Bounding_CCTs}
            Let~\gtk be an instance of~\mwnsshort, and let~$S^* \subseteq V(G) \setminus T$ be a 1-redundant MWNS of~$(G, T)$. Let~$\calC_{M}$ be the set of \ccts of~$G-S^*$ marked by Marking Scheme~\ref{MarkingScheme:Bounding_Interesting_CCTs} with respect to the set~$S^*$. Let~$T' := T \cap (\bigcup_{C \in \calC_{M}} V(C))$. 
            If~$T \setminus T' \neq \emptyset$ then we convert the terminals of~$T \setminus T'$ to non-terminals. The new instance is~$(G, T', k)$.
        \end{reductionrule}

A delicate analysis shows that applying these reduction rules indeed leads to an equivalent instance with the desired bound on the number of terminals, which leads to a proof of  Theorem~\ref{theorem:terminalbounding}. \rearng{The details are given in Appendix~\ref{sec:ProofOfTerminalBoundingThm}.} \sk{The details are given in the full version.}

\section{Conclusions}

In this paper we initiated the study of the \mwns problem, a generalization of \textsc{Multiway Cut} focused on reducing the connectivity between each pair of terminals. We developed reduction rules to reduce the number of terminals in an instance, aided by a constant-factor approximation algorithm for the problem of finding a near-separator avoiding a given vertex~$x$ in an instance for which~$\{x\}$ is a near-separator. Our work leads to several follow-up questions. First of all, one could consider extending the notion of near-separation to allow for larger (but bounded) connectivity between terminal pairs in the resulting instance, for example by requiring that in the graph~$G-S$, for each pair of distinct terminals~$t_i, t_j$ there is a vertex set of size at most~$c$ whose removal separates~$t_i$ and~$t_j$. The setting we considered here is that of~$c=1$, which allows us to understand the structure of the problem based on the block-cut forest of~$G-S$. For~$c=2$ it may be feasible to do an analysis based on the decomposition of~$G-S$ into triconnected components, but for larger values of~$c$ the structure may become significantly more complicated. One could also consider a generic setting (see \cite{BarmanC10}) where for each pair of terminals~$t_i, t_j$, some threshold~$f(t_i, t_j) = f(t_j, t_i)$ is specified such that in~$G-S$ there should be a set of at most~$f(t_i, t_j)$ nonterminals whose removal separates~$t_i$ from~$t_j$. Note that if the values of~$f$ are allowed to be arbitrarily large, this generalizes \textsc{Multicut}. Is the resulting problem fixed-parameter tractable parameterized by~$k$?

Another direction for future work lies in the development of a polynomial kernel. For the \textsc{Multiway Cut} problem with delectable terminals, as well as the setting with undeletable but constantly many terminals, a polynomial kernel is known based on matroid techniques~\cite{KratschW20}. Does the variant of \textsc{Multiway Near-Separator} with deletable terminals admit a polynomial kernel?

\bibliography{bibliography.bib}

\begin{thebibliography}{10}

\bibitem{BarmanC10}
Siddharth Barman and Shuchi Chawla.
\newblock Region growing for multi-route cuts.
\newblock In Moses Charikar, editor, {\em Proceedings of the Twenty-First
  Annual {ACM-SIAM} Symposium on Discrete Algorithms, {SODA} 2010, Austin,
  Texas, USA, January 17-19, 2010}, pages 404--418. {SIAM}, 2010.
\newblock \href {https://doi.org/10.1137/1.9781611973075.34}
  {\path{doi:10.1137/1.9781611973075.34}}.

\bibitem{BlazejCKKSV21}
V{\'{a}}clav Blazej, Pratibha Choudhary, Dusan Knop, Jan~Maty{\'{a}}s Kristan,
  Ondrej Such{\'{y}}, and Tom{\'{a}}s Valla.
\newblock Constant factor approximation for tracking paths and fault tolerant
  feedback vertex set.
\newblock In Jochen K{\"{o}}nemann and Britta Peis, editors, {\em Approximation
  and Online Algorithms - 19th International Workshop, {WAOA} 2021, Lisbon,
  Portugal, September 6-10, 2021, Revised Selected Papers}, volume 12982 of
  {\em Lecture Notes in Computer Science}, pages 23--38. Springer, 2021.
\newblock \href {https://doi.org/10.1007/978-3-030-92702-8\_2}
  {\path{doi:10.1007/978-3-030-92702-8\_2}}.

\bibitem{BousquetDT18}
Nicolas Bousquet, Jean Daligault, and St{\'{e}}phan Thomass{\'{e}}.
\newblock Multicut is {FPT}.
\newblock {\em {SIAM} J. Comput.}, 47(1):166--207, 2018.
\newblock \href {https://doi.org/10.1137/140961808}
  {\path{doi:10.1137/140961808}}.

\bibitem{BringmannHML16}
Karl Bringmann, Danny Hermelin, Matthias Mnich, and Erik~Jan van Leeuwen.
\newblock Parameterized complexity dichotomy for steiner multicut.
\newblock {\em J. Comput. Syst. Sci.}, 82(6):1020--1043, 2016.
\newblock \href {https://doi.org/10.1016/j.jcss.2016.03.003}
  {\path{doi:10.1016/j.jcss.2016.03.003}}.

\bibitem{ChenLL09}
Jianer Chen, Yang Liu, and Songjian Lu.
\newblock An improved parameterized algorithm for the minimum node multiway cut
  problem.
\newblock {\em Algorithmica}, 55(1):1--13, 2009.
\newblock \href {https://doi.org/10.1007/s00453-007-9130-6}
  {\path{doi:10.1007/s00453-007-9130-6}}.

\bibitem{ChitnisCHPP16}
Rajesh Chitnis, Marek Cygan, MohammadTaghi Hajiaghayi, Marcin Pilipczuk, and
  Michal Pilipczuk.
\newblock Designing {FPT} algorithms for cut problems using randomized
  contractions.
\newblock {\em {SIAM} J. Comput.}, 45(4):1171--1229, 2016.
\newblock \href {https://doi.org/10.1137/15M1032077}
  {\path{doi:10.1137/15M1032077}}.

\bibitem{CyganFKLMPPS15}
Marek Cygan, Fedor~V. Fomin, Lukasz Kowalik, Daniel Lokshtanov, D{\'{a}}niel
  Marx, Marcin Pilipczuk, Michal Pilipczuk, and Saket Saurabh.
\newblock {\em Parameterized Algorithms}.
\newblock Springer, 2015.
\newblock \href {https://doi.org/10.1007/978-3-319-21275-3}
  {\path{doi:10.1007/978-3-319-21275-3}}.

\bibitem{CyganPPW13a}
Marek Cygan, Marcin Pilipczuk, Michal Pilipczuk, and Jakub~Onufry Wojtaszczyk.
\newblock On multiway cut parameterized above lower bounds.
\newblock {\em {ACM} Trans. Comput. Theory}, 5(1):3:1--3:11, 2013.
\newblock \href {https://doi.org/10.1145/2462896.2462899}
  {\path{doi:10.1145/2462896.2462899}}.

\bibitem{CyganPPW13}
Marek Cygan, Marcin Pilipczuk, Michal Pilipczuk, and Jakub~Onufry Wojtaszczyk.
\newblock Subset feedback vertex set is fixed-parameter tractable.
\newblock {\em {SIAM} J. Discret. Math.}, 27(1):290--309, 2013.
\newblock \href {https://doi.org/10.1137/110843071}
  {\path{doi:10.1137/110843071}}.

\bibitem{DahlhausJPSY92}
E.~Dahlhaus, D.~S. Johnson, C.~H. Papadimitriou, P.~D. Seymour, and
  M.~Yannakakis.
\newblock The complexity of multiway cuts (extended abstract).
\newblock In {\em Proceedings of the Twenty-Fourth Annual ACM Symposium on
  Theory of Computing}, STOC '92, page 241–251, New York, NY, USA, 1992.
  Association for Computing Machinery.
\newblock \href {https://doi.org/10.1145/129712.129736}
  {\path{doi:10.1145/129712.129736}}.

\bibitem{Diestel17}
Reinhard Diestel.
\newblock {\em Graph Theory, 5th Edition}, volume 173 of {\em Graduate texts in
  mathematics}.
\newblock Springer, 2017.
\newblock \href {https://doi.org/10.1007/978-3-662-53622-3}
  {\path{doi:10.1007/978-3-662-53622-3}}.

\bibitem{EdmondsK72}
Jack Edmonds and Richard~M. Karp.
\newblock Theoretical improvements in algorithmic efficiency for network flow
  problems.
\newblock {\em J. ACM}, 19(2):248–264, apr 1972.
\newblock \href {https://doi.org/10.1145/321694.321699}
  {\path{doi:10.1145/321694.321699}}.

\bibitem{FordF56}
L.~R. Ford and D.~R. Fulkerson.
\newblock Maximal flow through a network.
\newblock {\em Canadian Journal of Mathematics}, 8:399–404, 1956.
\newblock \href {https://doi.org/10.4153/CJM-1956-045-5}
  {\path{doi:10.4153/CJM-1956-045-5}}.

\bibitem{GolovachT19}
Petr~A. Golovach and Dimitrios~M. Thilikos.
\newblock Clustering to given connectivities.
\newblock In Bart M.~P. Jansen and Jan~Arne Telle, editors, {\em 14th
  International Symposium on Parameterized and Exact Computation, {IPEC} 2019,
  September 11-13, 2019, Munich, Germany}, volume 148 of {\em LIPIcs}, pages
  18:1--18:17. Schloss Dagstuhl - Leibniz-Zentrum f{\"{u}}r Informatik, 2019.
\newblock \href {https://doi.org/10.4230/LIPIcs.IPEC.2019.18}
  {\path{doi:10.4230/LIPIcs.IPEC.2019.18}}.

\bibitem{Guillemot11a}
Sylvain Guillemot.
\newblock {FPT} algorithms for path-transversal and cycle-transversal problems.
\newblock {\em Discret. Optim.}, 8(1):61--71, 2011.
\newblock \href {https://doi.org/10.1016/j.disopt.2010.05.003}
  {\path{doi:10.1016/j.disopt.2010.05.003}}.

\bibitem{JansenP18}
Bart M.~P. Jansen and Marcin Pilipczuk.
\newblock Approximation and kernelization for chordal vertex deletion.
\newblock {\em {SIAM} J. Discret. Math.}, 32(3):2258--2301, 2018.
\newblock \href {https://doi.org/10.1137/17M112035X}
  {\path{doi:10.1137/17M112035X}}.

\bibitem{KratschW20}
Stefan Kratsch and Magnus Wahlstr{\"{o}}m.
\newblock Representative sets and irrelevant vertices: New tools for
  kernelization.
\newblock {\em J. {ACM}}, 67(3):16:1--16:50, 2020.
\newblock \href {https://doi.org/10.1145/3390887} {\path{doi:10.1145/3390887}}.

\bibitem{LokshtanovR0Z18}
Daniel Lokshtanov, M.~S. Ramanujan, Saket Saurabh, and Meirav Zehavi.
\newblock Reducing {CMSO} model checking to highly connected graphs.
\newblock In Ioannis Chatzigiannakis, Christos Kaklamanis, D{\'{a}}niel Marx,
  and Donald Sannella, editors, {\em 45th International Colloquium on Automata,
  Languages, and Programming, {ICALP} 2018, July 9-13, 2018, Prague, Czech
  Republic}, volume 107 of {\em LIPIcs}, pages 135:1--135:14. Schloss Dagstuhl
  - Leibniz-Zentrum f{\"{u}}r Informatik, 2018.
\newblock \href {https://doi.org/10.4230/LIPIcs.ICALP.2018.135}
  {\path{doi:10.4230/LIPIcs.ICALP.2018.135}}.

\bibitem{Marx06}
D{\'{a}}niel Marx.
\newblock Parameterized graph separation problems.
\newblock {\em Theor. Comput. Sci.}, 351(3):394--406, 2006.
\newblock \href {https://doi.org/10.1016/j.tcs.2005.10.007}
  {\path{doi:10.1016/j.tcs.2005.10.007}}.

\bibitem{MarxOR13}
D{\'{a}}niel Marx, Barry O'Sullivan, and Igor Razgon.
\newblock Finding small separators in linear time via treewidth reduction.
\newblock {\em {ACM} Trans. Algorithms}, 9(4):30:1--30:35, 2013.
\newblock \href {https://doi.org/10.1145/2500119} {\path{doi:10.1145/2500119}}.

\bibitem{MarxR14}
D{\'{a}}niel Marx and Igor Razgon.
\newblock Fixed-parameter tractability of multicut parameterized by the size of
  the cutset.
\newblock {\em {SIAM} J. Comput.}, 43(2):355--388, 2014.
\newblock \href {https://doi.org/10.1137/110855247}
  {\path{doi:10.1137/110855247}}.

\bibitem{Razgon11}
Igor Razgon.
\newblock Large isolating cuts shrink the multiway cut.
\newblock {\em CoRR}, abs/1104.5361, 2011.
\newblock URL: \url{http://arxiv.org/abs/1104.5361}, \href
  {http://arxiv.org/abs/1104.5361} {\path{arXiv:1104.5361}}.

\bibitem{ReedSV04}
Bruce~A. Reed, Kaleigh Smith, and Adrian Vetta.
\newblock Finding odd cycle transversals.
\newblock {\em Oper. Res. Lett.}, 32(4):299--301, 2004.
\newblock \href {https://doi.org/10.1016/j.orl.2003.10.009}
  {\path{doi:10.1016/j.orl.2003.10.009}}.

\bibitem{Schrijver03}
A.~Schrijver.
\newblock {\em Combinatorial Optimization - Polyhedra and Efficiency}.
\newblock Springer, 2003.

\bibitem{Thomasse10}
St{\'{e}}phan Thomass{\'{e}}.
\newblock A 4\emph{k}\({}^{\mbox{2}}\) kernel for feedback vertex set.
\newblock {\em {ACM} Trans. Algorithms}, 6(2):32:1--32:8, 2010.
\newblock \href {https://doi.org/10.1145/1721837.1721848}
  {\path{doi:10.1145/1721837.1721848}}.

\bibitem{Xiao10}
Mingyu Xiao.
\newblock Simple and improved parameterized algorithms for multiterminal cuts.
\newblock {\em Theory Comput. Syst.}, 46(4):723--736, 2010.
\newblock \href {https://doi.org/10.1007/s00224-009-9215-5}
  {\path{doi:10.1007/s00224-009-9215-5}}.

\end{thebibliography}

\appendix

\clearpage

\section{Additional preliminaries}\label{sec:AdditionalPreliminaries}
\subparagraph{Graphs.}\label{Subparagraph:AdditionalPrelimsAboutGraphs}
An~$x$-$y$ path~$P$ and a $q$-$r$ path~$Q$ are said to be \ivd if they do not share a common internal-vertex, i.e., we have~$(V(P) \setminus \{x, y\}) \cap (V(Q) \setminus \{q, r\}) = \emptyset$. Given~$X, Y \subseteq V(G)$, a path~$(v_1, \ldots, v_k)$ in~$G$ is called an~$X-Y$ path if~$v_1 \in X$ and~$v_k \in Y$.

\begin{theorem}[Menger's theorem] {\cite[Theorem 3.3.1]{Diestel17}}
    Let~$G$ be a graph and~$X, Y \subseteq V(G)$ be subsets of vertices such that~$X \cap Y = \emptyset$ and there does not exist an edge~$\{x, y\} \in E(G)$ for any~$x \in X$ and~$y \in Y$. Then the minimum number of vertices separating~$X$ from~$Y$ is equal to the maximum number of vertex-disjoint~$X-Y$ paths in~$G$.
\end{theorem}
A cutvertex in a graph is a vertex~$v$ whose removal increases the number of connected components. A graph is 2-connected if it has at least three vertices and does not contain any cutvertex. A \emph{block} of a graph~$G$ is a maximal connected subgraph~$B$ of~$G$ such that~$B$ does not have a cutvertex. Each block of~$G$ is either a 2-connected subgraph of~$G$, a single edge, or an isolated vertex.

\begin{definition}[Block-cut graph] {\cite[\S 3.1]{Diestel17}}
\label{definition:block-cut_graph}	 
	Given a graph~$G$, let~$A$ be the set of cutvertices of~$G$, and let~$\mathcal{B}$ be the set of its blocks. The \emph{block-cut graph}~$G'$ of~$G$ is the bipartite graph with partite sets~$A$ and~$\mathcal{B}$, and for each cut-vertex~$a \in A$, for each block~$B \in \mathcal{B}$, there is an edge~$\{a, B\} \in E(G')$ if~$a \in V(B)$.
\end{definition}
We also need the following simple but useful properties of block-cut graphs. 
\begin{lemma}{\cite[Lemma 3.1.4]{Diestel17}}
\label{lemma:BlockGraphOfCCT}
The block-cut graph of a connected graph is a tree.
\end{lemma}

\begin{observation}
\label{observation:propertyofacutvertex}
Consider an edge~$e$ of the block-cut tree~$\calT$ of a connected graph~$G$, let~$v$ be the unique cutvertex incident on~$e$, let~$\calT_1, \calT_2$ be the two trees of~$\calT - \{e\}$, and let~$Y_i := V_G (\calT_i)$ be the vertices of~$G$ occurring in blocks of~$\calT_i$, for~$i \in [2]$. Then all paths from a vertex of~$Y_1 \setminus \{v\}$ to a vertex of~$Y_2 \setminus \{v\}$ in~$G$ intersect~$v$.
\end{observation}

\begin{observation}
\label{observation:Separatingcutvertex}
Let~$G$ be a graph and~$T \subseteq V(G)$ be a set of terminals such that~$V(G) \setminus T \neq \emptyset$. If no block of~$G$ contains two or more terminals, then for each pair~$t_i, t_j \in T$ of distinct terminals there exists a vertex~$v \in V(G) \setminus T$ such that~$v_i$ and $v_j$ belong to different connected components of~$G-\{v\}$.
\end{observation}

\begin{observation} 
\label{Observation:Any_MWNS_contains_at_least_one_vertex_from_each_T-disjoint_cycle}
For any positive integer~$\ell$, if there are~$\ell$ cycles on~$x$ which are $(T \cup \{x\})$-disjoint in graph~$G$, then any $x$-avoiding~MWNS of~$(G, T)$ contains at least~$\ell$ vertices: at least one distinct non-terminal from each cycle.
\end{observation}

\begin{proposition}
\label{Proposition:ThereAre2PathsInsideBlock}
        Let~$G$, be a graph,~$B$ a block in~$G$, and consider three distinct vertices~$p, q, t \in V(B)$. There is a~$p$-$t$ path~$P$ and~$q$-$t$ path~$Q$ inside block~$B$ such that~$V(P) \cap V(Q) = \{t\}$.
\end{proposition}
\begin{claimproof}
        Since~$B$ is a block containing at least three vertices, it is a 2-connected graph. We first add a new vertex~$v_{p, q}$ and edges~$\{v_{pq}, p\}$, $\{v_{pq}, q\}$ to block~$B$ to obtain a new graph~$B'$. Observe that~$B'$ is still a 2-connected graph, as the above modification can be seen as adding a new~$B$-path~$v, v_{pq}, q$ to~$B$ and adding a~$B$-path to a 2-connected graph results in a 2-connected graph (see~\cite[Proposition 3.1.1]{Diestel17}).
        Next, we apply Menger's theorem with~$X = \{v_{pq}\}$ and~$Y = \{t\}$ in the (2-connected) graph~$B'$, which ensures that there are 2 \ivd $v_{pq}$-$t$ paths~$P'$ and~$Q'$ in~$B'$.
        By construction of~$B'$ and the fact that~$P'$ and~$Q'$ are \ivd $v_{pq}$-$t$ paths in~$B'$, we know that one of the~$v_{pq}$-$t$ paths (assume w.l.o.g.~$P'$) contains vertex~$p$ of block~$B$, whereas the other~$v_{pq}$-$t$ path~$Q'$ contains vertex~$q$ of block~$B$. 
        Hence the paths~$P = P'[p, t]$ and~$Q = Q'[q, t]$ are the desired paths with~$V(P) \cap V(Q) = \{t\}$.
\end{claimproof}

This following observation follows from Proposition~\ref{Proposition:ThereAre2PathsInsideBlock}, as explained below.

\begin{observation}
\label{Observation:Degree2Path}
    Let~$\calF$ be the block-cut forest of a graph~$G$. Suppose there is an~$x$-$y$ path~$\calP$ in~$\calF$ between cutvertices~$x, y$ of~$\calF$. Then for any distinct vertices~$v_1, v_2 \in V(G)$ from 2 distinct blocks on~$\calP$, the graph~$G$ has an~$x$-$y$ path~$P_{12}$ through~$v_1, v_2$, i.e., $v_1, v_2 \in V(P_{12})$.
\end{observation}

We can construct the desired path~$P_{12}$ by concatenating paths inside each block~$B$ on the $x$-$y$ path in~$\mathcal{F}$, where each path connects the cutvertex~$p$ connecting to the previous block, with the cutvertex~$q$ connecting through the next block. If a forced~$v_i$ is chosen from block~$B$, we can use Proposition~\ref{Proposition:ThereAre2PathsInsideBlock} to obtain a~$p$-$q$ path through~$t = v_i$.

\subparagraph{Parameterized algorithms}.\label{Subparagraph:AdditionalPrelimsAboutPC} 
A \emph{parameterized problem}~$L$ is a subset of $\Sigma^* \times \mathbb{N}$, where $\Sigma$ is a finite alphabet. 
A parameterized problem~$L \subseteq \Sigma^* \times \mathbb{N}$ is called \emph{fixed parameter tractable} (FPT) if there exists an algorithm which for every input~$(x, k) \in \Sigma^* \times \mathbb{N}$ correctly decides whether~$(x, k) \in L$ in~$f(k) \cdot |x|^{\Oh (1)}$ time, where~$f \colon \mathbb{N} \to \mathbb{N}$ is a computable function. We refer to~\cite{CyganFKLMPPS15} for more background on parameterized algorithms. 

Given a graph~$G$ and~$X, Y \subseteq V(G)$, a set~$S \subseteq V(G)$ is called an \emph{$(X, Y)$-separator} if there is no $x$-$y$ path in~$G-S$ for any~$x \in X \setminus S$ and~$y \in Y \setminus S$. The notion of important separator was defined by Marx~\cite{Marx06} to obtain an FPT algorithm for \textsc{Multiway Cut}. He also derived several useful properties.
\begin{definition}[Important separator]{\cite[ Definition 8.49]{CyganFKLMPPS15}}
Let~$G$ be an undirected graph, let~$X, Y \subseteq V(G)$ be two sets of vertices, and let~$V^\infty \subseteq V(G)$ be a set of undeletable vertices. Let~$S \subseteq V(G) \setminus V^\infty (G)$ be an~$(X, Y)$-separator and let~$R$ be the set of vertices reachable from~$X \setminus S$ in~$G - S$.
We say that~$S$ is an \emph{important~$(X, Y)$-separator} if it is inclusion-wise \emph{minimal} and there is no~$(X, Y)$-separator~$S' \subseteq V(G) \setminus V^\infty(G)$ with~$|S'| \leq |S|$ such that~$R \subset	R'$, where~$R'$ is the set of vertices reachable from~$X \setminus S'$ in~$G - S'$.
\end{definition}

\begin{proposition}{\cite[Proposition 8.50]{CyganFKLMPPS15}}
\label{proposition:existenceofimpsep}
Let~$G$ be an undirected graph and~$X, Y \subseteq V(G)$ be two sets of vertices, and let~$V^\infty \subseteq V(G)$ be a set of undeletable vertices. Let~$\hat{S} \subseteq V(G) \setminus V^\infty(G)$ be an~$(X, Y)$-separator and let~$\hat{R}$ be the set of vertices reachable from~$X \setminus \hat{S}$ in~$G-\hat{S}$. Then there is an important~$(X, Y)$-separator~$S' = N_G(R') \subseteq  V(G) \setminus V^\infty(G)$ such that~$|S'| \leq |\hat{S}|$ and~$\hat{R} \subseteq R'$. 
\end{proposition}

\begin{theorem}{\cite[Theorem 8.51]{CyganFKLMPPS15}}
\label{Theorem:Algorithm_To_Generate_Imp_Sep}
    Let~$X, Y \subseteq V (G)$ be two sets of vertices in an undirected graph~$G$, let~$k \geq 0$ be an integer, and let~$\mathcal{S}_k$ be the set of all $(X, Y )$-important separators of size at most~$k$. Then~$|\mathcal{S}_k| \leq  4^k$ and~$\mathcal{S}_k$ can be constructed in time~$\Oh(|\mathcal{S}_k| \cdot k^2 \cdot (n + m))$.
\end{theorem}

\section{Additional proofs}
\subsection{Hardness proof for \mwns}

\begin{lemma} \label{Lemma:NP-hardness_proof_of_MWNS}
    \mwns is~\nphard.
\end{lemma}
\begin{proof}
    We give a polynomial-time reduction from~\mws to~\mwnsshort. The~\mws problem is formally defined as follows.
    
    \defparprob{\mws(\mwsshort)}{An undirected graph~$G$, terminal set~$T \subseteq V(G)$, and a positive integer~$k$.}{$k$}{Is there a set~$S \subseteq V(G) \setminus T$ with~$|S| \leq k$ such that there does not exist a pair of distinct terminals~$t_i, t_j \in T$ for which there is a~$t_i$-$t_j$~path in~$G-S$?}
    
    \mws is~\nphard \cite{DahlhausJPSY92}.
    Given an instance~\gtk of \mwsshort we describe the construction of an instance~$(G', T, k)$ of \mwnsshort; it will be easy to see that it can be carried out in polynomial time. Consider an arbitrary ordering~$t_1, t_2, \ldots, t_{|T|}$ of the set~$T$. We construct the graph~$G'$ such that~$(G', T, k)$ is a YES-instance of~\mwnsshort if and only if~\gtk is a YES-instance of~\mwsshort.

    We begin with~$G' := G$. Then for each~$i \in [|T|-1]$, we create a vertex~$w_i$ in~$G'$, and insert edges~$\{t_i, w_i\}$ and~$\{w_i, t_{i+1}\}$ in~$G'$. This completes the construction of~$G'$. 
    
    Next, we prove that~\gtk is a YES-instance of~\mwsshort if and only if~$(G', T, k)$ is a YES-instance of~\mwnsshort.
    In the forward direction, assume that~\gtk is  a YES-instance of~\mwsshort, and let~$S \subseteq V(G) \setminus T$ be a \mwsshort of~\gtk. Note that by definition of~\mwsshort, for each pair of distinct terminals~$t_i, t_j \in T$, there is no~$t_i$-$t_j$ path in~$G-S$. Since the transformation into~$G'$ consists of adding degree-2 vertices connecting consecutive terminals, any pair of distinct terminals~$t_i, t_j \in T$ with~$i<j$ can be separated in~$G'-S$ by removing the vertex~$w_i$. Hence the same set~$S \subseteq V(G') \setminus T$ is also a \mwnsshort of~$(G', T, k)$.

    In the reverse direction, assume that~$(G', T, k)$ is a YES-instance of \mwnsshort and let~$S' \subseteq V(G') \setminus T$ be a solution. Let~$W := \bigcup_{i \in [|T|-1]} \{w_i\}$ be the set of newly added vertices in~$G'$. In the case when~$S' \cap W = \emptyset$, it is easy to see that the same set~$S'$ is also a \mwsshort of~\gtk: because if~$S$ is not a \mwsshort of~\gtk then it implies that there is a pair of distinct terminals~$t_i, t_j \in T$ connected by a~$t_i$-$t_j$ path say~$P$ in~$G-S$. By construction of~$G'$, there is also a unique $t_i$-$t_j$ path say~$P'$ in~$G' [W \cup T]$. Note that the paths~$P$ and~$P'$ are~$T$-disjoint, and neither~$P$ nor~$P'$ intersects the set~$S'$, a contradiction to the fact that~$S'$ is a solution of~$(G', T, k)$ since the paths witness that~$t_i, t_j$ cannot be separated by removing a single non-terminal.

    Hence we need to show how to prove the case when~$S' \cap W \neq \emptyset$. In this case, note that due to Lemma~\ref{lemma:easilyseparableterminal}, there is a terminal~$t \in T$ and a non-terminal~$x \in V(G') \setminus T$ such that the set~$S' \cup \{x\}$ is a~\ttsep. So by applying Lemma~\ref{lemma:easilyseparableterminal} at most~$|T|-1$ times we obtain a set~$S^* \subseteq V(G') \setminus T$ such that the set~$S^*$ is a~\mwsshort of~$(G', T)$. Moreover, note that~$|S^*| \leq |S'| + (|T|-1)$ because in each step of Lemma~\ref{lemma:easilyseparableterminal} we add at most 1 additional vertex to separate a terminal. Next, we obtain a new set~$\hat{S} := S^* \setminus W$. Now, it is easy to observe that the set~$\hat{S} \subseteq V(G') \setminus T$ is a solution of~$(G', T, k)$ such that~$\hat{S} \cap W = \emptyset$ thus (like in the previous case) we have the property that the set~$\hat{S}$ is also a \mwsshort of~\gtk. Hence, \gtk is a YES-instance of~\mwsshort.
\end{proof}

\subsection{Proof of Proposition~\ref{Proposition:DifferentWaysOfLookingAtMWNS}}
We repeat the statement of the proposition for convenience.
\DifferentWaysOfLookingAtMWNSState*

\begin{claimproof}
    We prove that~$\ref{Condition1:DifferentWaysOfLookingAtMWNS} \Rightarrow \ref{Condition2:DifferentWaysOfLookingAtMWNS} \Rightarrow \ref{Condition3:DifferentWaysOfLookingAtMWNS} \Rightarrow \ref{Condition1:DifferentWaysOfLookingAtMWNS}$, thereby finishing the proof. 

    $(\ref{Condition1:DifferentWaysOfLookingAtMWNS} \Rightarrow \ref{Condition2:DifferentWaysOfLookingAtMWNS}):$
    Consider a pair of distinct termianls~$t_i, t_j \in T$. Since~$S \subseteq V(G) \setminus T$, we have~$t_i, t_j \in V(G-S)$. If~$t_i, t_j$ already belong to different \ccts of~$G-S$ then for each vertex~$v \in S$, Condition~\ref{Condition2:DifferentWaysOfLookingAtMWNS} holds.

    Therefore assume that both~$t_i, t_j$ belong to the same \cct~$C$ of~$G-S$. Let~$B_i, B_j$ be arbitrary blocks of~$C$ containing terminal~$t_i, t_j$, respectively. First we show that~$B_i \neq B_j$. Towards this assume for a contradiction that~$B_i = B_j = B$. Then block~$B$ of~$G-S$ contains distinct terminals~$t_i, t_j \in T$. If~$\{t_i, t_j\} \in E(B)$ then we have \ivd $t_i$-$t_j$ paths~$P_1 = P_2 = \{t_i, t_j\}$, a contradiction to Condition~\ref{Condition1:DifferentWaysOfLookingAtMWNS}.
    Whereas if~$\{t_i, t_j\} \notin E(B)$, then we apply Menger's theorem with~$X= \{t_i\}$ and~$Y = \{t_j\}$ in 2-connected graph~$B$, which ensures that there are \ivd $t_i$-$t_j$ paths~$P_1, P_2$ inside block~$B$ of~$G-S$, a contradiction to Condition~\ref{Condition1:DifferentWaysOfLookingAtMWNS}.
    By the same argument, each block of~$C$ contains at most 1 terminal. Hence we have~$V(C) \setminus T \neq \emptyset$. Then by Observation~\ref{observation:Separatingcutvertex}, there is a vertex~$v \in V(C) \setminus T$ such that~$t_i, t_j$ belong to different \ccts of~$C-\{v\}$. Thus for the same vertex~$v$, it holds that~$t_i, t_j$ belong to different \ccts of~$G-(S \cup \{v\})$.

    $(\ref{Condition2:DifferentWaysOfLookingAtMWNS} \Rightarrow \ref{Condition3:DifferentWaysOfLookingAtMWNS}):$ 
    To show this we use a proof by contradiction, i.e., assume that Condition~\ref{Condition2:DifferentWaysOfLookingAtMWNS} holds whereas Condition~\ref{Condition3:DifferentWaysOfLookingAtMWNS} does not hold.\
    This implies that either the set $T$ is not an independent set or there is a simple cycle~$C$ in~$G-S$ containing two distinct terminals $t_i, t_j \in T$.
    Note that if~$T$ is not an independent set and there is an edge~$\{t'_i, t'_j\} \in E(G)$ for distinct~$t'_i, t'_j \in T$, then for the pair~$t'_i, t'_j \in T$, there does not exist a vertex~$v \in V(G) \setminus T$ such that $t'_i$ and $t'_j$ belong to different connected components of~$G-(S \cup \{v\})$ as~$t'_i, t'_j \in V(G-S)$ and~$\{t'_i, t'_j\} \in E(G-S)$, a contradiction to the fact that Condition~\ref{Condition2:DifferentWaysOfLookingAtMWNS} holds.
    
    Whereas in the other case when there is a simple cycle~$C$ in~$G-S$ containing distinct terminals $t_i, t_j \in T$, the removal of a single vertex~$v \in V(G) \setminus T$ does not disconnect $t_i$ and~$t_j$ in~$G-S$. Hence for the pair~$t_i, t_j \in T$, there is no vertex~$v \in V(G) \setminus T$, such that~$t_i$ and~$t_j$ belong to different \ccts of~$G-(S \cup \{v\})$, a contradiction to the fact that Condition~\ref{Condition2:DifferentWaysOfLookingAtMWNS} holds. 

    $(\ref{Condition3:DifferentWaysOfLookingAtMWNS} \Rightarrow \ref{Condition1:DifferentWaysOfLookingAtMWNS}):$
    Assume for a contradiction that Condition~\ref{Condition3:DifferentWaysOfLookingAtMWNS} holds whereas Condition~\ref{Condition1:DifferentWaysOfLookingAtMWNS} does not hold.
    This implies that there exists a pair of distinct terminals~$t_i, t_j \in T$ such that there are \ivd $t_i$-$t_j$ paths~$P_1, P_2$ in $G-S$.
    If either~$P_1$ or~$P_2$ does not have internal vertices then we have~$\{t_i, t_j\} \in E(G)$, a contradiction to the fact that~$T$ is an independent set (and hence Condition~\ref{Condition3:DifferentWaysOfLookingAtMWNS}).
    Hence assume for the case when both~$P_1, P_2$ have some internal vertices. Then we can obtain a cycle~$C := P_1 \cdot P_2$ in~$G-S$ containing terminals~$t_i, t_j$, a contradiction to the fact that Condition~\ref{Condition3:DifferentWaysOfLookingAtMWNS} holds.
\end{claimproof}

\subsection{Proof of Lemma~\ref{lemma:easilyseparableterminal}}
\EasilySeparableTerminalState*
\begin{proof}
First note that if there is a~\cct of~$G-S$ containing exactly one terminal~(say~$t \in T$), i.e., if~$S$ itself is a~\ttsep for some~$t \in T$, then for each~$v \in S$, the set~$S \cup \{v\}$ is a~\ttsep. Thus the above lemma holds. Therefore assume that each \cct of~$G-S$ either contains no terminal or at least two terminals. Choose an arbitrary \cct~$C$ of~$G-S$ containing at least two terminals. Note that such a~\cct $C$ exists because~$|T| \geq 2$ (for a non-trivial instance of~\mwnsshort) and~$S \subseteq V(G) \setminus T$ alone is not a~\ttsep for any~$t \in T$. Let~$C'$ be the block-cut graph of the \cct~$C$. By Lemma~\ref{lemma:BlockGraphOfCCT}, we know that~$C'$ is a tree. 
Let~$\calT$ be a rooted tree formed by fixing an arbitrary block of~$C'$ as the root node. Next, for each terminal~$t \in V(C) \cap T$, we define the \emph{depth} of~$t$ with respect to the block-cut tree~$\calT$, denoted by~$d(t)$, as the shortest distance from the root to a block that contains~$t$. Note that~$t$ may appear in multiple blocks but its depth is uniquely defined by those blocks which contain~$t$ and are nearest to the root node in the tree~$\calT$.
Let~$t \in V(C) \cap T$ be an arbitrary terminal of the \cct~$C$ such that the depth of~$t$ is maximum, i.e., for all~$t' \in V(C) \cap T$, we have~$d(t') \leq d(t)$.
Let~$B \in \mathcal{B}$ be a block in the tree~$\calT$ that contains terminal~$t$ and has the property that among all the blocks of~$\calT$ containing~$t$, block~$B$ has the minimum depth,i.e., the distance from the root node to block~$B$ is~$d(t)$. Using the following claim we prove that~$B$ is not the root of~$\calT$.

\begin{claim}
\label{Claim:property_of_the_block_B}
	The block~$B$ defined above is not the root node of the block-cut tree~$\calT$.
\end{claim}

\begin{claimproof}
Assume for a contradiction that~$B \ni t$ is the root node of the block-cut tree~$\calT$. Let~$t' \in V(C) \cap \tminust$ be a terminal. Note that such a terminal~$t'$ exists as we have~$|V(C) \cap T| \geq 2$ by our choice of~$C$. Due to Observation~\ref{Observation:Block_contains_atmost_1-terminal}, the fact that $B$ is a block of~$G-S$ containing terminal~$t$, and $S$ is a MWNS of~$(G, T)$, we have~$t' \notin V_G(B)$. By definition of the block-cut tree and the fact that~$t' \in V(C)$, the terminal~$t'$ is present in at least one block of~$\calT$. Since~$t$ is present in the root block of~$\calT$ and the other terminal~$t'$ is present in a block other than the root of~$\calT$, by definition of the depth, we have~$d(t') > d(t)$, a contradiction to the fact that the depth of~$t$ is maximum.
\end{claimproof}

Let~$v$ be the parent of block~$B$ in~$\calT$, which exists by the preceding claim. Note that $v$ is a cutvertex and~$v \in V_G(B)$ by Definition~$\ref{definition:block-cut_graph}$ of the block-cut graph. Next, we use our specific choice of the terminal~$t$ and block~$B$ in the following two claims to show that~$v \notin T$ and~$S \cup \{v\}$ is a~\ttsep, thereby finishing the proof of Lemma~$\ref{lemma:easilyseparableterminal}$.
\begin{claim}
	\label{claim:v_is_notin_T}
	    $v \notin T$.
\end{claim}
\begin{claimproof}
First note that $v \neq t$, because if~$v=t$ then by the definition of the block-cut graph, the parent of~$v$ (which exists because~$\calT$ is rooted at a block and~$B$ is not the root) also contains ~$t$, a contradiction to our choice of~$B$.	Also, it is easy to observe that~$v \notin T \setminus \{t\}$, because if~$v \in T \setminus \{t\}$ then the block~$B$ of $G-S$ contains two distinct terminals~$v$ and~$t$, a contradiction Observation~\ref{Observation:Block_contains_atmost_1-terminal}. 
\end{claimproof}

\begin{claim}
\label{Claim:Construced_Set_is_separating-t}
	$S \cup \{v\}$ is a~\ttsep.
\end{claim}

\begin{claimproof}
Assume for a contradiction that~$S \cup \{v\}$ is not a~\ttsep. This implies that there exists~$t' \in T\setminus \{t\}$ such that there is a~$t$-$t'$ path~$P$ in~$G - (S \cup \{v\})$. Since~$t$ belongs to the \cct~$C$ of~$G-S$ and the path~$P$ is also present in~$G-S$, the terminal~$t'$ also belongs to~$C$, i.e., $t' \in V(C) \cap T$. Hence, the path~$P$ is present in the graph~$C-\{v\}$. Next, we look how the $t$-$t'$ path~$P$ appear inside the block-cut tree~$\calT$.  Towards this, let~$Y := V_G (\calT_B)$. Then by Observation~\ref{observation:propertyofacutvertex}, we know that there is no path from a vertex in~$Y\setminus \{v\}$ to a vertex in~$V(C) \setminus Y$ in the graph~$C - \{v\}$. Whereas the~$t$-$t'$ path~$P$ is present in~$C- \{v\}$, therefore either both~$t, t'$ belong to $Y \setminus \{v\}$ or both~$t, t'$ belong to the set $V(C) \setminus Y$.

First consider the case when both~$t$ and $t'$ belong to $V(C) \setminus Y$. Since~$t \in V(C) \setminus Y$, there exists a block~$B_t$ in~$\calT - \calT_B$ such that~$t \in V_G(B_t)$. Moreover, by definition of~$B$, the block~$B$ of subtree~$\calT_B$ also contains terminal~$t$. Then due to Observation~\ref{observation:property_of_blocks} there is a path (of length 2) from~$B$ to~$B_t$ in the block-cut tree~$\calT$ with~$t$ as an intermediate vertex as~$V_G(B) \cap V_G(B_t) = \{t\}$. Due to Claim~\ref{claim:v_is_notin_T} we have~$t \neq v$.
Hence for any vertex~$p \in V_G(B) \setminus \{v\}$ and~$q \in V_G(B_t) \setminus \{v\}$, there is a $p$-$q$ path~$P' := P_B[p, t] \cdot P_{B_t}[t, q]$, where~$P_B$ is a $p$-$t$ path in~$B - \{v\}$ and~$P_{B_t}$ is a $t$-$q$ path in~$B_t - \{v\}$. Note that paths~$P_B, P_{B_t}$ exist because~$B, B_t$ are 2-connected graphs. Since~$P'$ is a path from a vertex of~$Y \setminus \{v\}$ to a vertex of~$V(C) \setminus Y$ such that~$v \notin V(P)$, a contradiction to Observation~\ref{observation:propertyofacutvertex}. This concludes the proof for this case.

Now consider the other case when both~$t$ and~$t'$ belong to~$Y$. First note that the block~$B$ of $G-S$ which already contains terminal~$t$ cannot contain another terminal~$t'$ due to Observation~\ref{Observation:Block_contains_atmost_1-terminal}. Hence there exists a block~$B' \neq B$ in the subtree~$\calT_B$ with~$t' \in V_G(B')$.
Moreover, there is no block~$\hat{B}$ in $\calT-\calT_B$ with~$t' \in V_G(\hat{B})$ because if there exists such a block~$\hat{B}$ then we can use a similar argumentation as in the previous case (when both~$t, t' \in V(C) \setminus Y$) to obtain a $y$-$z$ path~$R$ for any~$y \in V_G(B') \setminus \{v\}$, and~$z \in V_G(\hat{B}) \setminus \{v\}$ such that~$v \notin V(R)$, a contradiction to Observation~\ref{observation:propertyofacutvertex}. Since all the blocks containing terminal~$t'$ lie in the subtree~$\calT_B$, $t' \notin V(B)$, and~$t \in V(B)$, it holds that~$d(t') > d(t)$, a contradiction to our choice of~$t$.

Since we show that both the cases lead to a contradiction, the set~$S \cup \{v\}$ is a~\ttsep.
\end{claimproof}
This concludes the proof of Lemma~$\ref{lemma:easilyseparableterminal}$.
\end{proof}

\subsection{Proof of Lemma~\ref{lemma:pushinglemmaformwns}}
\label{lable:RemainingProofOfPushingLemma}
\PushingLemmaForMWNSState*
\begin{proof}
We do a case distinction based on whether or not~$S$ alone is a~\ttsep for some~$t \in T$. The argumentation for the case when $S$ is a \ttsep for some~$t\in T$ is similar to the proof of the original pushing lemma for \textsc{Multiway Cut}~\cite{Marx06}~\cite[Lemma 8.18]{CyganFKLMPPS15}.

\subparagraph{1. If $S$ is a \ttsep for some~$t\in T$.}
Towards this, let~$R \subseteq V(G)$ be the set of vertices reachable from~$t$ in the graph~$G-S$ and note that~$N(R)$ is an~\ttsep.
If~$N(R) \subseteq S$ is also an \emph{important}~\ttsep then Condition~$\ref{condition:pushinglemmaformwns1}$ of the lemma already holds for~$S^* = S$ and~$S^*_t = N(R)$.
Therefore assume that~$N(R)$ is not an \emph{important}~\ttsep. We invoke Proposition~\ref{proposition:existenceofimpsep}, with~$X= \{t\}$, $Y = \tminust$, $V^\infty (G) = T$, and $\hat{S} = N(R)$ which ensures  that there exists an important~\ttsep~$N(R') \subseteq V(G) \setminus T$ with~$|N(R')| \leq |N(R)|$ and~$R \subseteq R'$. Let~$S^* := (S \setminus N(R)) \cup N(R')$.

Next, we show that~$|S^*| \leq |S|$ and~$S^*$ is a \mwnsshort of~\gtk, thereby proving Condition~$\ref{condition:pushinglemmaformwns1}$ of the lemma for~$S_t^* = N(R')$. By construction of~$S^*$ and the fact that~$S$ is a \mwnsshort of~\gtk and~$|N(R')| \leq |N(R)|$, we have~$|S^*| \leq |S| \leq k$. Moreover, as $S, N(R') \subseteq V(G) \setminus T$, we have~$S^* \subseteq V(G) \setminus T$.
It remains to show that there does not exist a pair of distinct terminals~$t_i, t_j \in T$ such that there are 2 \ivd $t_i$-$t_j$ paths in~$G-S^*$. Assume for a contradiction that there exist distinct~$t_i, t_j \in T$ with 2 \ivd $t_i$-$t_j$ paths (say~$P_1$ and~$P_2$) in~$G-S^*$. Note that~$S^*$ is a \ttsep as $S^*$ contains the \ttsep~$N(R')$. This implies~$t_i, t_j \in T \setminus \{t\}$.
Observe that neither~$P_1$, nor~$P_2$ can intersect the set~$N(R)$: because if for some~$z \in [2]$, the path~$P_z$ goes through the set~$N(R)$, then it also goes though the set~$N(R') \subseteq S^*$ as~$R \subseteq R'$, a contradiction to the fact that~$P_z$ is present in~$G-S^*$. Since the paths~$P_1, P_2$ are disjoint from~$N(R)$ and~$S \setminus N(R) \subseteq S^*$, the paths are also present in~$G-S$, a contradiction to the fact  that~$S$ is a \mwnsshort.
	
\subparagraph{2. If $S$ is not a \ttsep for any~$t \in T$.}
We first invoke~Lemma~\ref{lemma:easilyseparableterminal} which ensures that there is a terminal~$t \in T$ and~a non-terminal~$v$ such that the set~$S_v = S \cup \{v\}$ is a~\ttsep. 
Let~$R$ be the set of vertices reachable from~$t$ in~$G - S_v$. If~$N(R) \subseteq S_v$ is an important~\ttsep then Condition~$\ref{condition:pushinglemmaformwns2}$ of the lemma holds for~$S^* = S, S_t = N(R)$, and~$v = v$.
Therefore assume that~$N(R)$ is not an important~\ttsep. Again, due to Proposition~\ref{proposition:existenceofimpsep} there is an important~\ttsep $N(R') \subseteq V(G) \setminus T$ with~$R \subseteq R'$ and~$|N(R')| \leq |N(R)|$. Since~$S$ alone is not a~\ttsep for any~$t \in T$, there exists a path~$P$ between~$t$ and some other terminal~$t' \in T\setminus \{t\}$ in~$G-S$ such that the path~$P$ passes though vertex~$v \in S_v$. Note that the same~$t$-$t'$ path~$P$ must contain some vertex~$v' \in N(R')$ (where~$v'$ may be the same as~$v$) since~$N(R')$ is a \ttsep. Let~$S^* := (S \setminus N(R)) \cup (N(R') \setminus \{v'\})$. Next we show that~$|S^*| \leq |S|$ and~$S^*$  is a \mwnsshort of~\gtk, implying that Condition~$\ref{condition:pushinglemmaformwns2}$ of the lemma holds for~$S_t = N(R')$ and~$v=v'$. By construction of~$S^*$ and the fact that~$N(R) \setminus \{v\} \subseteq S, v' \in N(R')$ and~$N(R') \leq N(R)$, we have~$|S^*| \leq |S| \leq k$. Moreover, we have~$S^* \subseteq V(G) \setminus T$ as both $S, N(R') \subseteq V(G) \setminus T$.

It remains to show that there does not exist a pair of distinct terminals~$t_i, t_j \in T$ such that there are 2 \ivd $t_i$-$t_j$ paths in~$G-S^*$. Assume for a contradiction that there exist distinct~$t_i, t_j \in T$ with 2 \ivd $t_i$-$t_j$ paths (say~$P_1$ and~$P_2$) in~$G-S^*$. First note that neither~$t_i$ nor~$t_j$ can be the vertex~$t$: there are 2 \ivd $t_i$-$t_j$ paths in~$G-S^*$, whereas the set~$S^* \cup \{v'\}$ is a~\ttsep (as~$N(R') \setminus \{v'\} \subseteq S^*$). Hence~$t_i, t_j \in T \setminus \{t\}$. Note that if both the paths~$P_1, P_2$ do not intersect with the set~$N(R) \setminus \{v\}$ then they also do not intersect the set~$S \subseteq S^* \cup (N(R) \cup \{v\})$, a contradiction to the fact that~$S$ is a solution of~\gtk. Hence there exists~$z \in [2]$ such that the path~$P_z$ intersects the set~$N(R) \setminus \{v\}$. Let~$w$ be an arbitrary vertex in $V(P_z) \cap (N(R) \setminus \{v\})$. Note that~$w \neq v'$ as~$w \in N(R) \setminus \{v\} \subseteq S, v'\in V(P)$, and the path~$P$ is present in~$G-S$.
We divide the path~$P_z$ into two internally vertex disjoint subpaths~$P_z[t_i, w]$ and~$P_z[w, t_j]$. Moreover, since~$P_z$ is a simple path, these subpaths are internally vertex disjoint and hence at most one of the subpaths (say~$P_z[t_i, w]$; the other case is symmetric) can intersect vertex~$v'$ of~$N(R')$. Let~$\hat{S} = S^* \cup \{v'\}$, note that the set~$\hat{S}$ is a \ttsep since~$N(R') \subseteq \hat{S}$. Observe that the subpath~$P_z[w, t_j]$ is still present in~$G - \hat{S}$.	
Since~$w \in N(R) \setminus \{v\}$, by definition of the set~$R$ there exists a~$t$-$w$ path~$P_{t, w}$ in~$G - \hat{S}$. Let~$P := P_{t, w} \cdot P_z[w, j]$.
Note that the~$t-t_j$ path~$P$ is present in~$G - \hat{S}$, a contradiction to the fact that~$N(R') \subseteq \hat{S}$ is a~\ttsep.

Since for both cases the above lemma holds, this concludes the proof of Lemma~\ref{lemma:pushinglemmaformwns}.
\end{proof}

\subsection{Proof of Lemma~\ref{lemma:Compression_Step_for_MWNS}}
\CompressionStepForMWNSState*
\begin{proof}
    We assume that~$T$ is an independent set in~$G$, otherwise the answer is trivially NO.
    Then we invoke the algorithm of Theorem~\ref{theorem:terminalbounding} with the \mwnsshort instance~\gtk and the MWNS~$\hat{S}=S_{k+1}$ of~$(G, T)$. In polynomial time, it outputs an equivalent instance~$(G', T', k')$ such that $G'$ is an induced subgraph of $G$, $T' \subseteq T$ with $|T'| = k^{\Oh(1)}$, and~$0 \leq k' \leq k$. If~$k' < 0$ then we know that~$(G', T', k')$ is a NO-instance of \mwnsshort. Since~\gtk and~$(G', T', k')$ are equivalent instances of \mwnsshort, it holds that~\gtk is a NO-instance of \mwnsshort. If~$k' = 0$ and there is a $T'$-cycle in~$G'$ then $(G', T', k')$ is a NO-instance of \mwnsshort. Hence~\gtk is also a NO-instance of~\mwnsshort. If~$k' \geq 0$ and there is no $T'$-cycle in~$G'$, then by Proposition~\ref{Proposition:DifferentWaysOfLookingAtMWNS} we know that~$\emptyset$ is a MWNS of~$(G', T')$ and hence~$(G', T', k')$ is a YES-instance of~\mwnsshort. Moreover, in this case, the algorithm of Theorem~\ref{theorem:terminalbounding} also ensures that in polynomial-time it can obtain a \mwnsshort~$S$ for~\gtk. 

    Else, $k' > 0$ and there is a $T'$-cycle in~$G'$. In this case, we use recursive branching to obtain a~\mwnsshort of~$(G', T', k')$ if it exists.
    Towards this, for each terminal~$t \in T'$ which is not nearly-separated, we invoke Theorem~\ref{Theorem:Algorithm_To_Generate_Imp_Sep} with~$G=G', X = \{t\}, Y = T' \setminus \{t\}$, and~$k = k'+1$, which in time~$\Oh (|\mathcal{S}_{k'+1}| \cdot (k'+1)^2 \cdot (n+m))$ outputs a set~$\mathcal{S}_{k'+1}$ of size at most~$4^{k'+1}$ consisting of every important~$(t, T' \setminus \{t\})$-separator of size at most~$k'+1$. Then for each important separator~$I \in \mathcal{S}_{k+1}$, we do one of the following.
    \begin{enumerate}
        \item If~$|I| \leq k$ then we recursively solve the \mwnsshort instance~$(G'-I, T', k' - |I|)$ and for each~$v \in I$, we recursively solve the \mwnsshort instance~$(G-(I \setminus \{v\}), T', k' - (|I|-1))$.

        \item If~$|I| = k+1$, then for each vertex~$v \in I$ we recursively solve the \mwnsshort instance~$(G-(I \setminus \{v\}), T', 0)$. 
    \end{enumerate}
If one of these branches with the \mwnsshort instance~$(G'-I', T', k'-|I'|)$ returns a solution~$\hat{S}$, then it is easy to observe that the set~$I' \cup \hat{S}$ is a \mwnsshort of~$(G', T', k')$. Hence again by spending polynomial-time, we can obtain a \mwnsshort $S$ of the original instance~\gtk as ensured by Theorem~\ref{theorem:terminalbounding} and output it. If all branches fail, we conclude that~$(G,T,k)$ does not have a solution.
The correctness of the algorithm is clear from Lemma~\ref{lemma:pushinglemmaformwns}.

Compared to the textbook FPT algorithm for \textsc{Multiway Cut}, our algorithm does significantly more work since it recurses for important~$(t,T' \setminus \{t\})$-separators for \emph{all} terminals in~$T'$ which are not nearly-separated. However, using the bound~$|T'| \in k^{\Oh(1)}$ and the fact that most choices of an important separator significantly decrease the parameter, we can still prove the following claim. Let~$L(k')$ be the number of leaves in the recursion tree for a call with parameter~$k'$.
\begin{restatable}{claim}{ABoundOnNumberOfLeavesInBranchingTreeState}
\label{claim:ABoundOnNumberOfLeavesInBranchingTree}
    $L(k') = 2^{\Oh(k \log k)}$.
\end{restatable}
\begin{claimproof}
Towards this, we have the following recurrence.
\begin{equation}
    L(k') \leq \begin{cases}
        1 & \mbox{if $k' \leq 1$}\\
        |T'| \cdot \sum_{j=2}^{k'+1} 4^j \cdot L(k'- (j-1)) & \mbox{otherwise.} 
    \end{cases}
\end{equation}
We simplify the above equation by plugging $i= j-1$ to obtain the following.
\begin{equation}
    L(k') \leq \begin{cases}
        1 & \mbox{if $k \leq 1$}\\
        |T'| \cdot \sum_{i=1}^{k'} 4^{i+1} \cdot L(k'-i) & \mbox{otherwise.} 
    \end{cases}
\end{equation}

We prove by induction on~$k'$ that $L(k') \leq (32|T'|)^{k'}$. Note that the inequality trivially holds if~$k' \leq 1$ as~$|T| \geq 1$. Suppose that the inequality holds for all values below~$k'$ then we obtain:
\begin{equation}
    \begin{split}
        L(k')    &\leq |T'| \cdot \sum_{i=1}^{k'} 4^{i+1} \cdot L(k'-i)
                = 4|T'| \cdot \sum_{i=1}^{k'} 4^i \cdot L(k'-i) \\
                &\leq 4|T'| \cdot \sum_{i=1}^{k'} 4^i \cdot (32|T'|)^{k'-i} &\mbox{by induction}\\
                &\leq 4|T'| \cdot 4^{k'} \cdot \sum_{i=1}^{k'} (4|T'|)^{k'-1} \cdot 2^{k'-i} &\mbox{as $(4|T'|)^{k'-1} \geq (4|T'|)^{k'-i}$}\\
                &\leq (4|T'|)^{k'} \cdot 4^{k'} \cdot 2^{k'} &\mbox{as $\sum_{i=0}^{k'-1} 2^i < 2^{k'}$}\\
                &\leq (32|T'|)^{k'}.
    \end{split}
\end{equation}

As~$|T'| = k^{\Oh(1)}$ and~$k' \leq k$, we get the desired bound of~$2^{\Oh (k \log k)}$ on the number of leaves in the recursion tree.
\end{claimproof}

Therefore the total time spent at leaves can be bounded by~$2^{\Oh (k \log k)} n^{\Oh(1)}$ as~$k \leq n$. Moreover, as the height of the recursion tree is at most~$k$, the number of nodes is bounded by~$k \cdot 2^{\Oh (k \log k)}$. 
As we use the algorithm of Theorem~\ref{Theorem:Algorithm_To_Generate_Imp_Sep} to construct the set~$\mathcal{S}_k$, the time spent at an internal node is bounded by~$\Oh(4^{k+1} \cdot k^2 \cdot (n+m))$.
Therefore the total time spent at internal nodes is bounded by~$2^{\Oh (k \log k)} \cdot k^3 \cdot (n+m)$. 
As~$k \leq n$ and~$m \leq n^2$, the total time spent by the algorithm is~$2^{\Oh (k \log k)} n^{\Oh(1)}$. This concludes the proof of Lemma~\ref{lemma:Compression_Step_for_MWNS}.
\end{proof}

\subsection{Proof of Theorem~\ref{theorem:MWNS_is_FPT}}
\label{section:ProofOfThm:MWNS_Is_FPT}
\MWNSIsFPTState*
\begin{proof}
    Given an instance~\gtk of \mwnsshort, we use the following steps to compute a solution of~\gtk if it exists or return NO if~\gtk is a NO-instance of \mwnsshort. If~$T$ is not an independent set then~\gtk is a NO-instance of \mwnsshort and hence return NO. If there is no $T$-cycle in $G$ then return~$\emptyset$.
    If there is a $T$-cycle in $G$ and~$ k=0$ then \gtk is a NO-instance of~\mwnsshort and hence return NO.

    Else we know that there is a $T$-cycle in~$G$ and~$k > 0$. 
    Let~$V := V(G) \setminus T$. 
    If~$|V| \leq k$ then note that the whole set~$V$ is a solution of~\gtk. 
    Thus assume that~$|V| > k$.
    Let~$\{v_1, \ldots, v_{|V|}\}$ be an arbitrary ordering of~$V$. For~$\ell \in [|V|]$, we use~$V_\ell := \{v_1, \ldots, v_\ell\}$ as shorthand.
    Let~$G_{k+1} := G[V_{k+1} \cup T]$.    
    We invoke Lemma~\ref{lemma:Compression_Step_for_MWNS} with~$(G = G_{k+2}, T, k)$ and the set~$S_{k+1} = V_{k+1}$ which in time~$2^{\Oh (k \log k)} n^{\Oh(1)}$ either concludes that~$(G_{k+2}, T, k)$ is a NO-instance of \mwnsshort or outputs a set~$S^*$ of size~$k$ such that~$S^*$ is a solution of~$(G_{k+1}, T, k)$.
    Note that if~$(G_{k+1}, T, k)$ is a NO-instance of \mwnsshort then~\gtk is also a NO-instance of~\mwnsshort and hence we return NO in this case.
    Whereas when the algorithm outputs a $k$-sized set~$S^*$ and there exists a vertex~$v_{k+2} \in V$, we iteratively call Lemma~\ref{lemma:Compression_Step_for_MWNS} with the \mwnsshort instance~$(G_{k+2}, T, k)$ and the set~$S_{k+1} := S^* \cup \{v_{k+2}\}$.
    In the case when the algorithm outputs a $k$-sized set~$S^*$ and there does not exist a vertex~$v_{k+2} \in V$, we have~$G_{k+1} = G$ and hence~$S^*$ is a solution of~\gtk.

    Note that the above steps can be applied at most $n$ times until we either conclude that~\gtk is a no instance or iteratively get to the whole graph~$G$ by adding an additional vertex in each iteration.
    Thus to solve~\gtk, we spent polynomial time and then call Lemma~\ref{lemma:Compression_Step_for_MWNS} at most~$n$ times. Hence the total running time of the algorithm is~$2^{\Oh (k \log k)} n^{\Oh(1)}$.
\end{proof}

\subsection{Proof of Claim~\ref{Claim:Z_is_a_MWNS_of_Subtree_below_d}}
\ZisaMWNSofSubtreeBelowdState*
\begin{claimproof}
                In the case that~$d$ is a non-terminal cutvertex (Case a) or a terminal cutvertex (Case b), then as remarked after the construction of~$Z$, it is easy to observe that~$Z$ is a MWNS of~$(G_d + x, T)$.
                Hence it remains to prove the claim for the case when~$d$ is a block (Case c). Towards this, assume for a contradiction that~$Z$ is not a MWNS of~$(G_d + x, T)$. First note that~$Z \cap T = \emptyset$ by construction of~$Z$. This implies that there is a $T$-cycle $C$ in~$(G_d + x) - Z$. Since~$\{x\}$ is a MWNS of~$(G, T)$ and~$(G_d + x)-Z \subseteq G$, we have~$x \in V (C)$. Let~$\puv := C - \{x\}$ be the path with endpoints~$u, v$. Note that~$\puv$ is contained in~$G_d - Z$.
        
                Next, we look at how~$\puv$ appears in the subtree~$\calF_d$ of the block-cut forest~$\calF$. First observe that the~$\puv$ must contain at least 2 vertices from the block~$d$, because if~$|V(\puv) \cap V(d)| \leq 1$ then we get a contradiction to the fact that~$d$ is a deepest node for which~$G[V_G(\mathcal{F}_d) \cup \{x\}]$ has a $T$-cycle.
                Let~$c_u, c_v \in V (d)$ be the first and the last vertex of~$d$ on the path~$\puv$ while traversing from~$u$ to~$v$. By the argument above, such~$c_u, c_v$ exist and~$c_u \neq c_v$. We orient the path~$\puv$ in such a way that the number of terminals on the subpath~$\puv(c_v,v] := \puv[c_v, v] - \{c_v\}$ is at least the number of terminals on the subpath~$\puv(c_u,u] := \puv[c_u, u] - \{c_u\}$, i.e., $|V (\puv (c_v, v]) \cap T| \geq |V (\puv (c_u, u]) \cap T|$. Note that since~$|V (\puv)\cap T| \geq 2$, $|V (\puv [c_u, c_v]) \cap T| \leq 1$, and~$|V (\puv (c_v, v]) \cap T| \geq |V (\puv (c_u, u]) \cap T|$, we have~$|V (\puv (c_v, v]) \cap T| \geq 1$. Next, we do a case distinction based on whether or not~$\puv$ visits a terminal while passing through~$d$.\\
                \textbf{Case 1. If~$|V (\puv [c_u, c_v] ) \cap T| = 0$.}\\
                We further divide Case 1 into the following two sub-cases based on whether or not the subpath~$\puv (c_u, u]$ contains a terminal.
        
                \textbf{Case 1.a. If~$|V (\puv (c_u, u]) \cap T| \geq 1$.}
        
                        By definition of the set~$\calC_1^d$ and~$\calC_{\geq 2}^d$ (in Step~\ref{Algo:A:Step:ConstructionOfZWhendIsABlock}) and the fact that $|V (\puv (c_v, v]) \cap T| \geq 1$ and~$|V (\puv (c_u, u]) \cap T| \geq 1$, we have~$c_u, c_v \in Q = \calC_1^d \cup \calC_{\geq 2}^d$.
                Since there are distinct vertices~$c_u, c_v \in Q$ such that there is a $c_u$-$c_v$ path~$\puv [c_u, c_v]$ in~$D_T - Z$ and~$Z_1 \subseteq Z$, this yields a contradiction to the fact that~$D_T - Z_1$ has no $Q$-path.
        
                \textbf{Case 1.b. If~$|V (\puv (c_u, u]) \cap T| = 0$.}
        
                Since the path~$\puv$ contains at least 2 terminals, $|V (\puv (c_u, u]) \cap T| = 0$, and~$|V (\puv [c_u, c_v]) \cap T| = 0$, we have~$|V_G (\puv (c_v, v]) \cap T| \geq 2$.
                Then we have~$c_u \in \calC_0^d \cup (N (x) \cap V (d))$ and~$c_v \in \calC_{\geq 2}^d$ such that there is a $c_u$-$c_v$ path~$\puv [c_u, c_v]$ in~$D_T - Z$ with~$Z_2 \subseteq Z$, yielding a contradiction to the fact that the set~$Z_2$ is an $(A, B)$ separator for~$A = \calC_{\geq 2}^d$ and~$B = \calC_0^d \cup (N (x) \cap V (d))$.\\
                \textbf{Case 2. If~$|V (\puv [c_u, c_v]) \cap T| = 1$.}\\
                Let~$t_i \in V (\puv) \cap (T \cap V (d))$. Next, we do a case distinction based on whether or not~$c_v = t_i$.
        
                \textbf{Case 2.a. If~$c_v \neq t_i$.}
                Since~$|V (\puv (c_v, v]) \cap T| \geq 1$ and~$c_v \notin t_i$, we have~$c_v \in \calC_1^d \cup \calC_{\geq 2}^d$. Let~$D^*$ be the \cct of~$D_T - Z$ which contains the vertex~$c_v$. Note that such a $D^*$ exists and we have~$N (t_i) \cap V (D^*) \neq \emptyset$, because the subpath~$\puv (t_i, c_v]$ is contained inside the \cct~$D^*$.
                Since the vertex~$c_v \in \calC_1^d \cup \calC_{\geq 2} ^d$ is contained in the \cct~$D^*$ of~$D_T - Z$ such that $N (t_i) \cap V (D^*) \neq \emptyset$, we have~$c_v \in Z_3 \subseteq Z$, a contradiction to the fact that the path~$\puv$ is present in~$G_d - Z$.
        
                \textbf{Case 2.b. If~$c_v = t_i$.}
                Since~$|V (\puv (c_v, v]) \cap T| \geq 1$, the graph~$G_{t_i}$ has a vertex~$v \in N (x)$ such that the $t_i$-$v$ path~$\puv [t_i, v]$ contains at least 2 terminals. By construction of the set~$Z_4 \subseteq Z$, such a path does not exist in~$G_d - Z$, a contradiction to the fact that the path~$\puv$ is present in~$G_d - Z$.

                Since all the above cases lead to a contradiction, the set~$Z$ is a MWNS of~$(G_d + x, T)$. This completes the proof of Claim~\ref{Claim:Z_is_a_MWNS_of_Subtree_below_d}.
            \end{claimproof}

\subsection{Proof of Claim~\ref{Claim:Contains_at_least_one_vertex_from_both_inside_and_outside}}
\AtLeastOneVertexFromBothInsideAndOutsideState*

\begin{proof}
                Consider a $T$-cycle $C$ in~$G- (Z \cup \hat{S})$. Since~$\{x\}$ is a MWNS of~$(G, T)$ and~$G-(Z \cup \hat{S}) \subseteq G$, we have~$x \in V (C)$.
                Let~$\puv:= C-\{x\}$ be the path with endpoints~$u, v$. First, note that~$\puv$ can not be completely contained in~$G_d = G[V_G (\calF_d)]$ because if~$\puv$ is completely contained in~$G_d$ then the $T$-cycle $C$ is present in~$G[V_G(\calF_d) \cup \{x\}]$ while~$C$ does not intersect~$Z$, a contradiction to the fact that~$Z$ is a MWNS of~$(G[V_G (\calF_d) \cup \{x\}], T)$.
                Next, note that the path~$\puv$ can not be completely contained in~$G' = G[V_G (\calF) \setminus V_G (\calF_d)]$ because if~$\puv$ is completely contained in~$G'$ then the $T$-cycle~$C$ 
                is present in~$G' + x$ while~$C$ does not intersect~$\hat{S} \cup V_G (\calF_d) \supseteq S^*_x$, a contradiction to the fact that~$S^*_x$ is a MWNS of~$(G, T)$.
                Since the path~$\puv$ can neither be completely contained in~$G_d$ nor in~$G[V_G (\calF) \setminus V_G (\calF_d)]$, it contains at least one vertex from both~$V_G (\calF_d)$ and~$V_G(\calF) \setminus V_G(\calF_d)$, which implies the same for the $T$-cycle $C \supseteq \puv$.
            \end{proof}

\subsection{Proof of Claim~\ref{Claim:Any_Optimal_MWNS_excluding_x_picks_enough_vertices_below_d}}
\AnyOptimalMWNSEPicksEnoughVerticesBelowdState*
            \begin{claimproof}             
                First we do a case distinction (similar to Step~\ref{Algo:A:Step:CaseDistinction}) based on the type of the deepest node~$d$ selected by the algorithm in Step~$({\ref{Algo:A:Step:ChooseADeepestNode}})$.
                
                \subparagraph{Case 1. If~$d$ is a non-terminal cutvertex.}
                In this case~$Z = \{d\}$. Moreover, since there is a~$T$-cycle in~$G_d + x$ by our choice of~$d$, and~$S^*_x$ is an $x$-avoiding MWNS of~$(G, T)$, we have~$|S^*_x \cap V_G (\calF_d)| \geq 1$. Hence the inequality of Claim~\ref{Claim:Any_Optimal_MWNS_excluding_x_picks_enough_vertices_below_d} holds.

                \subparagraph{Case 2. If~$d$ is a terminal cutvertex.}
                In this case we have~$Z = \mathcal{C}_{\geq 1} (d)$ by definition, so that~$Z$ consists of some cutvertices of~$G-\{x\}$ which are grandchildren of~$d$ in the block-cut forest~$\calF$.
                For each cutvertex~$z_i \in Z$, the graph~$G_{z_i}$ contains a vertex~$p_i \in N_G(x)$ such that there is a~$z_i$-$p_i$ path~$P_i$ in~$G_{z_i}$ which contains at least one terminal.
                Moreover, for distinct~$z_i, z_j \in Z$, the corresponding paths~$P_i$ and~$P_j$ in~$G_{z_i}$ and~$G_{z_j}$, respectively are vertex-disjoint since~$V_G (\calF_{z_i}) \cap V_G (\calF_{z_j}) = \emptyset$. Since~$d$ is a deepest node with a~$T$-cycle in~$G_d + x$, the graph~$G[V_G (\calF_B) \cup \{x\}]$ does not contain a $T$-cycle for each block~$B \in \mbox{Child}_\calF (d)$. Then due to Observation~\ref{Proposition:NumberOfInterestingCutVerticesAreBounded}, we know that for each block~$B \in \text{Child}_\calF (d)$ there is at most one vertex from~$Z$ in~$\text{Child}_\calF (B)$. 
                Moreover, for any vertex~$z_i \in Z$, there is a $d$-$x$ path~$R_i$ in~$G_d + x$ such that~$z_i \in V(R_i)$ and all internal vertices of~$R_i$ belong to~$V_G(\calF_{\mbox{parent}(z_i)})$.
                Let~$R_i := W_i[d, z_i] \cdot P_i [z_i, p_i] \cdot \{p_i, x\}$, where~$W_i$ is a $d$-$z_i$ path inside the block~$\mbox{parent}_{\calF} (z_i)$. By construction of~$R_i$, it is clear that~$z_i \in V(R_i)$ and all internal vertices of~$R_i$ belong to~$V_G(\calF_{\mbox{parent}(z_i)})$.
                Also, for distinct~$z_i, z_j \in Z$, the corresponding $d$-$x$ paths~$R_i$ and~$R_j$ are \ivd as $V_G(\calF_{\mbox{parent}(z_i)}) \cap V_G(\calF_{\mbox{parent}(z_j)}) = \{d\}$.
                Thus, the graph~$G_d + x$ has~$|Z|$ pairwise \ivd~$d$-$x$ paths, each containing at least one terminal other than~$d$ itself.
                Hence, any MWNS excluding the vertices~$x, d$ (thus~$S^*_x$) intersects all but at most one of those paths and hence~$|S^*_x \cap V_G (\calF_d)| \geq \max \{1, |Z|-1\} \geq \max \{1, \frac{|Z|-2}{6}\}$ since~$|Z| \geq 1$ as there is a $T$-cycle in~$G_d + x$. Hence the inequality of Claim~\ref{Claim:Any_Optimal_MWNS_excluding_x_picks_enough_vertices_below_d} holds in the case when~$d$ is a terminal cutvertex.

                \subparagraph{Case 3. If~$d$ is a block.}
                 Let~$Z_1, Z_2, Z_3, Z_4, Z_5, Z \subseteq V_G (\calF_d) \setminus T$ be the sets computed in Step~\ref{Algo:A:Step:ConstructionOfZWhendIsABlock}, where~$Z = \bigcup_{i=1}^5 Z_i$. 
                 Next, for each~$i \in [5]$, we do an analysis for the set~$Z_i$.
                 
                 \textbf{Analysis for the set~$Z_1$:}
                        While choosing~$|Z_1|$ vertices in Step~\ref{Algo:A:Step:ConstructionOfZWhendIsABlock} (using Gallai's theorem), we ensured that there is a family~$\mathcal{P}_Q$ containing at least~$\frac{|Z_1|}{2}$ pairwise vertex-disjoint $Q$-paths in the graph~$D_T = G [V_G (d) \setminus T]$, for~$Q = \calC_{\geq 2}^d \cup \calC_1^d$. For each~$q_i \in Q$, the definitions of~$Q, \calC_{\geq 2}^d$, and $\calC_1^d$ ensure that there exists a $q_i$-$p_i$ path~$P_i$ in~$G_{q_i}$ for some~$p_i \in N_G(x)$ such that~$P_i$ contains at least one terminal. Hence for each~$q_i$-$q_j$ path $Q_{ij}$ in~$\mathcal{P}_Q$, there is a $T$-cycle on~$x$, say~$C_{ij}$ which can be defined as follows.
                        Let~$C_{ij} := \{x, p_i\} \cdot P_i [p_i, q_i] \cdot Q_{ij} [q_i, q_j] \cdot P_j [q_j, p_j] \cdot \{p_j, x\}$.  

                        Moreover, it is easy to observe that there are~$|\mathcal{P}_Q| \geq \frac{|Z_1|}{2}$ pairwise $\{x\}$-disjoint cycles on~$x$ (one corresponding to each $Q$-path of~$\mathcal{P}_Q$) in~$G_d + x$ because for each~$q_i \in Q$ the $q_i$-$p_i$ path~$P_i$ contains vertices from~$V_G(\calF_{q_i})$.
                        Hence by Observation~\ref{Observation:Any_MWNS_contains_at_least_one_vertex_from_each_T-disjoint_cycle}, we have~$|S^*_x \cap V_G (\calF_d)| \geq \frac{|Z_1|}{2}$.

                         \textbf{Analysis for the set~$Z_2$:}
                         Since~$|Z_2|$ is the cardinality of a minimum~$(A, B)$ separator in the graph~$D_T$, by Menger's theorem there exists a family~$\calP_2$ of~$|Z_2|$ pairwise vertex-disjoint $(A, B)$ paths in~$D_T$, for~$A= \calC_{\geq 2}^d$ and~$B = \calC_0^d \cup (N_G (x) \cap V_G (D_T))$. Similar to the above analysis done for the set~$Z_1$, for any path~$R \in \calP_2$ with endpoints~$c_2 \in A$ and~$c_0 \in B$, we can obtain a $T$-cycle~$C_{20}$ as follows.
                         Let~$C_{20} := \{x, p_0\} \cdot P_0 [p_0, c_0] \cdot R[c_0, c_2] \cdot P_2 [c_2, p_2] \cdot \{p_2, x\}$, where~$P_0$ is a $c_0 - p_0$ path in the graph~$G_{c_o}$ for some~$p_0 \in N_G (x)$, and $P_2$ is a $c_2 - p_2$ path in the graph~$G_{c_2}$ for some~$p_2 \in N_G (x)$ with~$|V_G (P_2) \cap T| \geq 2$. By definition of~$A$ and~$B$, such paths exist. Again, it is easy to observe that there are at least~$|Z_2|$ pairwise $(T \cup \{x\})$-disjoint $T$-cycles on~$x$. Hence we have~$|S^*_x \cap V_G (\calF_d)| \geq |Z_2|$.

                         \textbf{Analysis for the set~$Z_3$:}
                         By definition of~$Z_3 \subseteq \calC_1^d \cup \calC_{\geq 2}^d$, for each~$z_i \in Z_3$ there exists a \cct~$D_i$ in~$D_T - (Z_1 \cup Z_2)$ such that~$z_i \in V_G (D_i)$ and~$N_G (t) \cap V_G (D_i) \neq \emptyset$.
                         Let~$n_i \in N_G (t) \cap V_G (D_i)$. Since~$z_i \in \calC_1^d \cup \calC_{\geq 2}^d$, by definition of the set~$\calC_1^d$ and~$\calC_{\geq 2}^d$, the graph~$G_{z_i}$ contains a vertex~$p_i \in N_G (x)$ such that there is a $z_i$-$p_i$ path~$P_i$ in~$G_{z_i}$ with~$|V_G(P_i) \cap T| \geq 1$.          
                         Thus for each vertex~$z_i \in Z_3$, there is a $t$-$x$ path~$W_i$ containing the vertex~$z_i$ and at least one terminal other than~$t$, which can be obtained as follows.
                         Let~$W_i := \{t, n_i\} \cdot R_i [n_i, z_i] \cdot P_i [z_i, p_i] \cdot \{p_i, x\}$, where~$R_i$ is a path between vertices~$n_i, z_i \in V_G (D_i)$ in side \cct~$D_i$.
                         Observe that~$W_i$ contains at least one terminal other than~$t$ as the subpath~$P_i$ of~$W_i$ contains at least one terminal other than~$t$. Moreover, each \cct $D_i$ of~$D_T - (Z_1 \cup Z_2)$ contains at most one vertex from~$\calC_1^d \cup \calC_{\geq 2}^d$ as the set~$Z_1$ hits all~$(\calC_1^d \cup \calC_{\geq 2}^d)$-paths. Then by definition of~$Z_3$, construction of~$W_i$ and the fact that~$V_G (\calF_{z_i}) \cap V_G(\calF_{z_j}) = \emptyset$ for distinct~$z_i, z_j \in Z_3$, it is easy to observe that the graph~$G_d + x$ has~$|Z_3|$ \ivd $t$-$x$ paths, each containing at least one terminal other than $t$. Hence, any~MWNS excluding the vertices~$t, x$ (thus~$S^*_x$) intersects all but at most one of those paths and hence $|S^*_x \cap V_G (\calF_d)| \geq |Z_3| - 1$.

              \textbf{Analysis for the set~$Z_4$ and~$Z_5$:}
                         By construction of~$Z_4$, if~$\text{child}_{\calF_d} (d) \cap T = \emptyset$ then~$Z_4 = \emptyset$. Otherwise, let~$t \in \text{child}_{\calF_d} (d) \cap T$. Since~$d$ is a deepest node with a $T$-cycle in~$G_d + x$, there does not exist a~$T$-cycle in~$G_t + x$. 
                         Thus when we apply Proposition~\ref{Proposition:NumberOfInterestingCutVerticesAreBounded} with respect to the vertex~$t$, we have~$|Z_4| = 1$. Overall, we have~$|Z_4| \leq 1$.
                         Also, by definition of~$Z_5$, we have~$|Z_5| \leq 1$.

                Now we combine the analysis of~$Z_i$'s to prove that Claim~\ref{Claim:Any_Optimal_MWNS_excluding_x_picks_enough_vertices_below_d} holds in the case when $d$ is a block. Towards this, let~$Z_{123} = Z_1 \cup Z_2 \cup Z_3$. Since $|Z_4|, |Z_5| \leq 1$, we know that $|Z_{123}| \geq |Z|-2$. Hence to prove Claim~\ref{Claim:Any_Optimal_MWNS_excluding_x_picks_enough_vertices_below_d}, it suffices to show that~$|S^*_x \cap V_G (\calF_d)| \geq \max \{1, \frac{|Z_{123}|}{6} \}$. 
                If~$|Z_{123}| \leq 6$ then the right-hand side of the inequality is 1 and the result follows since our choice of~$d$ ensures there is a~$T$-cycle on~$x$ in~$G_d + x$, any~$x$-avoiding MWNS (and hence~$S^*_x$) contains at least one vertex from~$V_G (\calF_d)$. 
                So we assume~$|Z_{123}| > 6$. Let~$Z_i$ be a set with the maximum cardinality among~$Z_1,Z_2,Z_3$, then~$|Z_i| \geq |Z_{123}| / 3 \geq 6/3 = 2$. 
                To prove that~$|S^*_x \cap V_G (\calF_d)| \geq \frac{|Z_{123}|}{6}$, it now suffices to show~$|S^*_x \cap V_G (\calF_d)| \geq \frac{|Z_i|}{2}$. This follows from the construction of~$Z_1, Z_2, Z_3$: if~$Z_i = Z_1$ we already show that~$|S^*_x \cap V_G (\calF_d)| \geq \frac{|Z_1|}{2}$, if~$Z_i = Z_2$, we already show~$|S^*_x \cap V_G (\calF_d)| \geq |Z_2|$, and if~$Z_i = Z_3$ we have~$|S^*_x \cap V_G (\calF_d)| \geq |Z_3|-1 \geq \frac{|Z_3|} {2}$ since we could assume~$|Z_i| \geq 2$.
                This concludes the proof of Claim~\ref{Claim:Any_Optimal_MWNS_excluding_x_picks_enough_vertices_below_d}.
            \end{claimproof}

\subsection{Proof of Claim~\ref{Claim:ShatIsAMWNSIfOptContains1VertexFromSubtreeBelow_d}}
\ShatIsAMWNSIfOptContainsOneVertexFromSubtreeBelowdState*
\begin{claimproof}
                Assume for a contradiction that~$|S^*_x \cap V_G (\calF_d)| = 1$ whereas~$\hat{S}$ is not a MWNS of~$(G-Z, T)$. Since~$\hat{S} \subseteq S^*_x \subseteq V(G) \setminus T$ and~$\hat{S}$ is not a MWNS of~$(G-Z, T)$, there exists a $T$-cycle $C$ in~$G- (Z \cup \hat{S})$. Moreover, $x \in V(C)$ as~$\{x\}$ is a MWNS of~$(G, T)$. Let~$\puv := C - \{x\}$ be the path with the endpoints~$u, v$. Due to Claim~\ref{Claim:Contains_at_least_one_vertex_from_both_inside_and_outside}, one endpoint say $u$ of the simple path~$\puv$ is contained in~$V_G(\calF_d)$ whereas the other endpoint~$v$ is contained in the set~$V_G (\calF) \setminus V_G(\calF_d)$. Also, note that since~$x \in V(G) \setminus T$ and~$|V (C) \cap T| \geq 2$, we have~$|V_G (\puv) \cap T| \geq 2$.
                Let~$w$ be the last vertex of~$V_G (\calF_d)$ on the path~$\puv$ while traversing from~$u$ to~$v$.
                First, observe that~$w \in T$, because if~$w \notin T$ then by construction of~$Z$, we have~$w \in Z$ (when~$d$ is a cutvertex we have~$d=w$ and~$Z = \{w\}$, whereas when~$d$ is a block we have~$w \in Z_5 \subseteq Z$), a contradiction to the fact that the cycle~$C$ is present in~$G- (Z \cup \hat{S})$. 
                
                Next, we show that the subpath~$\puv [u, w]$ does not contain any terminal other than~$w$.
                    \begin{equation}
                    \label{Equation:Property1}
                        V (\puv [u, w]) \cap T = \{w\}.
                    \end{equation}
                    
                        We will prove that the above equation holds before continuing the rest of the proof of Claim~\ref{Claim:ShatIsAMWNSIfOptContains1VertexFromSubtreeBelow_d}.
                        We prove this using a case distinction similar to Step~\ref{Algo:A:Step:CaseDistinction}.
                        First observe that since~$w \in T$, we are not in the case when~$d$ is a non-terminal cutvertex, i.e. Case a. Hence, we are only left with cases when~$d$ is either a terminal cutvertex or~$d$ is a block.
                        First, consider the case when~$d$ is a terminal cutvertex. Note that in this case by definition of the vertex~$w$, we have~$d= w \in T$.
                            Then by the construction of the set~$Z = \mathcal{C}_{\geq 1} (d)$ in Step~\ref{Algo:A:Step:ConstructionOfZWhendIsATerminalcutvertex} and the fact that for each block~$B \in \mbox{child}_\calF (d)$ we have~$V(B) \cap T = \{d\}$, there can not be any other terminal on the subpath~$\puv[u, w]$ except~$d = w$.
                        Next, we consider the case when~$d$ is a block. Note that in this case, by definition of~$w$ and the fact that~$u \in V_G (\calF_d)$ and~$v \in V_G (\calF) \setminus V_G (\calF_d)$, 
                        we have~$w \in V(d)$ and~$w = \text{parent}_{\calF} (d)$. Again by the construction of the set~$Z_3 \subseteq Z$ (in Step~\ref{Algo:A:Step:ConstructionOfZWhendIsABlock}) and the fact that block~$d$ (of~$G-\{x\}$) contains at most one terminal, there can not be any other terminal on the subpath~$\puv[u, w]$ except~$w$. This concludes the proof of Equation~\ref{Equation:Property1}.

                We now continue with the proof of Claim~\ref{Claim:ShatIsAMWNSIfOptContains1VertexFromSubtreeBelow_d}.
                Since~$\puv$ contains at least 2 distinct terminals and the subpath~$\puv[u, w]$ contains exactly 1 terminal, namely~$w$, the remaining subpath~$\puv(w, v]$ must contain at least one terminal. Moreover, by our choice of~$d$ in Step~\ref{Algo:A:Step:ChooseADeepestNode} there is a $T$-cycle~$C'$ in the~$G_d + x$ and~$V(C') \cap V(d) \neq \emptyset$. Also, we have~$x \in V(C')$ as~$\{x\}$ is a MWNS of~$(G, T)$. Let~$\puvprime := C' - \{x\}$ be the path with endpoints~$u'$ and~$v'$. Next, we do a case distinction based on whether~$d$ is a cutvertex or a block.

                    \textbf{Case 1. If~$d$ is a cutvertex.}\newline
                        We have~$d = w  \in T$ in this case. Since~$d$ is a deepest node satisfying the conditions of Step~\ref{Algo:A:Step:ChooseADeepestNode}, we have~$d \in V(\puvprime)$.
                        Consider the subpaths~$\puvprime [u', d]$ and~$\puvprime [d, v']$. Note that~$V (\puvprime [u', d]) \cap V (\puvprime [d, v']) = \{d\}$.
                        Since~$d \in T$, $S^*_x \subseteq V(G) \setminus (T \cup \{x\})$, and~$|S^*_x \cap V_G (\calF_d)| = 1$, the set~$S^*_x$ can intersect at most on of the subpaths among~$\puvprime [u', d]$ and~$\puvprime [d, v']$.
                        Assume w.l.o.g. that~$S^*_x \cap V (\puvprime [u', d]) = \emptyset$.
                        Let~$C^* := \{x, u'\} \cdot \puvprime [u', d] \cdot \puv [d, v] \cdot \{v, x\}$. Note that~$C^*$ is a $T$-cycle on~$x$ avoiding the set~$S^*_x$, a contradiction to the fact that the set~$S^*_x$ is a MWNS of~$(G, T)$.

                    \textbf{Case 2. If~$d$ is a block.}\newline
                        As observed earlier (during the proof of Equation~\ref{Equation:Property1}), we have~$w \in V (d)$ and~$w = \text{parent}_{\calF} (d)$.
                        Also, observe that in this case we have, $|V (P'[u'v']) \cap V(d)| \geq 2$, because if~$|V (P'[u'v']) \cap V(d)| \leq 1$ then we get a contradiction to the fact that~$d$ is a deepest node with a $T$-cycle in~$G_d + x$.
                        Let~$c_{u'}, c_{v'} \in V (d)$ be the first and the last vertex of~$d$ on the path~$\puvprime$ while traversing from~$u'$ to~$v'$. By the argument above, such~$c_{u'}, c_{v'}$ exist and~$c_{u'} \neq c_{v'}$.
                        Next, observe that since~$S^*_x$ is an $x$-avoiding MWNS of~$(G, T)$ and~$C'$ is a $T$-cycle in~$G_d+x \subseteq G$, we have~$V (\puvprime) \cap S^*_x \neq \emptyset$. Moreover, since~$|S^*_x \cap V_G (\calF_d)| = 1$ and the path~$\puvprime$ is present in~$G_d$, we have~$|S^* \cap V (\puvprime)| = 1$.
                        Consider the (partition of~$\puvprime$ into) pairwise vertex disjoint subpaths~$\puvprime [u', c_{u'}]$, $\puvprime [c_{v'}, v']$, and~$\puvprime (c_{u'}, c_{v'}) := \puvprime [c_{u'}, c_{v'}] - \{c_{u'}, c_{v'}\}$.
                        Next, we do a case distinction based on whether or not the set~$S^*_x$ intersects the subpath~$\puvprime (c_{u'}, c_{v'})$.

                            \textbf{Case I. If~$|S^*_x \cap \puvprime (c_{u'}, c_{v'})| = 1$.}\newline
                                Since~$d$ is a biconnected component of~$G-\{x\}$, $|S^*_x \cap \puvprime (c_{u'}, c_{v'})| = 1$, and~$S^*_x \cap T = \emptyset$, the graph~$G[V(d) \setminus S^*_x]$ has a path say~$R$ between vertices~$c_{u'}$ and~$w \in T$ in block~$d$. Let~$C^* := \{x, u'\} \cdot \puvprime [u', c_{u'}] \cdot R[c_{u'}, w] \cdot P_{uv} [w, v] \cdot \{v, x\}$. Note that~$C^*$ is a $T$-cycle in~$G-S^*_x$, a contradiction to the fact that the set~$S^*_x$ is a MWNS of~$(G, T)$.

                            \textbf{Case II. If~$|S^*_x \cap V (\puvprime [u', c_{u'}])| =  1$ or $|S^*_x \cap V (\puvprime [v', c_{v'}])| = 1$.}\newline
                                Since the subpaths~$\puvprime [u', c_{u'}]$ and~$\puvprime [v', c_{v'}]$ are vertex disjoint and $|S^*_x \cap V_G (\calF_d)| = 1$, the set~$S^*_x$ either intersects the subpath~$\puvprime [u', c_{u'}]$ or~$\puvprime [v', c_{v'}]$, but not both.
                                Assume w.l.o.g. that~$S^*_x$ intersects the subpath~$\puvprime [v', c_{v'}]$. Let~$R$ be a path between vertices~$c_{u'}, w \in V (d)$ inside the block~$d$.
                                Let~$C^* := \{x, u' \} \cdot \puvprime [u', c_{u'}] \cdot R [c_{u'}, w] \cdot \puv [w, v] \cdot \{v, x\}$.
                                Note that~$C^*$ is a $T$-cycle in~$G- S^*_x$, a contradiction to the fact that the set~$S^*_x$ is a MWNS of~$(G, T)$. 

                        Since both cases lead to a contradiction, our assumption that~$\hat{S}$ is not a MWNS of~$(G-Z, T)$ is wrong, i.e., the set~$\hat{S}$ is a MWNS of~$(G, T)$.
                        This concludes the proof of Claim~\ref{Claim:ShatIsAMWNSIfOptContains1VertexFromSubtreeBelow_d}.
            \end{claimproof}

\subsection{Proof of Claim~\ref{Claim:S-hat_is_at_most_1_away_from_a_MWNS}}

\ShatISaNearMWNSState*
\begin{claimproof}
                Since~$\hat{S}$ is not a MWNS of~$(G-Z, T)$ and~$\hat{S} \cap T = \emptyset$, there is a $T$-cycle~$C$ in~$G- (Z \cup \hat{S})$. 
                Note that since~$\{x\}$ is a MWNS of~$(G, T)$ and~$G-(Z \cup \hat{S}) \subseteq G$, we have~$x \in V (C)$.
                Consider the derived path~$\puv := C-\{x\}$ such that~$u, v$ are endpoints of~$\puv$ with~$u \in V_G (\calF_d)$ and~$v \in V_G (\calF) \setminus V_G (\calF_d)$.
                Since~$Z$ is a MWNS of~$(G_d+x, T)$ and the~$T$-cycle~$C$ is present in~$G-(Z \cup \hat{S})$, we have such a path~$\puv$ due to Observation~\ref{Claim:Contains_at_least_one_vertex_from_both_inside_and_outside}.
                Moreover, since~$x \notin T$ and~$|V(C) \cap T| \geq 2$, we have~$|V (\puv) \cap T| \geq 2$.
                Let~$p = \text{parent}_{\calF} (d)$. Consider the edge~$e = \{d, p\}$ of the block-cut forest~$\calF$. Let~$c \in \{d, p\}$ be the unique cutvertex incident on edge~$e$. Note that due to Observation~\ref{Claim:Contains_at_least_one_vertex_from_both_inside_and_outside}, we have~$c \in V (\puv)$. 
                Note that~$c \in T$ because if~$c \notin T$ then by construction of~$Z$, we have~$c \in Z$, a contradiction to the fact that the~$T$-cycle~$C$ does not intersect the set~$Z$.
                Moreover, due to the construction of~$Z$ and the fact that a block of~$\calF$ can contain at most 1 terminal, the subpath~$\puv [u, c]$ does not contain any terminal other than~$c$ itself, i.e., we have~$V(\puv [u, c]) \cap T = \{c\}$. Since~$\puv$ contains at least 2 terminals and the subpath~$\puv[u,c ]$ contains exactly one terminal~$c$, we have~$V(\puv[c, v]) \cap (T \setminus \{c\}) \neq \emptyset$, i.e., the remaining subpath~$\puv[c, v]$ contains a terminal (say~$t_2 \neq c \in T$).
                
                Let~$\hat{c}$ be the cutvertex after the cutvertex~$c$ on the path~$\puv$ while traversing from~$u \in V_G (\calF_d)$ to~$v \in V_G (\calF) \setminus \calF_d$. Note that~$\hat{c} \notin T$ because a block of~$\calF$ can contain at most 1 terminal.
                Next, we show that the set~$S' : = \hat{S} \cup \{\hat{c}\}$ is a MWNS of~$(G-Z, T)$. Towards this, assume for a contradiction that~$S'$ is not a MWNS of~$(G-Z, T)$. Since~$S' \cap T = \emptyset$, implying that there is a~$T$-cycle $C'$ in~$G-(Z \cup S')$. Again, consider the derived path~$\puvprime := C' - \{x\}$ such that~$u' \in V_G (\calF_d)$ and~$v' \in V_G (\calF) \setminus V_G (\calF_d)$ are endpoints of~$\puvprime$. Note that~$c \in V (\puvprime)$. Next, we construct a $T$-cycle~$C^*$ as follows.
                Let~$C^* = \{x, v\} \cdot \puv[v, c] \cdot \puvprime[c, v'] \cdot \{v', x\}$. Note that~$C^*$ is a $T$-cycle which does not intersect the set~$\hat{S} \cup (V_G (\calF_d) \setminus T) \supseteq S^*_x$, a contradiction to the fact that~$S^*_x$ is a MWNS of~$(G, T)$. This concludes the proof of Claim~\ref{Claim:S-hat_is_at_most_1_away_from_a_MWNS}.
            \end{claimproof}

\subsection{Remaining proof of Lemma~\ref{Lemma:OPTDecreasesByAFactorOf14}}
\label{subsection:RemainingProof:Lemma:OPTDecreasesByAFactorOf14}
We continue to prove Lemma~\ref{Lemma:OPTDecreasesByAFactorOf14}. Towards this, we do a case distinction based on whether or not the set~$\hat{S}$ is a MWNS of~$(G-Z, T)$.

           \textbf{Case 1. If~$\hat{S}$ is a MWNS of~$(G-Z, T)$}.
           
                Since~$\hat{S}$ is a MWNS of~$(G-Z, T)$ excluding vertex~$x$, we have:
                        \begin{equation}
                            \begin{split}
                                \text{\optx} (G-Z, T) &\leq |\hat{S}|
                                                        = |S^*_x \setminus V_G (\calF_d)|                      &&\text{by definition of~$\hat{S}$} \\
                                                        &= \text{\optx} (G, T) - |S^*_x \cap V_G (\calF_d)|     &&\text{by definition of~$S^*_x$} \\
                                                        & \leq \text{\optx} (G, T) - \max \{1, \frac{|Z|-2}{6}\}    &&\text{by Claim~\ref{Claim:Any_Optimal_MWNS_excluding_x_picks_enough_vertices_below_d}}                                        
                            \end{split}
                        \end{equation}

             We further analyze the above inequality based on whether or not~$|Z| \leq 8$.

                        \textbf{When~$|Z| \leq 8$}
                            \begin{equation}
                                \begin{split}
                                    \text{\optx} (G-Z, T)   & \leq 
                                                                  \text{\optx} (G, T) - 1                         
                                                             \leq  \text{\optx} (G, T) - \frac{|Z|}{14}           
                                \end{split} 
                            \end{equation}

                        \textbf{When~$|Z| > 8$}
                            \begin{equation}
                                \begin{split}
                                    \text{\optx} (G-Z, T)    &\leq  
                                                               \text{\optx} (G, T) - \frac{|Z|-2}{6}       
                                                            \leq  \text{\optx} (G, T) - \frac{|Z|}{14}        
                                \end{split}
                            \end{equation}        
        
            Since both the cases give the desired bound hence Lemma~\ref{Lemma:OPTDecreasesByAFactorOf14} holds in the case when the set~$\hat{S}$ is a MWNS of~$(G-Z, T)$.

            \textbf{Case 2. If~$\hat{S}$ is not a MWNS of~$(G-Z, T)$}.
            
            First note that since~$\hat{S}$ is not a MWNS of~$(G-Z, T)$, we have~$|S^*_x \cap V_G (\calF_d)| \geq 2$ due to Claim~\ref{Claim:ShatIsAMWNSIfOptContains1VertexFromSubtreeBelow_d}.
            Moreover, due to Claim~\ref{Claim:S-hat_is_at_most_1_away_from_a_MWNS}, we know that there is a non-terminal~$c \in V(G) \setminus (T \cup \{x\})$ such that the set~$\hat{S} \cup \{x\}$ is a MWNS of~$(G-Z, T)$. Since~$\hat{S} \cup \{c\}$ is a MWNS of~$(G-Z, T)$ excluding vertex~$x$, we have:
                        \begin{equation}
                            \begin{split}
                                \text{\optx} (G-Z, T)   &\leq |\hat{S} \cup \{x\}|
                                                        =  |\hat{S}| + 1                               
                                                        =  |S^*_x \setminus V_G (\calF_d)| + 1                     &&\text{by definition of~$\hat{S}$}\\
                                                        &=  \text{\optx} (G, T) - |S^*_x \cap V_G (\calF_d)| + 1    &&\text{by definition of~$S^*_x$}\\
                                                        &=  \text{\optx} (G, T) - \max \{2, \frac{|Z|-2}{6}\} + 1       &&\text{by Claim~\ref{Claim:Any_Optimal_MWNS_excluding_x_picks_enough_vertices_below_d} and Claim~\ref{Claim:ShatIsAMWNSIfOptContains1VertexFromSubtreeBelow_d}}
                            \end{split}
                        \end{equation}
                        We further analyze the above inequality based on whether or not~$|Z| \leq 14$.

                        \textbf{When~$|Z| \leq 14$}
                            \begin{equation}
                                \begin{split}
                                    \text{\optx} (G-Z, T)   & \leq  
                                                                 \text{\optx} (G, T) - 2 + 1                            
                                                             \leq  \text{\optx} (G, T) - \frac{|Z|}{14}                    
                                \end{split} 
                            \end{equation}

                        \textbf{When~$|Z| > 14$}
                            \begin{equation}
                                \begin{split}
                                    \text{\optx} (G-Z, T)   & \leq  
                                                              \text{\optx} (G, T) - \frac{|Z|-2}{6} + 1                  
                                                             \leq  \text{\optx} (G, T) - \frac{|Z|}{14}                   
                                \end{split}
                            \end{equation}        
                        
                    Since both the cases give the desired bound hence Lemma~\ref{Lemma:OPTDecreasesByAFactorOf14} holds in the case when~$\hat{S}$ is not a MWNS of~$(G-Z, T)$.
                    Note that the approximation factor of 14 derived above is tight because in the case when~$\hat{S}$ is not a MWNS of~$(G-Z, T)$ and when~$|Z| \leq 14$ we need a factor 14 to hold the inequalities.
            This concludes the proof of Lemma~\ref{Lemma:OPTDecreasesByAFactorOf14}.

\subsection{Proof of Theorem~\ref{theorem:blocker}}
\label{subsection:ProofOfBlocker}
\BlockerState*

\begin{proof}
     Let~$S_x$ be the set returned by the above algorithm on the input~$(G, T, x)$, i.e., $S_x$ is the output of $\mbox{Blocker} (G, T, k)$.     
    It is easy to observe that~$\mbox{Blocker}(G, T, x)$ takes polynomial time because each step of the recursive algorithm takes polynomial time, and the size of the input graph reduces by at least one (as~$|Z| \geq 1$) in each recursive call. Due to Proposition~\ref{Proposition:Grandchild_of_a_Terminal_can_NOT_be_a_Terminal} and construction of the set~$S_x$, we have~$S_x \cap (T \cup \{x\}) = \emptyset$.
    Since the algorithm only stops when there is no $T$-cycle on~$x$, there is no $T$-cycle on~$x$ in the graph~$G-S_x$. Moreover, every $T$-cycle in~$G$ should intersect~$x$ as~$\{x\}$ is a~MWNS of~$(G, T)$. Thus~$S_x$ hits all~$T$ cycles of~$G$, and hence, the set~$S_x$ is a MWNS of~$(G, T)$ by Observation~\ref{Proposition:DifferentWaysOfLookingAtMWNS}.
    
    It remains to prove that~$S_x$ is a 14-approximate $x$-avoiding MWNS of~$(G, T)$. Towards this, assume that~\optx$(G, T)$ is the cardinality of a minimum~$x$-avoiding MWNS of~$(G, T)$. Then we need to show that~$|S_x| \leq 14 \cdot$\optx$(G, T)$. We prove this by induction on~$|V(G)|$.\\ 
    \textbf{Base Case: When~$|V(G)| \leq 2$.} 
    Since any~$T$-cycle should contain at least 2 vertices from~$G-\{x\}$, 
    there is no $T$-cycle in~$G$. Moreover, we have~$S_x = \emptyset$ as there is no $T$-cycle on~$x$. Hence Theorem~\ref{theorem:blocker} holds in this case.\\
    \textbf{Induction Hypothesis (IH):} Assume that for all instances~$(G', T, x)$ with~$G' \subseteq G$ and~$|V(G')| < |V(G)|$, we have~$|S'_x| \leq 14 \cdot$ \optx $(G', T)$, where~$S'_x$ is an $x$-avoiding MWNS of~$(G', T)$ computed by~$\mbox{Blocker}(G', T, x)$.\\
    \textbf{Induction Step:} Let~$Z \subseteq V(G)$ be the set computed by the first iteration of~$\mbox{Blocker} (G, T, x)$. If~$Z = \emptyset$ then Theorem~\ref{theorem:blocker} trivially holds. Hence assume that~$Z \neq \emptyset.$
    Since~$Z \subseteq S_x$ and~$S_x \setminus Z$ is the set computed by~$\mbox{Blocker}(G-Z, T, x)$, we have~$|S_x \setminus Z| \leq 14 \cdot \text{OPT}_x (G-Z, T)$ due to IH. Thus we have:
    \begin{equation}
        \begin{split}
            |S_x| &= |S_x \setminus Z| + |Z| 
            \leq 14 \cdot \text{OPT}_x (G-Z, T) + |Z| &&\text{due to IH}\\
            &\leq 14 \cdot \Bigl( \text{OPT}_x (G, T) - \frac{|Z|}{14} \Bigr) + |Z| &&\text{due to Lemma~\ref{Lemma:OPTDecreasesByAFactorOf14}}\\
            &= 14 \cdot \text{OPT}_x (G, T).
        \end{split}
    \end{equation} 
    This concludes the proof of Theorem~\ref{theorem:blocker}.
\end{proof}
            
\subsection{Safeness of reduction rules}
\label{section:SafenessOfRRs}

\begin{restatable}{lemma}{SafenessOfWeaklySeparableTerminalRRState} 
\label{Lemma:SafenessOfWeaklySeparableTerminalRR}
    If~\gtk reduces to~\gtminustk by Reduction Rule~\ref{RR:RemoveWeaklySeparatedTerminal}, then~\gtk is a YES-instance of~\mwnsshort if and only if~\gtminustk is a YES-instance of~\mwnsshort. Moreover, given a solution~$S'$ of~\gtminustk, in polynomial time we can obtain a solution~$S$ of~\gtk.
\end{restatable}

\begin{proof}
    In the forward direction, let~\gtk be a YES-instance of~\mwnsshort, and~$S \subseteq V(G) \setminus T$ is a solution of~\gtk. Observe that the same set~$S$ is also a solution of~\gtminustk as~$\tminust \subseteq T$. 
    In the backward direction, let~\gtminustk is a~YES-instance of~\mwnsshort, and~$S'\subseteq V(G) \setminus (\tminust)$ is an inclusion-wise minimal solution of~\gtminustk. Using the following claim we first show that the set~$S'$ does not contain any terminal from the set~$T$.
    \begin{claim}
        If~$S' \subseteq \tminust$ is an inclusion-wise minimal~solution of~\gtminustk then~$S' \cap T = \emptyset$.
    \end{claim}
    \begin{claimproof}
        Since~$S'$ is a solution of~\gtminustk, we have~$S' \subseteq \tminust$. Hence it remains to prove that~$t \notin S'$. Towards this, assume for a contradiction that~$t \in S'$. 
        Since~$t$ belongs to an inclusion-wise minimal solution~$S'$ of~\gtminustk, there exist 2 distinct terminals~$t_i, t_j \in \tminust$ such that the graph~$G -(S' \setminus \{t\})$ has 2~\ivd $t_i$-$t_j$ paths (say~$P_1, P_2$).
        For~$z \in [2]$, let~$P_z$ be a~$t_i$-$t_j$ path which intersects~$t \in S'$. Note that such a path~$P_z$ exists as~$S'$ is a solution of~\gtminustk.

        Next, we construct 2~\ivd $t$-$t_i$ paths~$P$ and~$P'$ from the paths~$P_z$ and~$P_{3-z}$ as follows.
        Let~$P = P_z [t_i, t]$ and~$P' = P_z[t, t_j] \cdot P_{3-z}$. 
        By construction, both~$P$ and~$P'$ are $t$-$t_i$ paths. Moreover, the paths~$P$ and~$P'$ are~\ivd because for distinct~$t_i, t_j, t$ the subpaths~$P_z[t_i, t]$ and~$P_z[t, t_j]$ of the simple~$t_i$-$t_j$ path~$P_z$ only intersect at~$t$, whereas the other~$t_i$-$t_j$ path~$P_{3-z}$ (which is~\ivd from~$P_z$ and thus also~\ivd from its subpaths~$P_z[t_i, t]$) only intersect with~$P_z[t_i - t]$ at~$t_i$. 
        Hence for~$t_i \in \tminust$ there are 2~\ivd $t_i$-$t$ paths in~$G$, a contradiction to the assumption that for any~$t'\in \tminust$, there do not exist 2~\ivd $t$-$t'$ paths in~$G$.
    \end{claimproof}
    Due to the above claim, we have~$S' \subseteq V(G) \setminus T$. Now we finish the proof of Claim~\ref{Lemma:SafenessOfWeaklySeparableTerminalRR} by showing that the same set~$S' \subseteq V(G) \setminus T$ is also a solution of~\gtk. Towards this assume for a contradiction that~$S'$ is not a solution of~\gtk. This implies that there exist distinct terminals~$t_i, t_j \in T$ such that there are 2 \ivd $t_i$-$t_j$ paths in~$G-S'$. Next, we do a case distinction based on whether or not~$t_i$ or~$t_j$ is equal to~$t$.

    \textbf{Case 1. When neither~$t_i$ nor~$t_j$ is equal to~$t$.}\newline
    In this case observe that both~$t_i, t_j \in \tminust$ such that there are 2 \ivd $t_i$-$t_j$ paths in~$G-S'$, a contradiction to the fact that~$S'$ is a solution of~\gtminustk.
    
    \textbf{Case 2. When either~$t_i = t$ or~$t_j = t$.}\newline
    Assume that~$t_j = t$, as the argument for the other case when~$t_i = t$ is symmetric. Since in this case, for~$t_i \in \tminust$, there exist 2 \ivd~$t_i$-$ t$ paths in~$G - S'$, the same paths are also resent in the supergraph~$G \supseteq G-S'$. This yields a contradiction to the fact that for any~$t'\in \tminust$, there do not exist 2~\ivd $t-t'$ paths in~$G$.
    Since both the above cases lead to a contradiction, our assumption that~$S'$ is not a solution of~\gtk is wrong, i.e., the set~$S'$ is a solution of~\gtk. Thus, \gtk is a Yes-instance of~\mwnsshort.   
    This concludes the proof of Lemma~\ref{Lemma:SafenessOfWeaklySeparableTerminalRR}.
\end{proof}

\begin{restatable}{lemma}{SafenessOfDegreeTwoPathBoundingRRState}
\label{Lemma:SafenessOfDegreeTwoPathBoundingRR}
    If~\gtk reduces to~$(G, T\setminus \{t\}, k)$ by Reduction Rule~\ref{RR:BoundingBlocksOfDegree2Path}, then~\gtk is a YES-instance of~\mwnsshort if and only if $(G, T\setminus\{t\}, k)$ is a YES-instance of~\mwnsshort. Moreover, given a solution~$S'$ of~\gtminustk, in polynomial time we can obtain a solution~$S$ of~\gtk.
\end{restatable}
\begin{claimproof}
   In the forward direction, assume that~\gtk is a YES-instance of~\mwnsshort, and $S \subseteq V(G) \setminus T$ is a solution of~\gtk. Note that the same set~$S$ is also a solution of~\gtminustk, because~$S \subseteq V(G) \setminus T \subseteq V(G)\setminus (T \setminus \{t\})$, $|S| \leq k$, and since~$S$ hits all~$T$-cycles and every~$T\setminus \{t\}$-cycle is also a~$T$-cycle, the set~$S$ hits all~$T\setminus \{t\}$-cycles.

   In the backward direction, assume that~$(G, T\setminus \{t\}, k)$ is a YES-instance of~\mwnsshort, and $S' \subseteq V(G) \setminus (T \setminus \{t\})$ is a solution of~$(G, T\setminus \{t\}, k)$.
   Next, depending on whether or not the set~$S'$ contains a vertex from the \cct~$D$, we construct a set~$S \subseteq V(G) \setminus T$ such that~$S$ is a solution of~\gtk.
        \begin{equation}
        S = \begin{cases}
        S' & \mbox{if $S' \cap V(D) = \emptyset$} \\
        (S' \setminus V(D)) \cup \{x\} & \mbox{otherwise.}
        \end{cases}
        \end{equation}
    Now we prove that~$S$ is a solution of~\gtk. Towards this, first note that by construction of~$S$, we have~$|S| \leq |S'| \leq k$. Moreover, as~$S'\subseteq V(G) \setminus (T \setminus \{t\})$, $t \in V(D)$, $V(D) \cap S = \emptyset$, and~$x \notin T$, we have~$S \subseteq V(G) \setminus T$. Thus, it remains to prove that the set~$S$ hits all~$T$-cycles. Towards proving this, assume for a contradiction that there exists a~$T$-cycle~$C$ in~$G - S$. 
    We now do a case distinction based on whether or not the~$T$-cycle~$C$ contains a vertex from the \cct~$D$.

    \textbf{Case 1. If~\boldmath{$V(C) \cap V(D) = \emptyset$.}}\newline
    Since~$V(C) \cap V(D) = \emptyset$ and $t \in V(D)$, the cycle~$C$ is also a~$T\setminus \{t\}$-cycle.
    Since~$C$ is preset in~$G-S$ and $V(C) \cap V(D) = \emptyset$, the $T\setminus \{t\}$-cycle $C$ does not intersect the set~$S \cup V(D) \supseteq S'$, a contradiction to the fact that~$S'$ is a solution of~\gtminustk.
    
    \textbf{Case 2. If~\boldmath{$V(C) \cap V(D) \neq \emptyset$}.}\newline
    Due to Condition~\ref{RR:BoundingBlocksOfDegree2Path:NoT-Cycle} there is no~$T$-cycle in~$G[V(D) \cup \{x, y\}]$, thus the~$T$-cycle~$C$ contains at least one vertex from~$V(G) \setminus (V(D) \cup \{x, y\})$. 
    Since~$D$ is a \cct in~$G-\{x, y\}$ and the simple cycle~$C$ contains at least one vertex from both~$V(D)$ and~$V(G) \setminus (V(D) \cup \{x, y\})$, we have~$x, y \in V(C)$.
    As~$x \in V(C)$ and the cycle~$C$ is present in~$G- S$, we have~$x \notin S$. Hence, by construction of~$S$, we have~$S = S'$ and $S' \cap V(D) = \emptyset$.
    Next, we construct a~$\tminust$-cycle~$C'$ avoiding the set~$S'$, yielding a contradiction to the fact that~$S'$ is a solution of~\gtminustk.
    Towards this, let~$P_c[x, y]$ be the path obtained by removing vertices of \cct~$D$ from the cycle~$C$, i.e., $P_c[x, y] := C \setminus V(D)$. Note that~$P_c$ is an $x-y$ path because~$C$ is a simple cycle that contains at least one vertex from both the sets~$V(D)$ and~$V(G)\setminus (V(D) \cup \{x, y\})$ and~$D$ is a \cct of~$G-\{x, y\}$.
    Let~$C' := P_c[x, y] \cdot P[x, y]$, where~$P$ is an~$x-y$ path in~$G[V(D) \cup \{x, y\}]$ containing 2 distinct terminals~$t_1, t_2 \in T\setminus \{t\}$ due to Condition~\ref{RR:BoundingBlocksOfDegree2Path:KeepingAPathP} of Reduction Rule~\ref{RR:BoundingBlocksOfDegree2Path}. 

    To prove~$V(C') \cap S' = \emptyset$, first note that~$P_c[x, y] \cap S' = \emptyset$, because by definition~$P_c [x, y] \subseteq C$,  we have~$V(C) \cap S = \emptyset$ as the cycle~$C$ is present in~$G-S$, and~$S = S'$ (as shown earlier). Moreover, as~$S' \cap V(D) = \emptyset$ and the path~$P(x, y)$ is present inside the \cct~$D$, we have~$P(x, y) \cap S' = \emptyset$.
    Overall, we have~$V(C') \cap S' = \emptyset$, a contradiction to the fact that~$S'$ is a solution of~$(G, T\setminus \{t\}, k)$. Since both the above cases lead to a contradiction, our assumption that~$S$ does not hit a~$T$-cycle is wrong, i.e., the set~$S$ hits all~$T$-cycles. Hence~$S$ is a solution of~\gtk. Hence~\gtk is a YES-instance of~\mwnsshort. Moreover, given a solution~$S'$ of~\gtminustk, it is easy to observe that the set $S$ defined above can be constructed in polynomial time.
    This concludes the proof of Lemma~\ref{Lemma:SafenessOfDegreeTwoPathBoundingRR}.
\end{claimproof}

Before proving the correctness of Reduction Rule~\ref{RR:Bounding_CCTs}, we note the following property about the reduction rule. 
        It follows from the fact that each component~$C^{x,y}_j$ marked for distinct~$x, y \in S^*$ contains an $x$-$y$ path containing a terminal. Since paths from distinct components are pairwise vertex-disjoint, any pair of such components yields a cycle containing two terminals of marked components, which therefore stay terminals in~$T'$. Since any set of size~$k$ must avoid two of the~$k+2$ marked components, to hit all these $T'$-cycles a solution must contain~$x$ or~$y$.
        
        \begin{observation}        \label{Observation:MarkingScheme_For-bounding_CCTs}
            Suppose~\gtk reduces to~$(G, T', k)$ by Reduction Rule~\ref{RR:Bounding_CCTs} and there exists a pair of distinct vertices~$x, y \in S^*$ for which there are~$k+2$ \ccts of~$G-S^*$ that are marked with respect to the pair~$x, y$ during Marking Scheme~\ref{MarkingScheme:Bounding_Interesting_CCTs}. Then any solution of~$(G, T', k)$ contains~$x$ or~$y$.
        \end{observation}

\begin{restatable}{lemma}{SafenessOfCCTsBoundingRRState}
        \label{Lemma:SafenessOfCCTBoundingRR}
        If~\gtk reduces to~$(G, T', k)$ by Reduction Rule~\ref{RR:Bounding_CCTs}, then~\gtk is a YES-instance of~\mwnsshort if and only if~$(G, T', k)$ is a YES-instance of~\mwnsshort. Moreover, given a solution~$S'$ of~$(G, T', k)$, in polynomial time we can obtain a solution~$S$ of~\gtk.
\end{restatable}
\begin{proof}
            In the forward direction, it is easy to note that if~\gtk is a YES-instance of \mwnsshort then~$(G, T', k)$ is also a YES-instance of~\mwnsshort as~$T' \subseteq T$.       
            In the reverse direction, assume that~$(G, T', k)$ is a YES-instance of~\mwnsshort and let~$S' \subseteq V(G) \setminus T'$ be an inclusion-wise minimal solution of $(G, T', k)$.
            Using the following claim we first show that any inclusion-wise minimal solution (and hence~$S'$) of~$(G, T', k)$ cannot contain any terminal from~$T$, i.e., we have~$S'\subseteq V(G \setminus T$.

            \begin{claim}
                If~$S'$ is an inclusion-wise minimal solution of~$(G, T', k)$ then~$S' \cap T = \emptyset$.
            \end{claim}
            \begin{claimproof}
                Assume for a contradiction that there exists a terminal~$t \in T$ such that~$t \in S'$. Since~$S'$ is a solution of~$(G, T', k)$, we have~$S' \subseteq V(G) \setminus T'$. Hence~$t \in T \setminus T'$, i.e., the terminal~$t$ is contained in an unmarked \cct of~$G-S^*$, where~$S^*$ is the 1-redundant set with respect to which we have done marking in Reduction Rule~\ref{RR:Bounding_CCTs}. Since~$t$ belongs to an inclusion-wise minimal solution $S'$ of~$(G, T', k)$, there exists a pair of distinct terminals~$t_i, t_j \in T'$ such that the graph~$G - (S' \setminus \{t\})$ has two \ivd $t_i$-$t_j$ paths~$P_1$ and~$P_2$. Let~$P \in \{P_1, P_2\}$ be the $t_i$-$t_j$ path which passes trough~$t \in S'$ and does not contain any vertex from~$S' \setminus \{t\}$, i.e., we have $V(P) \cap S' = \{t\}$.

                Next, we observe how the path~$P$ appears with respect to the decomposition of~$G$ into~$S^*$ and~$G-S^*$. Since~$S^* \subseteq V(G) \setminus T$, we have~$t \in V(G-S^*)$. Let~$C_t$ be the \cct of~$G-S^*$ which contains~$t$. Since~$t \in T \setminus T'$, the \cct $C_t$ is an unmarked \cct, thereby neither~$t_i \in T'$ nor~$t_j \in T'$ belong to the \cct $C_t$ of~$G-S^*$.
                Observe that the subpath~$P[t, t_i]$ must intersect the set~$S^*$, because if~$P[t, t_i]$ is completely contained in~$G-S^*$ then~$t$ and~$t_i$ belong to the same \cct~$C_t$, a contradiction to the fact that~$t_i \in T'$.
                Similarly, the subpath~$P[t, t_j]$ must intersect the set~$S^*$.
                Let~$x \in V(P[t, t_i])$ be the fist vertex of the subpath~$P[t, t_i]$ (while traversing from $t$ to $t_i$) that does not belong to the \cct~$C_t$ of~$G-S^*$.
                Let~$y \in V(P[t, t_j])$ be the fist vertex of the subpath~$P[t, t_j]$ (while traversing from $t$ to $t_j$) that does not belong to the \cct~$C_t$ of~$G-S^*$.
                By construction of~$x$ and~$y$, we have~$x, y \in S^*$ and~$x \neq y$. Moreover, the subpaths~$P[t, x]$ and~$P[t, y]$ of the simple~$t_i$-$t_j$ path~$P$ are \ivd.
                Observe that the \cct~$C_t$ satisfies all the conditions of Marking Scheme~\ref{MarkingScheme:Bounding_Interesting_CCTs} with respect to the pair~$x, y \in S^*$, whereas~$C_t$ is not marked, implying that the marking scheme has already marked~$k+2$ \ccts with respect to the pair~$x, y$. Hence, due to Observation~\ref{Observation:MarkingScheme_For-bounding_CCTs}, the solution $S'$ of~$(G, T', k)$ contains either~$x$  or~$y$, a contradiction to the assumption that~$P$ is a path in~$G-(S' \setminus \{t\})$.
            \end{claimproof}

            Next, we show that the set~$S' \subseteq V(G) \setminus T$ is also a solution of~\gtk, thereby proving that~\gtk is a YES-instance of \mwnsshort.
            Towards this, assume for a contradiction that~$S'$ is not a solution of~\gtk. This implies that there exists a pair of distinct terminals~$t_i, t_j \in T$ such that there are two \ivd $t_i$-$t_j$ paths (say~$P_1, P_2$) in~$G-S'$. Observe that at least one of~$t_i$ or~$t_j$ must belong to~$T \setminus T'$, because if both~$t_i, t_j \in T'$ then the two \ivd $t_i$-$t_j$ paths of~$G-S'$ yield a contradiction to the fact that~$S'$ is a solution of~$(G, T', k)$. W.l.o.g. (by symmetry) let~$t_j \in T \setminus T'$.
            Next, we look at how the~$t_i$-$t_j$ paths~$P_1, P_2$ appear when we decompose~$G$ into~$S^* \subseteq V(G) \setminus T$ (the set with respect to which we have done marking) and~$G-S^*$.
            Let~$C_j$ be the \cct of~$G-S^*$ which contains~$t_j$. We do a case distinction based on how many paths among~$P_1, P_2$ intersect~$S^*$.

            \subparagraph{Case 1. When exactly one of~$P_1$ or~$P_2$ intersects the set~$S^*$.} For~$z \in [2]$, let~$P_z$ be the path which intersects~$S^*$. Since in this case the other path~$P_{3-z}$ is present in~$G_S^*$, both~$T_i, t_j$ belong to the same \cct $C_j$ of~$G_S^*$.
            Let~$x, y$ be the first and the last vertex of the~$t_i$-$t_j$ path~$P_z$ in~$S^*$ while traversing from~$t_i$ to~$t_j$, respectively. Observe that~$x \neq y$ because~$S^*$ is a 1-redundant MWNS of~$(G, T)$ and the other path~$P_{3-z}$ does not intersect~$S^*$.
            Let~$p$ be the predecessor of~$x$ and~$q$ be the successor of~$y$ on the path~$P_z$ while traversing from~$t_i$ to~$t_j$. By definition, we have~$p \in N_G(x) \cap V(C_j)$ and~$q \in N(y) \cap V(C_j)$. Moreover, there is a~$p$-$q$ path~$P_j := P_z[p, t_i] \cdot P_{3-z}[t_i, t_j]$ in~$C_j$ which contains~$t_i, t_j$.
            Hence the \cct~$C_j$ satisfies all the conditions of Marking Scheme~\ref{MarkingScheme:Bounding_Interesting_CCTs}, whereas~$C_j$ is not marked. This implies that the marking scheme has already marked~$k+2$ \ccts of~$G_S^*$ with respect to the pair~$x, y \in S^*$. Hence, by Observation~\ref{Observation:MarkingScheme_For-bounding_CCTs}, either~$x \in S'$ or~$y \in S'$, a contradiction to the fact that~$P_1, P_2$ are preset in~$G-S'$.
            
            \subparagraph{Case 2. When both~$P_1, P_2$ intersect the set~$S^*$.} Let~$x$ and~$y$ be the first vertices (while traversing from~$t_j$ to~$t_i$) of the $t_i$-$t_j$ paths~$P_1$ and~$P_2$, respectively that is not present inside the \cct~$C_j$. The subpath~$P_1 [x, t_j]$ and~$P_2 [y, t_j]$ implies that with respect to the pair~$x, y \in S^*$, the \cct~$C_j$ satisfies all the properties of Marking Scheme~\ref{MarkingScheme:Bounding_Interesting_CCTs}, whereas~$C_j$ is not marked. This implies that the marking scheme has already marked~$k+2$ \ccts of~$G_S^*$ with respect to the pair~$x, y \in S^*$. Hence, by Observation~\ref{Observation:MarkingScheme_For-bounding_CCTs}, either~$x \in S'$ or~$y \in S'$, a contradiction to the fact that~$P_1, P_2$ are preset in~$G-S'$.

        Since we show that both the above cases lead to a contradiction, our assumption that~$S'$ is not a solution of~\gtk is wrong. Thus~$S'$ is a solution of~\gtk, implying~\gtk is  a YES-instance of \mwnsshort. This completes the proof of Lemma~\ref{Lemma:SafenessOfCCTBoundingRR}.
\end{proof}

\subsection{Construction of 1-redundant set}
\label{subsection:CostructionOfRedundantSet}
\begin{restatable}{lemma}{CostructionOfRedundantSetState}
\label{lemma:ConstructionOfRedundantSet}
There is a polynomial-time algorithm that, given an instance~\gtk of \mwnsshort and a MWNS~$\hat{S}$ of~$(G, T)$, outputs an equivalent instance~$(G_1, T, k_1)$ and a 1-redundant MWNS~$S^*$ of~$(G_1, T)$ such that: (a) $|S^*| \leq (14k+1) |\hat{S}|$, (b) $G_1$ is an induced subgraph of~$G$, and (c) $k_1 \leq k$.
Moreover, there is a polynomial-time algorithm that, given a solution~$S_1$ for~$(G_1, T, k_1)$, outputs a solution~$S$ of~\gtk.
\end{restatable}

\begin{proof}
    We describe an algorithm to construct a 1-redundant MWNS~$S^*$. During its execution, we may identify vertices that belong to every solution of~\gtk, which we use to reduce the graph. The algorithm starts by initializing~$S^* := \emptyset$, and a set~$X := \emptyset$ which will store vertices belonging to every solution of~\gtk.
    For each vertex~$x \in \hat{S}$, we invoke the algorithm of Theorem~\ref{theorem:blocker} with the graph~$G' = G- (\hat{S} \setminus {x})$, terminal set~$T$, and vertex~$x$. In polynomial time, it outputs a set~$S_x \subseteq V(G') \setminus (T \cup \{x\})$ such that~$S_x$ is a MWNS of~$(G', T)$ and~$|S_x| \leq$ 14 \optx$(G, T)$, where~\optx$(G, T)$ is the cardinality of a minimum~$x$-avoiding MWNS of~$(G', T)$.
    \begin{itemize}
        \item If~$|S_x| > 14 k$, add~$x$ to~$X$. 
        \item Otherwise, add~$S_x \cup \{x\}$ to~$S^*$.
    \end{itemize}
    We set~$G_1 := G-X$ and~$k_1 := k - |X|$ and output the instance~$(G_1, T, k_1)$ with the 1-redundant MWNS~$S^*$. This completes the description of the algorithm, which runs in polynomial time via Theorem~\ref{theorem:blocker}. 
    Note that the set~$X$ constructed above is contained in every solution of~\gtk: because for each vertex~$x \in X$, we have~$|S_x| > 14k$ 
    and since~$S_x$ is a 14-approximate $x$-avoiding MWNS of~$(G', T)$, every~$k$-sized MWNS of~$(G', T)$ must contain~$x$ and hence every solution of~\gtk contains~$x$ as~$G \supseteq G'$. 
    By the construction of~$(G_1, T, k_1)$, it is easy to observe that~$G_1$ is an induced subgraph of~$G$ and~$k_1 \leq k$. 
    Moreover, since~$X$ is contained in every solution of~\gtk, we have~$(G_1, T, k_1)$ is a YES instance of \mwnsshort if and only if~\gtk is a YES-instance of \mwnsshort. Also, given a solution~$S_1$ of~$(G_1, T, k_1)$, we can obtain a solution~$S$ of~\gtk by setting~$S:= S_1 \cup X$.
    Since for each vertex~$x' \in \hat{S} \setminus X$, we add both~$S_{x'}$ and~$x'$ to~$S^*$, the set~$S^*$ is a 1-redundant MWNS of~$(G_1, T)$.
    Moreover, for each~$x'$, we have~$|S_{x'}| \leq 14k$, we have ~$|S^*| \leq (14k +1) |\hat{S}|$.
    This concludes the proof of Lemma~\ref{lemma:ConstructionOfRedundantSet}.
\end{proof}

\section{Proof of Theorem~\ref{theorem:terminalbounding}}
\label{sec:ProofOfTerminalBoundingThm}
In order to prove Theorem~\ref{theorem:terminalbounding}, we take an instance that is reduced with respect to our reduction rules and consider the block-cut forest of the graph that remains when removing the multiway near-separator~$\hat{S}$ that is given in the context of the theorem. As the problem definition ensures that each such block contains at most one terminal, to obtain a bound on the number of terminals it suffices to bound the number of blocks that contain a terminal. The next two subsections are devoted to analyzing the number of blocks containing a terminal in two different scenarios: in Section~\ref{subsection:BoundingInterestingLeaves} we bound the number of blocks that contain a terminal and are at 'maximal depth' in a certain technical sense, while in Section~\ref{subection:BoundingTerminalsOnDegree2Path} we consider blocks occurring on a root-to-leaf path in the block-cut tree. Then we combine the results of these sections (along with the bound on the cardinality of a 1-redundant set) to prove Theorem~\ref{theorem:terminalbounding} in Section~\ref{section:WrappingUp}.

\subsection{\texorpdfstring{A bound on the number of blocks from the set~$\calL^*$}{Bounding interesting leaves}}
\label{subsection:BoundingInterestingLeaves}
In this section, we show that given a non-trivial instance~\gtk of~\mwnsshort, and a 1-redundant~MWNS~$S^* \subseteq V(G) \setminus T$ of~$(G, T)$, if none of the reduction rules discussed in the previous section are applicable on~\gtk with respect to the set~$S^*$, then we can bound the number of blocks of a rooted block-cut forest~$\calF^*$ of~$G-S^*$ which contain a terminal and are deepest in a technical sense, by~$\Oh(k \cdot |S^*|^3)$. The following definition captures the relevant set of blocks~$\mathcal{L}^*$.

\begin{definition}[Definition of $\calL^*$]
\label{Definiton:InterestingLeaves}
    Given an instance~\gtk of \mwnsshort and a 1-redundant MWNS $S^* \subseteq V(G)$ of~$(G, T)$, let~$\mathcal{F}^*$ be a rooted block-cut forest of~$G-S^*$. 
    For each terminal~$t \in T$, let~$B_t$ be the uppermost block of the block-cut forest~$\mathcal{F}^*$ containing~$t$. Let~$\mathcal{B}^* = \bigcup_{t \in T} \{B_t\}$.
    Let~$\calL^* \subseteq \mathcal{B}^*$ be the subset containing block~$L$ of~$\mathcal{B}^*$ if no proper descendent of~$L$ (in~$\mathcal{F}^*$) belongs to~$\mathcal{B}^*$.
\end{definition} 

By Observation~\ref{observation:property_of_blocks}, for each terminal~$t\in T$, the unique upper-most block containing~$t$ does not contain any other terminals. The reduction rules and definition of a 1-redundant MWNS allow us to prove the following bound on~$|\mathcal{L}^*|$. 

\begin{restatable}{lemma}{NumberofInterestingLeavesIsBoundedState}
\label{Lemma:NumberofInterestingLeavesIsBounded}
    Let~\gtk be a non-trivial instance of~\mwnsshort, and let~$S^* \subseteq V(G)$ be a 1-redundant MWNS of~$(G, T)$. Let~$\calF^*$ be a rooted block-cut forest of~$G-S^*$. 
    If none of
    Reduction Rule~\ref{RR:RemoveWeaklySeparatedTerminal}, Reduction Rule~\ref{RR:BoundingBlocksOfDegree2Path}, and Reduction Rule~\ref{RR:Bounding_CCTs} (with respect to the 1-redundant set~$S^*$) are applicable on~\gtk, then the number of blocks from the set~$\calL^*$ (defined above in Definition~\ref{Definiton:InterestingLeaves}) in the block-cut forest~$\calF^*$ is bounded by~$\Oh (k \cdot |S^*|^3)$.
\end{restatable}

Before proving that the number of blocks from~$\calL^*$ is bounded in the whole block-cut forest~$\calF^*$ (Lemma~\ref{Lemma:NumberofInterestingLeavesIsBounded}), we prove the following lemma which says that the number of blocks from~$\calL^*$ in a single block-cut tree~$\calT^*$ of~$\calF^*$ is bounded.
\begin{restatable}{lemma}{InterestingLeavesInaTreeIsBoundedState}
\label{Lemma:NumberofInterestingLeavesInATreeIsBounded}
     Let~\gtk be a non-trivial instance of~\mwnsshort, and let~$S^* \subseteq V(G)$ be a 1-redundant MWNS of~$(G, T)$. Let~$\calF^*$ be a rooted block-cut forest of~$G-S^*$. If none of
    Reduction Rule~\ref{RR:RemoveWeaklySeparatedTerminal}, Reduction Rule~\ref{RR:BoundingBlocksOfDegree2Path}, and Reduction Rule~\ref{RR:Bounding_CCTs} (with respect to 1-redundant set~$S^*$) are applicable on~\gtk then the number of blocks from the set~$\calL^*$ (defined in Definition~\ref{Definiton:InterestingLeaves}) in every block-cut tree~$\calT^* \in \calF^*$ is at most~$|S^*|$.
\end{restatable}

\begin{proof}
    Assume for a contradiction that there exists a block-cut tree~$\calT^*$ of~$\mathcal{F}^*$ such that the number of blocks from~$\calL^*$ in~$\calT^*$ is more than~$|S^*|$.
    First, we prove using the following claim that for each block~$L \in \calL^*$ of~$\calT^*$ containing terminal~$t$, any block of the subtree~$\calT^*_L$ cannot contain a  terminal other than~$t$.
    \begin{claim}\label{Claim:NoTerminalBelowInterestingBlock}
        Let~$L \in \calL^*$ be a block of~$\calT^*$ with~$t \in V_G(L)$ for some~$t \in T$. Then we have,~$V_G(\calT^*_L) \cap T = V_G(L) \cap T = \{t\}$.
    \end{claim}
    \begin{claimproof}
        First note that each block of~$\calT^*$ contains at most one terminal due to Observation~\ref{Observation:Block_contains_atmost_1-terminal} and the fact that~$S^*$ is a MWNS of~$(G, T)$. Hence we have~$V(L) \cap T = \{t\}$.
        Since block~$L$ is present in the subtree~$\calT^*_L$, we have~$t \in V_G(\calT^*_L)$, i.e., $V_G(\calT^*_L) \cap T \supseteq \{t\}$.
        Assume for a contradiction that~$V_G(\calT^*_L) \cap T \supsetneq \{t\}$, implying that there exists another terminal~$t'\neq t$ such that~$t'\in V_G (\calT^*_L)$. Note that~$t' \neq V(L)$ as block~$L$ already contains terminal~$t$.
        Thus the terminal~$t'$ is present in some block~$B'$ of the subtree~$\calT^*_L$.
        Let~$B_{t'}$ be the uppermost block of the block-cut forest~$\mathcal{F}^*$ containing~$t'$. Note that~$B_{t'} \neq L$ as~$t' \notin V(L)$.
        Moreover, the block~$B_{t'}$ cannot belong to the subtree~$\calT^*_L$ because if~$B_{t'}$ is present in the subtree~$\calT^*_L$ then we get a contradiction to the fact that~$L \in \calL^*$.
        Thus the block~$B_{t'}$ is present in~$\mathcal{F}^* - \calT^*_L$. Note that in this case the distance between block~$B'$ of the subtree~$\calT^*_L$ and block~$B_{t'}$ of~$\mathcal{F}^* - \calT^*_L$ is at least 3 because they are separated by block~$L$, a contradiction to Observation~\ref{observation:property_of_blocks}.
    \end{claimproof}
    
    Next we show that for each block~$L \in \calL^*$ of~$\calT^*$, there is a block~$B$ in subtree~$\calT^*_L$ containing a vertex~$w \in N(S^*) \setminus \{\mbox{parent}_{\calT^*}(L)\}$.
    \begin{claim}\label{Claim:EachInterestinLeafSeesAVertexinS*}
        Let~$L \in \calL^*$ be a block of~$\calT^*$ containing terminal~$t$. Then the subtree~$\calT^*_L$ contains a vertex~$w \in V_G(\calT^*_L) \cap N(S^*)$ such that there is a~$t$-$w$ path in~$G[V_G(\calT^*_L) \setminus \{p\}]$, where~$p$ is the parent of block~$L$ in~$\calT^*$.       
    \end{claim}
    \begin{claimproof}
        Since Reduction Rule~\ref{RR:RemoveWeaklySeparatedTerminal} is no longer applicable on~\gtk, whereas the block~$L \in \calL^*$ of the block-cut tree~$\calT^*$ contains a terminal~$t$, implying that there exists another terminal~$t' \in \tminust$ such that there are 2~\ivd~$t$-$t'$ paths say~$P_1$ and $P_2$ in graph~$G$.
        Next, we look at how the 
        paths~$P_1$ and~$P_2$ are present with respect to decomposition~$\calF^*$ and~$S^*$ of~$G$.

        First, note that the block~$L \in \calL^*$ cannot be the root~$\calT^*$, because if~$L$ is the root then by Claim~\ref{Claim:NoTerminalBelowInterestingBlock}, we know that there is no other terminal than~$t$ in~$V_G(\calT^*)$ and hence there is no block from~$\calL^*$ in~$\calT^*$ except the block~$L$, a contradiction to our assumption that~$\calT^*$ has more than~$|S^*| \geq 1$ (note that~$|S^*| \geq 1$ for a non-trivial instance of~\mwnsshort) blocks from~$\calL^*$. Let~$p$ be the parent of block~$L$ in~$\calT^*$. Observe that~$p \notin T$: First, it is easy to note that~$p \notin \tminust$ because both~$p, t$ belong to block~$L$ and a block can have at most one terminal. Secondly, we have~$p \neq t$ as the block~$L$ is the uppermost block of~$\calT^*$ containing~$t$ by definition of~$\calL^*$.
        Note that the terminal~$t'$ (defined above) does not belong to~$V_G(\calT^*_L) \ni p$ due to Claim~\ref{Claim:NoTerminalBelowInterestingBlock}.
        Moreover, the terminals~$t$ and~$t'$ belong to different \ccts of~$G - (S^* \cup \{p\})$ due to Observation~\ref{observation:propertyofacutvertex}, the fact that~$p \neq t'$, and~$S^* \subseteq V(G) \setminus T$.
       
        Since the terminals~$t$ and~$t'$ belong to different \ccts of~$G - (S^* \cup \{p\})$ and there are 2 \ivd~$t$-$t'$ paths~$P_1$ and~$P_2$ in~$G$, both the paths~$P_1$ and~$P_2$ intersect the set~$S^* \cup \{p\} \subseteq V(G) \setminus T$. 
        Note that at most one of~$P_1$ or~$P_2$ can intersect the vertex~$p$ as~$P_1$ and~$P_2$ are \ivd $t$-$t'$ paths. Hence there is at least one~$t$-$t'$ path say~$P \in \{P_1, P_2\}$ which intersects a vertex of~$S^*$ and avoids the vertex~$p$. 
        Let~$z$ be the first vertex of~$S^*$ on path~$P$ while traversing from~$t$ to~$t'$. 
        Let~$w$ be the predecessor of~$z$. By construction, the subpath~$P[t, w]$ of the~$t$-$t'$ path~$P$ does not intersect~$S^* \cup \{p\}$.

        Next, we look at the subpath~$P[t, z]$ of the path~$P$ with respect to the decomposition~$\calF^*$ and~$S^*$ of~$G$.
        By definition of block~$L$, we know~$t \in V(L)$. Moreover, we have~$t \notin V_G(\calF^* - \calT^*_L)$ due to Observation~\ref{observation:property_of_blocks} and the fact that~$p \neq t$.
        Hence the subpath~$P[t, w]$ is present in graph~$G[V_G(\calT^*_L) \setminus \{p\}]$ because~$t \notin V_G (\calF^* - \calT^*_L)$, the subpath~$P[t, w]$ does not intersect with the vertex set~$S^* \cup \{p\}$, and Observation~\ref{observation:propertyofacutvertex}.
        
        There is an edge~$\{w, z\}$ for~$z \in S^*$.
        This concludes the proof of Claim~\ref{Claim:EachInterestinLeafSeesAVertexinS*}.
    \end{claimproof}

    Since Claim~\ref{Claim:EachInterestinLeafSeesAVertexinS*} holds for each block~$L \in \calL^*$ and there are more than~$|S^*|$ blocks of~$\calL^*$ in the block-cut tree~$\calT^*$, by pigeon-hole principle there exists a vertex~$z^* \in S^*$ and two distinct blocks~$L_1, L_2 \in \calL^*$ in~$\calT^*$ containing distinct terminals~$t_1, t_2$, respectively, such that for each~$i \in [2]$, there exist~$w_i \in N(z^*)$ with $t_i$-$w_i$ path~$P_i$ in the graphs~$G[V_G(\calT^*_{L_i}) \setminus \{p_i\}]$, where~$p_i$ is the parent of block~$L_i$ in~$\calT^*$.
    
    For each~$i \in [2]$, let~$q_i \in V(L_i)$ be the last vertex of the~$t_i$-$w_i$ path~$P_i$ while traversing from~$t_i$ to~$w_i$. 
    Next, we apply proposition~\ref{Proposition:ThereAre2PathsInsideBlock} with~$B = L_i, t= t_i, p = p_i$, and~$q = q_i$ to obtain~$p_i$-$t_i$ path~$P^i$ and~$q_i$-$t_i$ path~$Q^i$ such that~$V(P^i) \cap V(Q^i) = \{t_i\}$.

    Moreover, since the blocks~$L_1$ and~$L_2$ belong to the same block-cut tree, there is a simple~$p_1$-$p_2$ path~$P_{12}$ in~$G$ which avoids vertices of the subtrees~$\calT^*_{L_1}$ and~$\calT^*_{L_2}$ except the vertices~$p_1, p_2$.
    Next, we obtain a $T$-cycle~$C_{12}$ as follows.
    $C_{12} := Q^1 [t_1, q_1] \cdot P_1 [q_1, w_1] \cdot \{w_1, z^*\} \cdot\{z^*, w_2\} \cdot P_2 [w_2, q_2] \cdot Q^2 [q_2, t_2] \cdot P^2 [t_2, p_2] \cdot P_{12} [p_2, p_1] \cdot P^1 [p_1, t_1]$. 
    Note that~$C_{12}$ is a $T$-cycle containing exactly one vertex~$z^*$ from the set~$S^*$, a contradiction to the fact that the set~$S^*$ is a 1-redundant MWNS of~$(G, T)$.
    This completes the proof of Lemma~\ref{Lemma:NumberofInterestingLeavesInATreeIsBounded}.
\end{proof}

Now we are equipped to prove Lemma~\ref{Lemma:NumberofInterestingLeavesIsBounded}.

\NumberofInterestingLeavesIsBoundedState*
\begin{proof}
Since Reduction Rule~\ref{RR:Bounding_CCTs} is not applicable on the \mwnsshort instance~\gtk with respect to the 1-redundant set~$S^*$, the number of connected components of~$G-S^*$ containing at least one terminal from~$T$ is bounded by~$\Oh (|S^*|^2 \cdot k)$. 
Moreover, due to Lemma~\ref{lemma:BlockGraphOfCCT}, each \cct~$C_i^*$ of~$G-S^*$ corresponds to a block-cut tree~$\calT_i^*$ in the block-cut forest~$\mathcal{F}^*$ of~$G-S^*$. Thus the number of block-cut trees in~$\mathcal{F}^*$ whose block contains a terminal is also bounded by~$\Oh (|S^*|^2 \cdot k)$.
Also, for any block-cut tree~$\calT^*$ of~$\mathcal{F}^*$, the number of blocks form~$\calL^*$ in~$\calT^*$ is bounded by~$|S^*|$ due to Lemma~\ref{Lemma:NumberofInterestingLeavesInATreeIsBounded}.
Since the number of block-cut trees in~$\mathcal{F}^*$ whose block contains a terminal is bounded by~$\Oh (|S^*|^2 \cdot k)$ and each block cut tree of~$\calF^*$ can have at most~$|S^*|$ blocks from~$\calL^*$, the number of blocks from~$\calL^*$ in the block-cut forest~$\calF^*$ is bounded by~$\Oh(|S^*|^3 \cdot k)$.
This concludes the proof of Lemma~\ref{Lemma:NumberofInterestingLeavesIsBounded}.
\end{proof}

\subsection{A bound on the number of terminals on a path}
\label{subection:BoundingTerminalsOnDegree2Path}
In this section, we show that given a non-trivial instance~\gtk of~\mwnsshort, and a 1-redundant MWNS~$S^* \subseteq V(G)$ of~$(G, T)$, if none of the reduction rules discussed in Section~\ref{sec:bounding_terminals} are applicable on~\gtk with respect to the set~$S^*$, 
for any rooted block-cut tree~$\calT^*$ of $G-S^*$, the number of blocks containing a terminal on a path from the root to a leaf of~$\calT^*$ is bounded by~$\Oh(|S^*|)$.

\begin{restatable}{lemma}{TerminalsOnDegreeTwoPathIsBoundedState}
\label{lemma:TerminalsOnDegreeeTwoPathIsBounded}
    Let~\gtk be a non-trivial instance of~\mwnsshort, and let~$S^* \subseteq V(G)$ be a 1-redundant MWNS of~$(G, T)$. Let~$\calF^*$ be a rooted (at blocks) block-cut forest of~$G-S^*$. Let~$\calT^*$ be a block-cut tree of~$\calF^*$, and let~$\calP$ be a path from the root to a leaf in~$\calT^*$.
If none of Reduction Rule~\ref{RR:RemoveWeaklySeparatedTerminal}, Reduction Rule~\ref{RR:BoundingBlocksOfDegree2Path}, and Reduction Rule~\ref{RR:Bounding_CCTs} (with respect to the 1-redundant set~$S^*$) are applicable on~\gtk then the number of blocks containing a terminal on the path~$\calP$ is less than~$9 \cdot (2|S^*| + 1)$.
\end{restatable}
\begin{proof}
Assume for a contradiction that there exists a tree~$\calT^* \in \calF^*$ such that there is a path~$\calP$ from the root to a leaf in~$\calT^*$ such that the number of blocks containing a terminal on~$\calP$ is at least~$9 \cdot (2|S^*| +1)$. Let~$\calP = B_1, c_1, B_2, c_2, \ldots, c_{\ell-1} B_\ell$, where $B_1$ is the root and for each~$i \in [\ell-1]$, node~$c_i$ is the cutvertex between blocks~$B_i$ and~$B_{i+1}$. Let~$\calB$ be the set containing blocks of~$\calP$ which contain a terminal.
Then we have~$|\calB| \geq 9 \cdot (2|S^*|+1)$.

    Next, we partition~$\calP$ (from top to bottom) into segments~$S_1, S_2, \ldots, S_{\ell'}$ such that for each~$i \in [\ell' -1]$, the segment~$S_i$ contains exactly 9 blocks from~$\calB$, whereas~$S_{\ell'}$ contains at most 9 blocks from~$\calB$. We call a segment~$S_i$ \emph{complete} if~$S_i$ has exactly 9 blocks from~$\calB$. Observe that every complete segment has at least 5 distinct terminals because there are 9 blocks containing a terminal in a complete segment, and a terminal can belong to at most 2 (consecutive) blocks of~$\calP$ due to Observation~\ref{observation:property_of_blocks} and the fact that~$\calT^*$ is a tree. 
    Moreover, since there are at least~$9 \cdot (2 |S^*| +1)$ blocks from~$\calB$ on~$\calP$ and each segment contains at most 9 blocks from~$\calB$, 
    there are at least~$(2 |S^*| + 1)$ complete segments. 
    Next, using the following claim we show that each complete segment~$S_i$ has a node~$y_i$ such that the set $V_G(\calT^*_{y_i} - \calT^*_{z_i})$ contains a vertex~$w_i \in N_G(S^*)$, where~$z_i$ is the node adjacent to~$y_i$ on the path~$\calP$ while traversing from top to bottom. 
    
    \begin{restatable}{claim}{EachSegementSeesAVertexOfSstarState} 
    \label{Claim:EachSegementSeesAVertexOfSolution}
         Each complete segment~$S_i$ contains a node~$y_i$ such that~$V_G(\calT^*_ {y_i} - \calT^*_{z_i}) \cap N_G (S^*) \neq \emptyset$, 
        where~$z_i$ is the node adjacent to~$y_i$ on the path~$\calP$ while traversing from top to bottom (if there is no~$z_i$ then $\calT^*_{z_i}$ is an empty subtree).
    \end{restatable}
    \begin{claimproof}
        Assume for a contradiction that there exists a complete segment~$S_i$ such that~$S_i$ does not contain a node~$y_i$ such that~$V_G(\calT^*_{y_i} - \calT^*_{z_i}) \cap N_G (S^*) \neq \emptyset$.
        Consider the subpath~$ \calP [B_q, B_r] := B_q, c_q, \ldots, c_{r-1}, B_r$ of segment~$S_i$ containing all blocks of~$S_i$ along with the cut-vertices between its blocks.
        
        First note that at least one of the cutvertices among~$c_q$ and~$c_{q+1}$ must be a non-terminal, because by definition of the block-cut tree both the cutvertices~$c_q, c_{q+1}$ belong to the block~$B_{q+1}$, and the block~$B_{q+1}$ of~$G-S^*$ can contain at most one terminal due to Observation~\ref{Observation:Block_contains_atmost_1-terminal}. 
        By the same argument, one of the cutvertices among~$c_{r-1}$ and~$c_{r-2}$ must be a non-terminal.
        Let~$x \in \{c_q, c_{q+1}\} \setminus T$ and~$y \in \{c_{r-1}, c_{r-2}\} \setminus T$. 
        Since there are 9 blocks from the set~$\calB$ in the (complete) segment~$S_i$ and independent of whether~$x = c_q$ or~$x = c_{q+1}$ and~$y = c_{r-1}$ or~$y \in c_{r-2}$, we still have at least 5 blocks from the set~$\calB$ on the subpath~$\calP(x, y)$. Since a terminal can appear in at most 2 distinct blocks of~$\calP$ (due to Observation~\ref{observation:property_of_blocks}) and there are at least 5 distinct blocks from~$\calB$ (each containing a terminal) on the subpath~$\calP(x, y)$, there are at least 3 distinct terminals in the graph~$D := G[V_G (\calP (x, y)) \setminus \{x, y\}]$.

        Moreover, since the segment~$S_i$ does not contain a node~$y_i$ such that~$V_G(\calT^*_{y_i} - \calT^*_{z_i}) \cap N_G (S^*) \neq \emptyset$ and the subpath~$\calP(x, y)$ is contained in segment~$S_i$, therefore~$D$ is a \cct in~$G - \{x, y\}$.
        Also, since a block of~$G-S^*$ can contain at most one terminal (due to Observation~\ref{Observation:Block_contains_atmost_1-terminal}) and there are 3 distinct terminals on the subpath~$\calP(x, y)$, the subpath~$\calP(x, y)$ contains distinct blocks~$B_1, B_2$ containing distinct terminals~$t_1, t_2$, respectively. Then due to Observation~\ref{Observation:Degree2Path}, there is a $x$-$y$ path in~$G[V_G(\calP[x, y])]$ containing distinct terminals~$t_1, t_2$.
        Note that we have two distinct non-terminal vertices~$x, y \in V(G) \setminus T$ such that the graph~$G - \{x, y\}$ contains a~\cct $D$ satisfying all the conditions of Reduction Rule~\ref{RR:BoundingBlocksOfDegree2Path}, a contradiction to the fact that Reduction Rule~\ref{RR:BoundingBlocksOfDegree2Path} is no longer applicable on~\gtk. This concludes the proof of Claim~\ref{Claim:EachSegementSeesAVertexOfSolution}.
    \end{claimproof}

    Since there are at least~$2 |S^*| + 1$ complete segments, and due to the above claim each complete segment contains a node~$y_i$ such that~$V_G(\calT^*_ {y_i} - \calT^*_{z_i}) \cap N_G (S^*) \neq \emptyset$, by the pigeonhole principle, there exists a vertex~$x \in S^*$ and at least 3 distinct complete segments~$S_1^x, S_2^x, S_3^x$ 
    (ordered from top to bottom on the path~$\calP$) such that for each~$i \in [3]$, the segment~$S_i^x$ contains a node~$y_i$ such that the set~$V_G(\calT^*_ {y_i} - \calT^*_{z_i})$ contains a vertex~$w_z \in N_G(S^*)$

    Next, we construct a~$T$-cycle~$C^*$ such that~$V(C^*) \cap S^* = \{x\}$. 
    Towards this, let~$x_1$ and~$x_3$ be the first and the last cutvertices (while traversing from the root to the leaf) on the subpath~$\calP [y_1, y_3]$, respectively. Note that for each~$k \in [2]$, either~$x_k = y_k$ or~$x_k \in V(y_k)$.
    Let~$P_{13}$ be an $x_1$-$x_3$ path in the graph~$G[V_G(\calP[x_1, x_3])]$ such that~$P_{13}$ contains at least 2 terminals. 
    Note that such a path~$P_{13}$ exists due to Observation~\ref{Observation:Degree2Path} and the fact that there is a complete segment~$S_2^x$ 
    on~$\calP [x_1, x_3]$ which contains at least 5 distinct terminals and hence at least 5 distinct blocks containing distinct terminals due to Observation~\ref{Observation:Block_contains_atmost_1-terminal}. Let~$C^* := \{x, w_1\} \cdot R_1[w_1, x_1] \cdot P_{13} [x_1, x_3] \cdot R_3[x_3, w_3] \cdot \{w_3, x\}$, where for~$j \in \{1,3\}$, $R_j$ is a $w_j$-$x_j$ path inside~$G[V_G (\calT^*_{y_i} - \calT^*_{z_i})]$. Note that~$C^*$ is a~$T$-cycle with~$V(C^*) \cap S^* = \{x\}$, a contradiction to the fact that the set~$S^*$ is a 1-redundant MWNS of~$(G, T)$.
    This concludes the proof of Lemma~\ref{lemma:TerminalsOnDegreeeTwoPathIsBounded}.
\end{proof}

\subsection{Wrapping-up the proof of Theorem~\ref{theorem:terminalbounding}}
\label{section:WrappingUp}
Given Lemma~\ref{Lemma:NumberofInterestingLeavesIsBounded},  Lemma~\ref{lemma:TerminalsOnDegreeeTwoPathIsBounded}, and the bound on the cardinality of 1-redundant \mwnsshort~$S^*$ (due to Lemma~\ref{lemma:ConstructionOfRedundantSet}), the proof of Theorem~\ref{theorem:terminalbounding} follows 
from the observation that each terminal that does not belong to a block of~$\mathcal{L^*}$, belongs on a path from a block of~$\mathcal{L^*}$ to the root of a block-cut tree.

\TerminalBoundingState*
\begin{proof} 
    We first invoke Lemma~\ref{lemma:ConstructionOfRedundantSet} with the \mwnsshort instance~\gtk and the MWNS~$\hat{S}$ of~$(G, T)$ which in polynomial time outputs an equivalent instance~$(G_1, T, k_1)$ and a 1-redundant MWNS~$S^*$ of~$(G_1, T)$ of size~$\Oh (k \cdot |\hat{S}|)$ such that $G_1$ is an induced subgraph of~$G$ and $k' \leq k$. Moreover, there is a polynomial-time algorithm that, given a solution~$S_1$ for~$(G_1, T, k_1)$ outputs a solution~$S$ of~\gtk. Then we exhaustively apply Reduction Rule~\ref{RR:RemoveWeaklySeparatedTerminal}, Reduction Rule~\ref{RR:BoundingBlocksOfDegree2Path}, and Reduction Rule~\ref{RR:Bounding_CCTs} (with respect to the 1-redundant set~$S^*$) on the \mwnsshort instance~$(G_1, T, k_1)$ to obtain a reduced instance~$(G', T', k')$. Since in each reduction rule, we only remove terminals from the set~$T$, we have~$G' = G_1$ and~$k' = k_1$. This completes the construction of~$(G', T', k')$.

    First observe that given the \mwnsshort instance~$(G_1, T, k_1)$ and the 1-redundant set~$S^*$ of~$(G_1, T)$, it takes polynomial time to obtain the reduced instance~$(G', T', k')$ because each reduction rule takes polynomial time and the cardinality of the terminal set decreases by at least one after application of each reduction rule. Since Lemma~\ref{lemma:ConstructionOfRedundantSet} outputs the \mwnsshort instance~$(G_1, T, k_1)$ and 1-redundant set~$S^*$ in polynomial time, the whole construction of~$(G', T', k')$ takes polynomial time.

     Due to Lemma~\ref{lemma:ConstructionOfRedundantSet}, we know that the \mwnsshort instances~\gtk and~$(G_1, T, k_1)$ are equivalent and there is a polynomial-time algorithm that, given a solution~$S_1$ for~$(G_1, T, k_1)$ outputs a solution~$S$ of~\gtk.     
     Also, due to the safeness proof (Lemma~\ref{Lemma:SafenessOfWeaklySeparableTerminalRR}, Lemma~\ref{Lemma:SafenessOfDegreeTwoPathBoundingRR}, and Lemma~\ref{Lemma:SafenessOfCCTBoundingRR}) of each reduction rule, it holds that if~$(G_1, T_i, k_1)$ reduces to~$(G_1, T_j, k_1)$ then~$(G_1, T_i, k_1)$ and~$(G_1, T_j, k_1)$ are equivalent instances of~\mwnsshort, and there is a polynomial-time algorithm that, given a solution~$S_j$ for~$(G_1, T_j, k_1)$ outputs a solution~$S_i$ of~$(G_1, T_i, k_1)$.
     Hence the~\mwnsshort instances~$(G', T', k')$ and~\gtk are equivalent and there is a polynomial-time algorithm that, given a solution~$S'$ for~$(G', T', k')$ outputs a solution~$S$ of~\gtk. 
     
    Since~$G'= G_1$ and~$k' = k_1$ and~$G'$ is an induced subgraph of~$G$ and~$k_1 \leq k$ (due to Lemma~\ref{lemma:ConstructionOfRedundantSet}), it holds that~$G'$ is an induced subgraph of~$G$ and~$k' \leq k$. 
    Moreover, we have~$T' \subseteq T$ because the terminal set~$T$ does not change after invoking Lemma~\ref{lemma:ConstructionOfRedundantSet}, and we only remove terminals from the set~$T$ to obtain~$T'$ during application of reduction rules.

    It remains to prove that~$|T'| = \Oh(k^5 \cdot |\hat{S}|^4)$. Towards this, 
    first note that if~$(G', T', k')$ is a trivial instance of \mwnsshort (i.e., if~$\emptyset$ is a solution of~$(G', T', k')$) then it implies that all the terminals of~$T'$ are already nearly-separated in~$G'$. Since Reduction Rule~\ref{RR:RemoveWeaklySeparatedTerminal} is no longer applicable on~$(G', T', k')$, we have~$|T'| = 0$. Thus assume that~$(G', T', k')$ is a non-trivial instance of \mwnsshort. Let~$\calF'$ be a rooted block-cut forest of~$G'-S^*$.
    Let~$\calL'$ be the set defined by Definition~\ref{Definiton:InterestingLeaves} for~$\calF'$.
    Since no reduction rule is applicable on~$(G', T', k')$ w.r.t. the 1-redundant set $S^*$, the number of blocks from~$\calL'$ in the block-cut forest~$\calF'$ is bounded by~$\Oh(k' \cdot |S^*|^3)$ by Lemma~\ref{Lemma:NumberofInterestingLeavesIsBounded}.    
    
   For any path from a root to a leaf in~$\mathcal{F'}$, the number of blocks on the path containing a terminal is bounded by~$\Oh(|S^*|)$ due to Lemma~\ref{lemma:TerminalsOnDegreeeTwoPathIsBounded}.
   Since for any vertex~$t \in T$, either $t$ belongs to a block of~$\calL'$ or there exists a tree~$\calT' \in \calF'$, and a block~$L \in \calL'$ such that $t$ belongs to a block on the path from the root to~$L'$ in~$\calT'$. Thus the number of terminals occurring in a block of~$\mathcal{F'}$ is bounded by~$\Oh(k' \cdot |S^*|^4)$.
   Since~$S^* \subseteq V(G) \setminus T$ and each block of~$G'-S^*$ contains at most one terminal, we have~$|T'| = \Oh(k^5 \cdot |\hat{S}|^4)$. This concludes the proof of Theorem~\ref{theorem:terminalbounding}.
\end{proof}

\end{document}